\documentclass{marticle}
\author{
    Vedat Levi Alev\thanks{
        University of Haifa ---
    \href{mailto:vedatalev@math.haifa.ac.il}{vedatalev@math.haifa.ac.il}}
    \and Uriya A.~First\thanks{
        University of Haifa ---
    \href{mailto:ufirst@univ.haifa.ac.il}{ufirst@univ.haifa.ac.il}}
}
\title{
    Hypergraph Samplers:\\
    Typical and Worst Case Behavior
}
\date{\today}

\begin{document}
\begin{titlepage}
    \def\thepage{}
    \thispagestyle{empty}
    \maketitle

    \begin{abstract}
        In this paper we study the utility and limitations of using
        $k$-uniform hypergraphs $H = ([n], E)$
        (where $n \ge \poly(k)$)
        in the context of error reduction for randomized algorithms for
        decision problems with one- or two-sided error.
        Concretely, our error reduction idea is based on sampling a uniformly
        random hyperedge of $H$, and repeating the algorithm $k$
        times with the role
        of the random seeds taken by the vertices in the hyperedge. We note that
        this is a very general paradigm, which captures every
        pseudorandom method for
        choosing $k$ seeds without repetition. We show  two results which
        together imply a \emph{gap} between
        the \emph{typical} behavior and
        the \emph{worst-case} behavior when using $H$ for error-reduction.

        Observing that in the context of one-sided error reduction,
        using $k$ IID
        seeds would drop the error probability from $p$ to $p^k$ at the cost of
        using $k \lg n$ random bits, we mention our first result: In the
        context of one-sided error reduction,
        if using a random hyperedge
        from $H$ decreases the
        error probability from $p$ to $p^k + \ee$, then $H$ cannot have too
        few edges, i.e.,~$|E| = \Omega(n k^{-1} \ee^{-1})$.
        Consequently, the number of random bits we need for dropping the error
        from $p$ to $p^k + \ee$
        cannot be reduced below $\lg n+\lg(\ee^{-1})-\lg k+O(1)$.
        Here, the  suppressed constants depend only on $p$.
        Equivalently, a pseudorandom method for choosing $k$ seeds
        using $\lg n+\lg r $ random bits cannot reduce the error probability
        below $p^k+\Omega(r^{-1}k^{-1})$ in general.
        This also holds for hypergraphs with \emph{average} uniformity $k$.
        Our results also imply new lowerbounds for dispersers and
        vertex-expanders in certain parameter regimes.

        Second, provided that the vertex degrees in $H$ are not too
        outlandishly distributed,
        we show that in \enquote{99.99\% of cases} (a $(1-o(1))$-fraction),
        choosing $k$ pseudorandom seeds
        using $H$ will reduce the error probability to at most $o(1)$
        above the error probability we would get by using $k$ IID seeds,
        for both algorithms with one- or two-sided error. Thus, in the
        context of error reduction, despite our lowerbound,
        for  \emph{all  but a  vanishing fraction} of  randomized
        algorithms (and inputs) for decision problems,
        the advantage of using  true  IID samples over samples obtained
        from a uniformly random  edge of \emph{any}  reasonable hypergraph
        is essentially \emph{negligible}. In particular, \emph{any}
        $D$-regular $k$-uniform hypergraph $H$  on  $n\geq \poly(D,k)$
        vertices gives rise to
        a pseudorandom seed-choosing method
        using
        $\lg |E| = \lg n + \lg D - \lg k$ random bits and having the property
        that it
        drops one-sided error-probability $p$
        to beneath $p^k + o(1)$ for \emph{all but a vanishing
        fraction} of randomized algorithms (and inputs). A similar
        statement holds true
        for randomized algorithms with two-sided error.
    \end{abstract}
\end{titlepage}
\tableofcontents
\section{Introduction}
\subsection{Motivation and General Overview}\label{ss:gen_overview}
\label{subsec:overview}

Repetition is a natural strategy of decreasing the error
probability of randomized algorithms for solving decision problems, see for
example \cite[Ch.~7.4]{AroraB09}. Suppose we are given an algorithm which
solves a decision problem with one sided error $p$,
i.e.,~one which never reports false negatives but reports
false positives with some probability $p \in (0, 1)$. It is easy to see
that running the algorithm $k$ times in parallel with IID random
seeds decreases the error
probability from $p$ to $p^k$ at a cost of increasing the random bits we
require from (say) $\lg n$ to $k \lg n$.\footnote{
    $\lg n$ means $\log_2 n$. We caution that $n$ is completely
    independent of the randomized alogrithm's time complexity, and
    quantifies the \emph{size}
    of the random seed, i.e.,~$\lg n$ is the number of bits required to encode
    the seed.
} Suppose we have an algorithm with
two sided error $p < 1/2$, i.e.,~one
which can err with probability $p < 1/2$ (by reporting both false
positives and negatives). The Chernoff bound implies that repeating the
algorithm $k$ times in parallel with IID random seeds and reporting
the majority answer
also leads to an exponential decrease in
the error probability, while again increasing the number of random bits we
require from $\lg n$ to $k \lg n$.

Whereas these methods  work well when one does not mind this blow-up in the
number of random bits required, when one also tries to ensure that not too many
random bits are used, other ideas are needed. A typical idea which comes to the
rescue is replacing the IID seeds with a \emph{pseudorandom} sequence of
seeds. For example, the expander hitting set lemma \cite{AlonFW95, Kahale95}
and the expander Chernoff bound \cite{Kahale97, Gillman98}, imply
that by picking the random seeds
to be situated on a uniformly random path of length $k$ in a constant
degree expander graph
on $n$ vertices,
we can decrease the random bits we require from $k \lg n$ to $\lg n +
O(k)$ in both cases, while
achieving similar guarantees to IID samples.

There is a wide body of work studying the power and limitation of
these techniques.
Random walks in expander graphs are expected to behave like sequences
of independent samples.
This point of view  was made precise in many works \cite{AlonFW95,
    Kahale97, Gillman98, Healy08,
CohenPT21, CohenMPPT22, GolowichV22, DiksteinH24, Ta-ShmaZ24}. These
works can be considered as far reaching generalizations of the
celebrated \emph{Expander Mixing Lemma} which is the simplest result
of this kind: In a  regular expander graph $G = (V, E)$, the fraction
of edges between two sets  $A, B\subseteq V$ is close to the fraction
of edges between $A$ and $B$ in the complete graph $(V, {{V}\choose 2})$.
%


In this work, we take
a different route and study general pseudorandom choices of seeds
(without repetition).
Formally, instead of using the vertices situated
on a random path of length $k$ in an expander graph, we take an
arbitrary $k$-uniform
hypergraph $H = (V, E)$ on $n$ vertices, i.e.,~where every hyperedge contains
$k$ vertices, and ask how well taking the vertices situated on a random
hyperedge as seeds performs in error reduction
settings. Concretely, given a set $A \subseteq V$ of density $p$
(i.e.,~$|A| = \lfloor p \cdot n \rfloor$),
we will be interested in the following quantities:
\begin{itemize}
    \item the \emph{number of hits}  ${\rv N}_H(A)=|\rv e\cap A|$,
        i.e.,~the number of times
        a uniformly random edge $\rv e$ in $H$ hits $A$.
    \item the \emph{confinement probability} $\Pp[\rv e\subseteq A]$,
        i.e.,~the probability that a uniformly random edge $\rv e$ in $H$
        is confined to $A$.
\end{itemize}
We note that \emph{the number of hits} is related to the \emph{two-sided error
reduction} problem described above. Indeed, let $A \subseteq V$ be
the set of seeds which
make our algorithm err. If we can show that under suitable
assumptions, $\rv N_H(A)$ enjoys a Chernoff type tail-bound, i.e.,~the
probability of it being $\ee$-far from the average  decays exponentially in $k$,
we can use it for error
reduction purposes.
Similarly, the \emph{confinement probability} is
the error probability (after $k$ repetitions) in the \emph{one-sided
error} setting if $A$ is the set of random seeds
which make an algorithm with one-sided error err.

Observe that in both the two-sided and the one-sided error setting,
the
natural parameters of the hypergraph $H = (V, E)$ control the
\emph{efficiency} of the error reduction. Specifically, the
uniformity of $H$, i.e.,~the
number of vertices in a hyperedge of $H$, is the
number of times the algorithm is repeated, and the
hyperedges-to-vertices ratio $|E|/n$ controls the randomness that one
needs to expend on top of a single uniform sample from $V$.
From this point of view, it makes sense to study the
\emph{confinement probability} and the \emph{number of hits} of a
\emph{not-necessarily-uniform} hypergraphs also, which we will do
for the former.

We measure the
the error-reduction performance of sampling $k$ seeds according to a
hyperedge in a $k$-uniform hypergraph $H=(V,E)$ by comparing it to
running $k$ parallel independent instances of an algorithm (without
repeating the same seed twice), which corresponds
to
choosing uniformly a random edge
in  $K=(V, \binom{V}{k})$, the complete $k$-uniform hypergraph on the
ground set $V$.
That is, we would like $\rv N_H(A)$ and $\rv N_K(A)$,
resp.\ $\Pp_{\rv e\sim E(H)}[\rv e\subseteq A]$
and  $\Pp_{\rv e\sim E(K)}[\rv e\subseteq A]$, to be similar for
every density-$p$ subset $A\subseteq V$. The  distribution of $\rv
N_K(A)$  is always the hypergeometric distribution $\Hyp(k,\lfloor pn
\rfloor, n)$,\footnote{i.e.,~the probability of producing $r$ hits is
    the same as pulling $r$ red balls from an urn containing $\lfloor
    pn \rfloor$
    red balls and $\lceil (1-p)n \rceil$ blue balls after $k$ pulls, not
    returning the balls to the urn after every pull. In other words,
    $\Pp[\rv N_H(A) = r] = \frac{\binom{k}{r} \binom{n-k}{\lfloor
pn\rfloor - r}}{\binom{n}{\lfloor pn \rfloor}}$.}
and   when   $k \ll n$, it tends to the binomial distribution $\Bin(k, p)$.
Similarly, $\Pp_{\rv e\sim E(K)}[\rv e\subseteq A]$ tends to $p^k$ for $k\ll n$.
We would therefore be interested in hypergraphs $H$ where $\rv
N_H(A)$ is as close as possible
to $\Bin(k, p)$, or where $\Pp_{\rv e \sim E(H)}[\rv e \subseteq A]$
is not much   larger
than $p^k$. To measure the distance between $\Pp_{\rv e\sim E(K)}[\rv
e\subseteq A]$ and
$\rv N_H(A)$, we use the total variation distance $\dist_{\TV}$;
recall that it is defined as follows
\[ \dist_{\TV}(\rv N_H(A), \Bin(k, p) ) := \frac12 \sum_{j = 0}^k
\Abs*{ \Pp[ \rv N_H(A) = j] - \binom{k}{j} p^j (1-p)^{k-j}}.\]

Our main results show that in \emph{any}  $k$-uniform hypergraph
$H=(V,E)$ on $n \ge \poly(k)$ vertices in which the vertex-degrees
are not significantly non-uniform, there
is a gap between the \emph{typical} behavior of $\rv N_H(A)$,
resp.\ $\Pp_{\rv e\sim E}[\rv e\subseteq A]$,
and its \emph{worst-case} behavior. More precisely, for every fixed
$0<p<1$, a $1-o(1)$ (\enquote{99.99\%}) fraction of the density-$p$ subsets
$A\subseteq E$
satisfies $\dist_{\TV}(\rv N_H(A), \Bin(k, p) )  \le o(1)$ and also a
Chernoff bound $\Pp\sqbr*{\Abs*{\rv N_H(A)-pk}\geq k\eta}
\leq 2\exp\parens*{-2k\eta^2}+o(1)$ (here, $0<\eta<1$). Consequently,
$\Pp_{\rv e\sim E}[\rv e\subseteq A]\in [p^k-o(1),p^k+o(1)]$
for all but an $o(1)$-fraction of $A$-s of density $p$. Otherwise said,
for \emph{almost all} density-$p$ subsets $A\subseteq V$, the behavior
of $\rv N_H(A)$ and $\Pp_{\rv e\sim E}[\rv e\subseteq A]$ is $o(1)$-close
to the optimum, i.e., to the behavior in a complete $k$-uniform hypergraph.
However, by contrast, assuming $r:=|E|/n$ is not too large relative
to $n$, there are density-$p$ subsets $A\subseteq V$
for which $\Pp_{\rv e\sim E}[\rv e\subseteq A]\geq p^k+\Omega(r^{-1}k^{-1})$,
and hence $\dist_{\TV}(\rv N_H(A), \Bin(k, p) )  \ge \Omega(r^{-1}k^{-1})$.
Equivalently, if    $\Pp_{\rv e\sim E}[\rv e\subseteq A]\leq p^k+\ee$
for all density-$p$ subsets $A\subseteq V$, then $H$ has at least
$\Omega(k^{-1}\ee^{-1} n)$ hyperedges, or rather,
$|E|/n\geq \Omega(k^{-1}\ee^{-1})$.
Here, and until the end of this section, the suppressed constants
in the the   asymptotic notations depend only on $p$,
which is assumed to be fixed.

Interpreted in the context of error-reduction, our results mean the
following: First,
\emph{any}  pseudorandom method for generating (exactly) $k$ seeds for a random
algorithm that allows
each individual seed to be picked with roughly the same probability
performs nearly
the same as $k$ independent repetitions \enquote{99.99\% of the time}, and is
therefore very good
in practice. On the other hand, in the context of $1$-sided
error reduction, any pseudorandom method for generating $k$ seeds
with random complexity $\lg n+\lg r$
cannot reduce
the error below $p^k+\Omega(r^{-1}k^{-1})$ in general.
Otherwise stated, it impossible to reduce the error below
$(p^k-o(1))$ --- as one would expect ---,
and even getting close to $p^k$, i.e.,~guaranteeing an error of $p^k +\ee$,
is expensive in the sense that it requires at least
$\lg n+\lg (\ee^{-1})-\lg k-O(1)$ random bits.
This lower bound, which appears to be the first one stated
explicitly, means in particular that (1) for a fixed $k$, approaching
an error probability of $p^k$
as $n$ grows
using only $\lg n+O(1)$ random bits is impossible
and (2)  making the error probability exponentially small in $k$
requires $\lg n+\Omega(k)$ random bits.
As the expander hitting set lemma shows that  $\lg n+O(k)$ random bits
are enough for an exponential reduction of the error, our lower bound is
optimal in this context. That said, from other perspectives,
e.g.~those of optimal sampling or hypergraph expansion, there is
still a gap between our lower bound
and what is provided by the expander hitting set lemma
and similar results.


In addition to the above, our results about the typical behavior of
hypergraph sampling also have implications to polling.
To illustrate   their utility,
suppose one is interested in estimating how common is a certain
feature among a large population of size $n$ by polling just $k$ of them.
The best way to do that would be to sample $k$ distinct members
chosen uniformly at
random and decide accordingly, but  this is not always easy, and
sometimes impossible.
Our result says that  \emph{if there is good reason to think that the
    set of members carrying
the feature behaves like a random subset} of density $p$ with $0\ll p \ll 1$,
then, with probability $1-o(1)$, any  reasonable  pseudorandom
sampling strategy (i.e.,~one which satisfies our modest assumptions)
would work as well as a truly random poll.

Finally, our results about the worst-case sampling behavior also
imply new bounds on dispersers
and vertex-expanders in certain parameter regimes.



The remainder of this introduction
gives a formal account of our results, their implications,
and comparison to other works.

\subsection{Samplers and Confiners}

Let $\ee > 0$ be some small value. In order to streamline our discussion,
we call a $k$-uniform hypergraph $H  = (V, E)$ on $n$ vertices an
\ref{eq:eps_sampler} for $p$-dense subsets, or an
\emph{$(\ee,p)$-sampler} for short, if the random variable $\rv
N_H(A)$ is \emph{$\ee$-close} in distribution to a binomial random
variable for every $A\subseteq V$ with $\lfloor pn\rfloor$ elements.
Formally, we require that  the distance in total variation between
$\rv N_H(A)$ and $\Bin(k,p)$
is at most $\ee$ for every $A\subseteq V$ of density $p$, that is,
\begin{equation}\label{eq:eps_sampler}\tag{$\ee$-sampler}
    \dist_{\TV}(\rv N_H(A), \Bin(k, p) )  \le \ee\qquad\forall~A \in
    \binom{V}{\lfloor p \cdot n \rfloor}.
\end{equation}
By the usual Chernoff bound for binomial distribution, every
\ref{eq:eps_sampler} for $p$-dense sets satisfies an approximate
Chernoff bound
\begin{equation}\label{eq:chernoff_show}
    \Pp\sqbr*{\Abs*{\rv N_H(A)-pk}\geq k\eta}
    \leq 2\exp\parens*{-2 k\eta^2}+\ee
    \qquad\forall 0<\eta\leq 1,~A \in \binom{V}{\lfloor p \cdot n \rfloor}.
\end{equation}
We call such hypergraphs \hyperref[eq:chernoff_show]{$\ee$-Chernoff
samplers} for $p$-dense
sets, or just $(\ee,p)$-Chernoff samplers.
(We caution that what is usually called
    a \enquote{sampler} in the literature refers, up to reparametrization,
    to the weaker notion
of a Chernoff sampler, cf.~\cite[Dfn.~3.29]{Vadhan12}.)


When $\Pp\sqbr*{\rv e \subseteq A}$
does not exceed 
$\epsilon > 0$ for every $A$ of density $p$, we  will call the
hypergraph $H = (V, E)$ an \ref{eq:eps_confiner} for $p$-dense
subsets, or an $(\epsilon,p)$-confiner for short. Formally,
\begin{equation}
    \label{eq:eps_confiner}\tag{{$\epsilon$}-confiner}
    \Pp\sqbr*{\rv e \subseteq A}  \le \epsilon\qquad\forall~A \in
    \binom{V}{\lfloor p \cdot n\rfloor}.
\end{equation}
Note that this definition also makes sense for  \emph{non-uniform}
hypergraphs of average uniformity $k$. As a consequence of one of our
results, any $k$-uniform
\hyperref[eq:eps_confiner]{$(\epsilon,p)$-confiner} with $n$ vertices
should necessarily satisfy $\epsilon \geq {n\choose k}^{-1}  { n  -
k\choose \lfloor pn \rfloor - k }   \ge p^k -\frac{k-1}{n-1}$.
The number $ {n\choose k}^{-1}  {n-k\choose \lfloor pn
\rfloor -  k}$ on the left hand side is the confinement
probability of any $p$-dense set of
vertices in the  complete $k$-uniform hyper-graph on $V$.

It is an easy observation that any
\hyperref[eq:eps_sampler]{$(\ee,p)$-sampler} is a
\hyperref[eq:eps_confiner]{$(p^k+\ee,p)$-confiner}.

We will sometimes talk about hypergraphs which
are \hyperref[eq:eps_sampler]{$\ee$-samplers}
(resp.~\hyperref[eq:eps_confiner]{$\epsilon$-confiners},
\hyperref[eq:chernoff_show]{$\ee$-Chernoff samplers}), for more
general families of subsets $A\subseteq V$,
e.g., a subcollection of the  density-$p$ subsets.

\begin{remark}
    An \hyperref[eq:eps_confiner]{$(\epsilon,p)$-confiner} on $n$
    vertices and $m$ hyperedges is
    essentially the
    same thing
    as a \emph{$(\lg((1-p)n),\epsilon)$-disperser},
    cf.\ \cite[Prop.~6.20]{Vadhan12}, \cite{Shaltiel04},
    and also essentially the same as a
    $({=(1-p)n},(1-\epsilon)m)$-vertex expander in the
    sense of \cite[Prop.~4.7]{Vadhan12}.
    Since we are interested in  confinement probability, the notion of
    an  \hyperref[eq:eps_confiner]{$(\epsilon,p)$-confiner} is more
    natural  for our purposes.
    We will explain the implications of our results to dispersers
    and vertex expanders after they are presented.
\end{remark}

In what follows, a \emph{family of hypergraphs}
is an infinite sequence of hypergraphs $\{H_i = (V_i,E_i)\}_{i\geq 1}$
with $|V_i|\to \infty$.
A hypergraph $H=(V,E)$ is \emph{$r$-sparse} if $|E|\leq r|V|$;
the \emph{sparsity} of $H$ is its (hyper)edge-to-vertex ratio $|E|/|V|$.
A   family of hypergraphs is said to be sparse
if there is $r>0$ such that all hypergraphs in the family are $r$-sparse.

Unless indicated otherwise, the asymptotic notation $O_{a,b,\dots}(\bullet)$
means that the suppressed constant in the $O$-notation depends (only)
on the parameters $a,b,\dots$.
The same convention applies when we shall write
$\Omega_{a,b,\dots}(\bullet)$, $\Theta_{a,b,\dots}(\bullet)$
or $\poly_{a,b,\dots}(\bullet)$.


\subsection{Our Contributions}\label{ss:our_contributions}

\subsubsection{Sampling in the Typical Case}

Our first result shows that \emph{any} uniform hypergraph $H = (V,
E)$ whose vertex degrees   are not too outlandishly distributed is
a very good sampler (resp.\ Chernoff sampler, confiner)
for \emph{almost} all $p$-dense subsets of $V$.

\begin{theorem}[Simplification of \cref{cor:fraction} and
    \cref{cor:thm0}]\label{cor:fraction-show}
    Let $0<p<1$
    and let $H = (V, E)$ be a $k$-uniform hypergraph on
    $n\geq \frac{4k^2}{p(1-p)}$ vertices.
    Let $\beta\in(0,1]$ and suppose that no vertex of $H$
    is contained in more than an $\parens*{n^{-\beta}}$-fraction of the edges
    (e.g., this holds for $\beta=1-\frac{\lg k}{\lg n}\geq \frac{1}{2}$
    when $H$ is regular).
    Then the following hold.
    \begin{enumerate}[(i)]
        \item For every $\alpha\in (0,1)$, we  have
            \begin{equation}\label{eq:dist-showcase}
                \dist_{\TV}( \rv N_H(A), \Bin(k, p)) \le
                n^{-\frac{\beta(1-\alpha)}{2}} + \frac{k-1}{n-1}
            \end{equation}
            for all but an $O(\frac{k^4}{p(1-p)})\cdot   n^{-\alpha\beta}$
            fraction of the subsets $A \subseteq V$ with $|A| = p \cdot n$.
        \item For all $\alpha\in (0,1]$, $\alpha'\in [0,1]$ and $\eta
            \geq 0$, we have
            \begin{equation}
                \label{eq:chernoff-showcase}
                \Pp\sqbr*{ \Abs*{ \rv N_H(A) - k p } > k \eta} \le
                2\exp(- k \eta^2) +
                \exp(-\textstyle{\frac{1-\alpha'}{2}}k\eta^2)
                n^{-\frac{\beta(1-\alpha)}{2}}
            \end{equation}
            for all but an  $ O(\frac{k^4}{p(1-p)})\exp(-\alpha'k\eta^2) \cdot
            n^{-\alpha\beta}$ fraction of the subsets $A \subseteq V$ with
            $|A| = p \cdot n$.
        \item For all $\alpha\in (0,1)$, we have
            \begin{equation}
                \label{eq:confinement-showcase}
                \Abs*{\Pp_{\rv e\sim E}\sqbr*{\rv e\subseteq A} - p^k}\leq
                n^{-\frac{1-\alpha}{2}}
            \end{equation}
            for all but an $O(\frac{k^2}{ 1-p}) p^k \cdot
            n^{-\alpha\beta}$-fraction of the subsets  $A \subseteq V$ with
            $|A| = p \cdot n$.
    \end{enumerate}
\end{theorem}

In other words,
provided the density $p$ is bounded away from $0$ and $1$ and $n\geq
\Omega_p(k^4)$, \emph{any} reasonable $k$-uniform hypergraph $H$ on $n$
vertices is an $o(1)$-sampler, an $o(1)$-Chernoff sampler and a
$(p^k+o(1))$-confiner for \emph{almost any} density-$p$ subset $A\subseteq V$.
This says that apart from a negligible amount of exceptions, the
sampling qualities of
$H$ are almost identical to those of the
complete graph $(V,{V\choose k})$. In the context of error-reduction,
this further means that for all but a vanishing fraction of
randomized algorithms (and inputs) for decision problems, the advantage of
repeating the algorithm using true IID seeds over samples obtained
from a uniformly random  edge of an approximately-degree-regular
uniform hypergraph is essentially \emph{negligible}.

\begin{remark}\label{rem:average_stuff}
    (i) The additive factor of $(k-1)/(n-1)$ occuring in
    \cref{eq:dist-showcase} accounts for the difference between the
    hypergeometric distribution $\Hyp(k,pn,n)$ and the binomal
    distribution $\Bin(k,p)$. Indeed, we show that this factor can be
    removed once  $\Bin(k,p)$ is replaced with $\Hyp(k,pn,n)$,
    q.v.~\cref{cor:fraction}.


    (ii) If $D$ is the maximum vertex degree of $H=(V,E)$,
    then the condition that no vertex is contained in no more than
    $n^{-\beta}$-faction
    of the edges appearing in \cref{cor:fraction-show} is equivalent to
    having $\frac{D}{|E|}\leq n^{-\beta}$.
    Presented this way, it is a very mild assumption.

    (iii) Part (iii) of \cref{cor:fraction-show}  implies that any
    $k$-uniform
    \hyperref[eq:eps_confiner]{$(\epsilon,p)$-confiner} with
    vertex-degrees not-too-large satisfies
    $\epsilon\geq p^k-o(1)$. By our next theorem,
    this actually holds without any assumption on the vertex degrees
    and under the milder assumption that $H$ has \emph{average} uniformity $k$.
    In fact, $k$ does not have to be an integer.
\end{remark}

\begin{theorem}[Simplified Version of
    \cref{cor:conf_lowerbound_main}]\label{TH:worst_error_confine}
    Fix some real numbers $0<p<1$ and $k>0$.
    Let $H=(V,E)$ be a hypergraph on $n\geq \frac{1}{p(1-p)}$
    vertices with average uniformity $k$.
    If $H$ is an \hyperref[eq:eps_confiner]{$(\epsilon,p)$-confiner},
    then $\epsilon\geq  p^k- 2k/n$.

    Moreover, when $H$ is  $k$-uniform, we have the improved lower bound
    $\epsilon \geq {n\choose k}^{-1}  {n-k \choose \lfloor pn\rfloor
    -k} \ge p^k -
    \frac{k-1}{n-1}$.
\end{theorem}

\subsubsection{Sampling in the Worst Case}
\label{subsec:worst-case}

\cref{cor:fraction-show} implies that
designing \hyperref[eq:eps_sampler]{$\ee$-samplers},
\hyperref[eq:chernoff_show]{$\ee$-Chernoff samplers} and
\hyperref[eq:eps_confiner]{$\epsilon$-confiners}
which perform almost optimally (like a complete hypergraph)
for all but a \emph{vanishing} fraction
of density-$p$ sets $A \subseteq V$ is very simple, as one only needs to control
the uniformity and the degree of the hypergraph $H = (V, E)$, and without
further considerations the prerequisites of \cref{cor:fraction-show}
are readily satisfied.
With that in mind,
it is natural to look for examples of $k$-uniform samplers
(resp.~Chernoff samplers, confiners)
$H=(V,E)$
for $p$-dense sets
whose behavior for \emph{all} $p$-dense $A\subseteq V$ is close to
the optimal, typical behavior.
For applications such as error reduction, these
examples should be as sparse as possible.
We shall survey the best known such constructions in
\cref{subsec:relation} below, but mention
right away that there are no known sparse samplers (resp.\ Chernoff
samplers, confiners)
in which the worst-case behavior is close  to the typical behavior
from \cref{cor:fraction-show}.

Our next result shows that finding sparse families of $k$-uniform
samplers and confiners (but not
Chernoff samplers)
whose performance on all density-$p$ sets approaches the optimal
performance on \enquote{99.99\%} of those sets is in fact impossible.
In the case of confiners, this moreover holds for hypergraphs with
\emph{average} uniformity $k$, which need not be integral.

\begin{theorem}[Simplified Version of
    \cref{TH:lower-bound}]\label{TH:hitting_bound_show}
    Fix $p\in (0,1)$. Let $H = (V, E)$ be a hypergraph of
    average uniformity $k\geq 2$ on $n $ vertices and $ r \cdot n$
    edges ($r>0$)
    such that $n\geq  \poly(r,k, p^{-1}, (1-p)^{-1} )$.
    If $H$ is an \hyperref[eq:eps_confiner]{$(\epsilon,p)$-confiner},  then
    \[ \epsilon \ge \min\left\{p^k + \Omega\left(\frac{p(1-p)^2}{rk}
        \right),1\right\}
    .\]
    Consequently, if $H$ is a $k$-uniform $(\ee,p)$-sampler and
    $\ee<1-p^k$, then
    $\ee\geq \Omega\left(\frac{p(1-p)^2}{rk}
    \right)$.
\end{theorem}

Otherwise stated, \hyperref[eq:eps_confiner]{$(p^k+\ee,p)$-confiners}
(and {\it a fortiori} $(\ee,p)$-samplers) cannot be too sparse:

\begin{corollary}[Simplified Version of
    \cref{TH:lower-bound-epsilon-view}]\label{CR:eps_view_show}
    Let $p,k,n,H$ be as in \cref{TH:hitting_bound_show}, and let $\ee\in (0,
    1-p^k)$.
    If $H$ is a  \hyperref[eq:eps_confiner]{$(p^k+\ee,p)$-confiner} (resp.~an
    $(\ee,p)$-sampler when $H$ is $k$-uniform),
    then
    \[|E|  \ge
        n \cdot  \min\left\{
            \Omega\parens*{\frac{p(1-p)^2}{k \ee}},
            \frac{n^c}{\poly(k,\frac{1}{p},\frac{1}{1-p})}
    \right\},\]
    where $c>0$ is an absolue constant (one can take $c=\frac{1}{13}$).
    In particular, if $n\geq \poly(k,\ee^{-1},p,\frac{1}{1-p})$,
    then
    \[|E|\geq n\cdot \Omega\parens*{\frac{p(1-p)^2}{k \ee}}.\]
\end{corollary}

%

\begin{remark}
    %
    \cref{TH:hitting_bound_show} holds for any $k>1$ provided one
    replaces $\Omega(\frac{p(1-p)^2}{rk}   )$ with a more complicated
    expression approaching $0$ as $k\to 1$; q.v.~\cref{TH:lower-bound}.
    The same is true for \cref{CR:eps_view_show};
    q.v.\ \cref{TH:lower-bound-epsilon-view}.
    However, both results are \emph{false} for $k=1$; see
    \cref{RM:k_less_than_1}.
\end{remark}

As explained in \cref{subsec:overview}, in the context of one-sided error
reduction, \cref{TH:worst_error_confine} and \cref{CR:eps_view_show}
imply the following novel
conclusions for one-sided error reduction:

\begin{corollary}\label{CR:random_algs}
    Let $\cal A$ be a randomized algorithm  outputing a subset
    of $\{1,\dots,n\}$   having $k$ elements on average.
    Given another  randomized   algorithm $\cal A'$ for a decision
    problem which takes a seed of $\lg n$ bits and has one-sided
    error probability at most $p\in (0,1)$,
    let $\cal A' := \cal A'[\cal A]$ denote the decision algorithm
    for that problem
    which  uses $\cal A$ to generate    $k$ seeds (on average),
    runs $\cal A'$ with those seeds, and decides accordingly.
    Then:
    \begin{enumerate}[(i)]
        \item
            Provided $n\geq \frac{1}{p(1-p)}$, there is   $\cal A'$
            as above such   that the error probability of $A'[A]$ is
            at least $p^k-O(\frac{k}{n})$.
        \item If for every $\cal A'$ as above, the error probability
            of $\cal A'[\cal A]$
            does not exceed $p^k+\ee$ ($0<\ee<1-p^k$), then
            the random complexity of $\cal A$ is at least
            \[
                \lg n+ \min\left\{\lg(\ee^{-1})-\lg k-O_p(1),~ c\lg
                n-O_p(\lg k)\right\}
            \]
            for some absolue constant $c>0$. In fact, one can take
            $c=\frac{1}{13}$.
    \end{enumerate}
\end{corollary}

In simpler words, no pseudorandom method for choosing $k$ random seeds
for $k$ repetitions of a randomized algorithm with one-sided error
probability $p$
can reduce the error probability below $p^k-O(\frac{k}{n})$ in general, and
if such a pseudorandom method is guaranteed to reduce the error probability
to $p^k+\ee$, then its
random complexity is at least
the minimum
of $\lg n+\lg(\ee^{-1})-\lg k-O_p(1)$
and $(1+c)\lg n-O_p(\lg k)$.
By contrast, \cref{cor:fraction-show}(iii)
which says that any pseudorandom method  for choosing $k$ seeds   that
picks each individual seed with roughly
the same probability will reduce the error probability from $p$ to $p^k+o(1)$
for all but a vanishing
fraction of randomized algorithms.

\begin{proof}
    There is $r>0$ such that the random complexity of $\cal A$
    is $\lg n+\lg r$. Think of all the possible $rn$ outputs of $\cal A$
    as the hyperedges of a hypergraph $H=(V,E)$ with vertex set
    $V=\{1,\dots,n\}$.

    By \cref{TH:worst_error_confine}, provided $n\geq\frac{1}{p(1-p)}$,
    there is a subset $A\subseteq V$
    of density at most $p$ such that $\Pp_{\rv e\in E}[\rv e\subseteq
    A]\geq p^k-\frac{2k}{n}$. Now, to prove (i), consider
    a trivial decision problem whose answer is \enquote{Yes} on every
    instance and take $\cal A'$ to be the randomized algorithm that returns
    \enquote{No} if and only if its random seed is in $A$.

    Next, the assumption in (ii) implies
    that $H$ must be a  \hyperref[eq:eps_confiner]{$(p^k+\ee,p)$-confiner},
    so by \cref{CR:eps_view_show}, we have
    $rn=|E|\geq \min\{\Omega_p(k^{-1}\ee^{-1}),\frac{n^c}{\poly_p(k)}\}$.
    Since the random complexity of $\cal A$ is $\lg n+\lg r$,
    this proves (ii).
\end{proof}

%

\begin{example}
    Assuming $n$ is large enough,
    \cref{CR:random_algs}(ii)
    says that in order to have the one-sided error
    probability of a randomized algorithm be reduced from $p$ to $(1+a) p^k$ via
    $k$ repetitions with pseudorandom seeds,
    one needs to expend at least $k \lg (p^{-1}) + \lg (a^{-1}) - \lg k -O_p(1)$
    random additional bits  on
    top of the random bits needed for the randomized algorithm.
    %
    More generally, making the error probability exponentially small in $k$
    requires $\Omega (k)-O_p(1)$ additional random bits.
\end{example}
%
%
Before we proceed further, a few remarks are in order.

First, the known constructions of sparse
\hyperref[eq:eps_confiner]{$(\epsilon,p)$-confiners}
with average uniformity $k$
still do not match our lower-bound of $
\Omega_p((\epsilon-p^k)^{-1}k^{-1})\cdot n$ on the number
of hyperedges. We discuss this in detail below.

Second, in many applications of one-sided error reduction, one wants
the original error probability
$p$ to be reduced to   $\epsilon:=p^k+\ee$ that
is exponentially small in $k$. However, in this case, once $k\geq
\Theta(\log n)$, the lower bound
on the random complexity   in \cref{CR:random_algs}(ii)
would be $(1+c)\lg n-O_p(\lg k)$, and further decreasing of $\ee$
would not improve that.
This raises the question of whether the term $c\lg n-O_p(\lg k)$ in
\cref{CR:random_algs}(ii)
could be increased in general. Since the random complexity of generating
$k$ independent
seeds in $\{1,\dots,n\}$ is $k\lg n$, the best we could hope for is a
lower bound
on the random complexity of   the form $\min\{\lg n+\lg(\ee^{-1})-\lg k-O_p(1),
\Theta_p(k\lg n)\}$. This is equivalent to relaxing the assumption
$n\geq \poly_p(r,k)$ in \cref{TH:hitting_bound_show} to $n\geq
\poly_p(r^\frac{1}{k},k)$.
We pose this as a problem.

\begin{problem}
    Does \cref{TH:hitting_bound_show} continue to hold when
    $n\geq \poly_p(r^\frac{1}{k},k)$?
\end{problem}



We finally remark that we do not know if an analogue of
\cref{TH:hitting_bound_show} holds
for Chernoff samplers (cf.\ \cref{eq:chernoff_show}). We pose this as a problem.


\begin{problem}
    Fix some $p\in (0,1)$, $r>0$, $k>1$ and small $\eta>0$.
    Let $H = (V, E)$ be a $k$-uniform hypergraph with $n$ vertices
    and at most $rn$ edges. Provided $k\geq C(\eta)$ for suitable
    $C(\eta)\geq 1$, is there a subset of vertices $A \subseteq V$ of
    density $p$ where the tails of $~\rv N_H(A)$
    are not light, i.e.,~$\Pp\sqbr*{ \Abs*{\rv N_H(A) - kp} > k \eta}$ is
    bounded from below by  $2\exp(-\eta^2k)+c(r,k,p,\eta)$ for some
    $c(r,k,p,\eta)>0$ (not depending on $n$)?
\end{problem}

A positive solution would allow us to draw lower bounds on the random
complexity of
two-sided error reduction analogous to the ones we showed for
one-sided error reduction in \cref{CR:random_algs}.


\paragraph{Comparison to Bounds on Dispersers.}

As we noted earlier, confiners and dispersers are the same objects up to
reparametrization --- although the parameter regimes in
which they are studied are typically different.
In \cite[Thm.~1.5(a), Thm.~1.10(a)]{RadhakrishnanTS00}, the authors
give  bounds on the performance
of dispersers and moreover show, using a random construction,
that they are essentially optimal in the range of parameters relevant for
dispersers.
We restate these results in the language of confiners in
\cref{apx:dispersers}; they say the following:
\begin{enumerate}[(a)]
    \item Let $H$ be an
        \hyperref[eq:eps_confiner]{$(\epsilon,p)$-confiner}
        ($0<\epsilon\leq \frac{1}{2}$) of average uniformity $k$
        with $n$ vertices and $nr$ edges.
        Then $r\geq \Omega(\epsilon^{-1}k^{-1}\lg(\frac{1}{1-p}))$
        (\cref{TH:disp_to_conf_lower_bound}).
    \item Fix some $0<c<1$. Provided  $\epsilon\geq  e^{-c (1-p)k+1 }$,
        for every $n\in\NN$,
        there exists a   regular $(\epsilon,p)$-confiner
        on $n$ vertices
        having  \emph{average} uniformity $\leq k$
        and sparsity $r\leq  \Omega_p((1-c)^{-1}\epsilon^{-1}k^{-1})$
        (simplification of \cref{TH:disp_to_conf_random_const}).
\end{enumerate}
Since $\epsilon\geq p^k-\frac{2k}{n}$ by our \cref{TH:worst_error_confine},
the bound (a) is weaker than the bound $r\geq
\Omega(\frac{p(1-p)^2}{(\epsilon-p^k) k})$
of \cref{CR:eps_view_show}. Moreover, the conclusion (a) is not strong enough
to imply there are no families of sparse confiners (and  {\it a
fortiori}  samplers)
for $p$-dense sets whose worst-case performance approaches the optimum.

However, when $\epsilon\geq e^{-c (1-p)k+1 }\geq e^{-(1-p)k}$,
the term $p^k$ is negligible relative   to $\epsilon$,
and as a result, \emph{in the regime $\epsilon\gg e^{-(1-p)k}$},
(a) recovers the lower bound of \cref{CR:eps_view_show},
and (b) shows that it is optimal up to constants depending on $p$.
In the context of one-sided error reduction,
\emph{under the assumption
    that
$\epsilon\gg  e^{-(1-p)k}$},
this recovers
our earlier conclusion that a pseudorandom method for choosing
$k$ seeds which is guaranteed to reduce one-sided error probability
from $p$ to $\epsilon$ must   need at least $\lg (\epsilon^{-1})-\lg k-O_p(1)$
additional random bits (on top of the random bits required for the algorithm),
and moreover means that adding just $\lg(\epsilon^{-1})-\lg k+O_p(1)$
random bits is in fact sufficient, provided one allows the pseudorandom method
to choose $k$ seeds \emph{on average}, rather than exactly $k$ seeds.
These observations, while derived from known results, seems new.

To conclude, our \cref{TH:hitting_bound_show} and
\cref{CR:eps_view_show} improves
upon
known bounds on  $(\epsilon,p)$-confiners in the regime $\epsilon\ll
e^{-(1-p)k}$ and in  particular when $\epsilon$ approaches $p^k$ from above.

\paragraph{Comparison to The Hitting Set Lemma.}

One of the most studied ways of getting $k$ pseudorandom samples
from a set $V$ is to perform a random walk
$\rv{\vec v} = (v_1,\ldots, v_k)$ of length $k-1$ on a
$d$-regular
expander  graph $G$ whose vertex set is $V$.
This is equivalent to picking a random edge in the hypergraph $H=(V,E)$
whose hyperedges are the length-$(k-1)$ paths in $G$.\footnote{
    \label{footnote:repetitions}
    Expander walks may have repetitions, so $H$ should be considered
    as a hypergraph where every hyperedge is equipped with weights on
    its vertices which indicate
    how many times each vertex is repeated. We ignore this technicality
    for simplicity. Moreover, it can be avoided by
    performing each step of the random walk according to a different expander,
    thus causing repetitions to very rare, hence negligible.
    Also,  if a random walk
    has repetitions, then we can artificially add more vertices
    until it sees exactly $k$ vertices, and this can only decrease
    the confinement probability of any set of vertices.
}
The celebrated \emph{Expander Hitting Set Lemma}  (\cite{Kahale95},
for instance) states that if $\lambda=\lambda(G)$
is the \emph{spectral expansion} of $G$,\footnote{
    The spectral expansion is
    largest magnitude of any non-trivial
    eigenvalue of the normalized adjacency matrix of $G$, i.e.,~$\lambda(G) =
    \max\set*{\lambda_2(G), |\lambda_{\min}(G)|}$
}
then
\[
    \Pp\sqbr*{ \rv{\vec v} = (v_1,\ldots, v_k) \subseteq A} \le p \cdot
    \parens*{(1-\lambda) p + \lambda}^{k-1}\qquad\forall~A \subseteq
    V~\textrm{of density $p$},
\]
Otherwise stated, it means that
the hypergraph $H$ is a \hyperref[eq:eps_confiner]{$(p \cdot
\parens*{(1-\lambda) p + \lambda}^{k-1},p)$-confiner}.

By picking $G$ to be an optimal expander graph with $\lambda
=\frac{2\sqrt{d-1}}{d}\approx 2 d^{-1/2}$, say
from the families constructed in
\cite{LubotzkyPS86, Morgenstern94},
and setting
$\ee := p ( (1-\lambda)p + \lambda)^{k-1} - p^k$,
we see that there exists a family
of $k$-uniform \hyperref[eq:eps_confiner]{$(p^k+\ee,p)$-confiners} with
sparsity
$r =d^{k-1}$.
By thinking of $r$, $d$ and $\lambda$
as functions of $\ee$,
elementary calculus implies
that as $\ee\to 0$,
we have $\lambda= p^{1-k}(1-p)^{-1}(k-1)^{-1}(\ee+O_{p,k}(\ee^2)) $,
and hence $r\approx (\frac{\lambda}{2})^{-2(k-1)}=(
2p^{k-1}(1-p)(k-1) )^{2(k-1)}\epsilon^{-2(k-1)}(1+O_{p,k}(\epsilon))$.
We conclude that once $k$ and $p$ are fixed,
for every
$\ee\geq 0$,
there exists a  family
of $(p^k+\ee,p)$-confiners with sparsity
\[
    r\leq \Omega_{k,p}(\ee^{-2(k-1)}).
\]
This, unfortunately, does not match our lower bound of $r\geq
\Omega_p(k^{-1}\ee^{-1})$.
It is an interesting problem to close the gap between these bounds.

We finally remark that constructing $k$-uniform sparse $(\epsilon,p)$-confiners
using the hitting set lemma is relevant only for $\epsilon\leq
e^{-\Theta_p(k)}$.
(This is opposite to the situation for confiners arising from random
    constructions
of dispersers, see our earlier discussion.)
Indeed, the promised confinement probability $\epsilon=p (
(1-\lambda)p + \lambda)^{k-1}$
is \emph{always} exponential in $k$.

\subsubsection{A Gap Phenomenon}

By putting \cref{cor:fraction-show}
and \cref{TH:hitting_bound_show} together,
we see that for sparse uniform hypergraphs with \emph{not-too-outlandish} vertex
degrees, there is a gap between the typical (\enquote{99.99\%})
confinement probability of density-$p$ sets and
the worst-case (\enquote{100\%}) confinement probability of such sets.
Formally:


\begin{corollary}
    \label{cor:gap_phenomenon_show}
    Fix some $p\in (0,1)$. Let $k>1$ be an integer, and let $r>0$ and
    $0<\beta<1$ be real numbers.
    Suppose $H = (V, E)$ is
    a $k$-uniform hypergraph  having $n\geq
    \poly(k,r,\frac{1}{p},\frac{1}{1-p})$ vertices and $r\cdot n$  edges
    such that no vertex is contained in more than an
    $n^{-\beta}$-fraction of the edges.
    Then the following statements hold
    simultaneously:
    \begin{itemize}
        \item (\emph{\enquote{99.99\%}-regime}) for every $\alpha\in (0,1)$,
            we have $\Pp_{\rv e}\sqbr*{\rv e \subseteq A} \le p^k +
            O(n^{-\frac{\beta(\alpha - 1)}2})$
            for all but an $O(\frac{k^2p^{k}}{1-p}) \cdot
            n^{-\alpha\beta}$-fraction
            of the density-$p$ subsets
            $A \subseteq V$;
        \item (\emph{\enquote{100.0\%}-regime}) there exists a
            density-$p$ subset
            $A \subseteq V$
            such that $\Pp_{\rv e}\sqbr*{\rv e \subseteq A} \ge p^k +
            \Omega (\frac{p(1-p)^2}{(r+1)k})$.
    \end{itemize}
\end{corollary}


%
We notice that dichotomies of this kind are not uncommon in average case
complexity. We refer to the discussion in \cite[Ch.~17.2]{AroraB09}:
\emph{Impagliazzo's Hard-Core Lemma} (\cite{Impagliazzo95}) implies
that a function which is hard to
compute in the \emph{average} case, should be \emph{extremely} hard to compute
on a small subset of inputs, but on the rest of the inputs it can indeed be
very easily computable.

\subsubsection{Fine-Scale Bounds on Confiners}
\label{subsec:fine-scale}

Changing our perspective a little, consider the situation where
we are given
\emph{fixed}
$r>0$, $k>1$ and $p\in (0,1)$ and are interested in finding an infinite  family
of \hyperref[eq:eps_confiner]{$(\epsilon,p)$-confiners}   each having
average uniformity at most $k$ and sparsity
(hyperedges-to-vertices
ratio) at most $r$. How small can $\epsilon$ be?
Our \cref{TH:worst_error_confine} implies that there is a  nontrivial
lower bound   $\epsilon\geq p^k+\Omega_p(r^{-1}k^{-1})$ when $k\geq 2$
(and a slightly more complicated one for $1<k<2$, q.v.\ \cref{TH:lower-bound}).
In fact, we can give a  better and more precise lower bound.

\begin{theorem}[Simplification of
    \cref{TH:lower-bound-strong}]\label{TH:lower_bound_strong_show}
    Let $r>0$, $k>1$ and $p\in (0,1)$. Suppose that there
    exists a   family of
    \hyperref[eq:eps_confiner]{$(\epsilon,p)$-confiners}  (with number
    of vertices tending to $\infty$) such that every member of the
    family has average uniformity $k$,  sparsity $\leq r$, and is
    either regular or has minimum vertex degree $\geq 3$.
    Then
    \[
        \epsilon\geq f_{k,r}(p):=1-(1-p)^{\frac{1}{rk}}
        +(1-p)^{\frac{1}{rk}}
        (1-(1-p)^{\frac{rk-1}{rk}})^k.
    \]
\end{theorem}

\begin{remark}
    If only assumes that every vertex of $H$ is included in some
    hyperedge, then the weaker
    bound $\epsilon\geq f_{k,3r}(p)$ holds. We conjecture
    that the stronger bound $\epsilon\geq f_{k,r}(p)$ should hold in general
    and therefore stick with it here for simplicity.
    Also, by a standard dualization argument
    (\cref{PR:conf_duality}), the theorem
    implies  that  $\epsilon\geq 1-f_{kr,1/r}^{-1}(1-p)$. This gives
    a negligible improvement
    for small $p>0$. See \cref{CR:all_lower_bounds}
    and \cref{CR:all_lower_bounds_dual} for
    a summary of all our lower bounds.
\end{remark}

It can be shown that for $k\geq 2$, we have $f_{k,r}(p)\geq
p^k+\Omega(\frac{p(1-p)^2}{rk})$ (\cref{TH:distance-to-delta-to-k}), and this
is actually how we prove \cref{TH:worst_error_confine}.
The gap between graphs of $f_{k,r}(p)$ and $p^k$  is   visible for
for small values of $k$ and $r$, as illustrated in \cref{FG:f-k-r}.

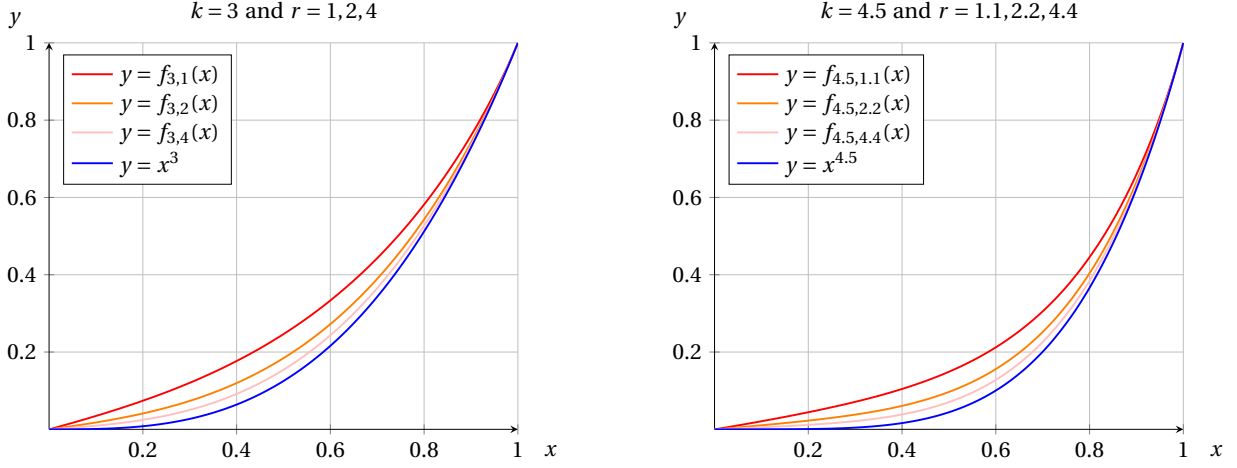
\begin{figure}[h]
    \centering
    \begin{minipage}{0.45\textwidth}
        \centering
        \begin{tikzpicture}[scale=0.9]
            \begin{axis}[
                    title={$k=3$ and $r=1,2,4$},
                    xlabel={$x$},
                    ylabel={$y$},
                    xmin=0, xmax=1,
                    ymin=0, ymax=1,
                    axis lines=middle,
                    grid=both,
                    grid style={line width=.1pt, draw=gray!10},
                    major grid style={line width=.2pt,draw=gray!50},
                    xtick={0, 0.2, 0.4, 0.6, 0.8, 1.0},
                    ytick={0, 0.2, 0.4, 0.6, 0.8, 1.0},
                    legend pos=north west,
                    label style ={at={(ticklabel cs:1.1)}},
                    legend cell align={left},
                ]
                \addplot[domain=0:1, samples=100, smooth, thick, red]
                {1-(1-x)^(1/3)+(1-x)^(1/3)*(1-(1-x)^(2/3))^3};
                \addlegendentry{$y=f_{3,1}(x)$};
                \addplot[domain=0:1, samples=100, smooth, thick, orange]
                {1-(1-x)^(1/6)+(1-x)^(1/6)*(1-(1-x)^(5/6))^3};
                \addlegendentry{$y=f_{3,2}(x)$};
                \addplot[domain=0:1, samples=100, smooth, thick, pink]
                {1-(1-x)^(1/12)+(1-x)^(1/12)*(1-(1-x)^(11/12))^3};
                \addlegendentry{$y=f_{3,4}(x)$};
                \addplot[domain=0:1, samples=100, smooth, thick, blue] {x^3};
                \addlegendentry{$y=x^3$};
            \end{axis}
        \end{tikzpicture}
    \end{minipage}
    \hfill
    \begin{minipage}{0.45\textwidth}
        \centering
        \begin{tikzpicture}[scale=0.9]
            \begin{axis}[
                    title={$k=4.5$ and $r=1.1,2.2,4.4$},
                    xlabel={$x$},
                    ylabel={$y$},
                    xmin=0, xmax=1,
                    ymin=0, ymax=1,
                    axis lines=middle,
                    grid=both,
                    grid style={line width=.1pt, draw=gray!10},
                    major grid style={line width=.2pt,draw=gray!50},
                    xtick={0, 0.2, 0.4, 0.6, 0.8, 1.0},
                    ytick={0, 0.2, 0.4, 0.6, 0.8, 1.0},
                    legend pos=north west,
                    label style ={at={(ticklabel cs:1.1)}},
                    legend cell align={left},
                ]
                \addplot[domain=0:1, samples=100, smooth, thick, red]
                {1-(1-x)^(1/4.95)+(1-x)^(1/4.95)*(1-(1-x)^(3.95/4.95))^4.5};
                \addlegendentry{$y=f_{4.5,1.1}(x)$};
                \addplot[domain=0:1, samples=100, smooth, thick, orange]
                {1-(1-x)^(1/9.9)+(1-x)^(1/9.9)*(1-(1-x)^(8.9/9.9))^4.5};
                \addlegendentry{$y=f_{4.5,2.2}(x)$};
                \addplot[domain=0:1, samples=100, smooth, thick, pink]
                {1-(1-x)^(1/19.8)+(1-x)^(1/19.8)*(1-(1-x)^(18.8/19.8))^4.5};
                \addlegendentry{$y=f_{4.5,4.4}(x)$};
                \addplot[domain=0:1, samples=100, smooth, thick, blue] {x^4.5};
                \addlegendentry{$y=x^{4.5}$};
            \end{axis}
        \end{tikzpicture}
    \end{minipage}
    \caption{The gap between the graphs $y=x^k$ and $y=f_{k,r}(x)$ for
    some values of $k$ and $r$.}
    \label{FG:f-k-r}
\end{figure}

For comparison, our earlier discussion of the hitting set lemma
shows that when $k$ is an integer and $r=d^{k-1}$ for an integer
$d\geq 3$, there are infinite families of
$r$-sparse $k$-uniform $(\epsilon,p)$-confiners
with
\[\epsilon= p\parens*{p+(1-p)\cdot\frac{2\sqrt{d-1}}{d}}^{k-1}=
    p\parens*{p+(1-p)\frac{2\sqrt{\sqrt[k-1]{r}-1}}{\sqrt[k-1]{r}}}^{k-1}
=:g_{k,r}(p).\]
Furthermore, by restating the results of
\cite{RadhakrishnanTS00} about random disperseres in the language of
confiners, which we do in \cref{apx:dispersers},
it follows that for every   $k>1$ and $\frac{1}{k} \lesssim
r\lesssim\frac{\ln((1-p)^{-1})}{k}e^{(1-p)k-1}$ ,
there exists an infinite family
of $r$-sparse $(\epsilon,p)$-confiners of average uniformity $\leq k$,
where
\[\epsilon\approx \frac{\ln((1-p)^{-1})}{kr}.
\]
(see \cref{TH:disp_to_conf_random_const}(ii) for a precise formulation of
    the constraints on $r$ and $\epsilon$ and note that
we always have $\epsilon\geq e^{-(1-p)k+1}$). Denoting this $\epsilon$
by $h_{k,r}(p)$, the graphs of
of the functions $f_{k,r}(p)$, $g_{k,r}(p)$ and $h_{k,r}(p)$ along
with that of $p^k$
are shown in \cref{FG:hsl_and_lower_bound} for some $r,k$
where all are defined.
From these graphs it can be seen that there is still
a considerable gap between our lower bound on
$(\epsilon,p)$-confiners and the confiners
arising from the  \ref{eq:hsl} or from random constructions.
This raises
the following problem, which seems very difficult.

\begin{problem}
    For every $r>0$, $k>1$ and $p\in (0,1)$,
    find the supremum $\epsilon(k,r,p)$ of the set of $\epsilon>0$ for which
    there exists an infinite family of $r$-sparse
    \hyperref[eq:eps_confiner]{$(\epsilon,p)$-confiners}
    with average uniformity $k$.
\end{problem}

As of now, we can only conclude that $\epsilon(k,r,p)\in
[f_{k,r}(p),\min\{g_{r,k}(p),h_{r,k}(p)\}]$,
but it seems likely that neither the lower bound nor the upper bound are tight.

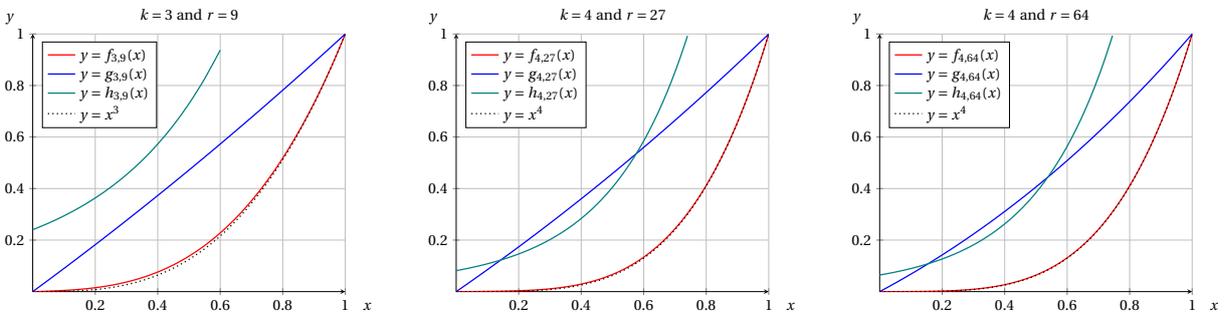
\begin{figure}[h]
    \centering
    \begin{minipage}{0.3\textwidth}
        \centering
        \begin{tikzpicture}[scale=0.6]
            \begin{axis}[
                    title={$k=3$ and $r=9$},
                    xlabel={$x$},
                    ylabel={$y$},
                    xmin=0, xmax=1,
                    ymin=0, ymax=1,
                    axis lines=middle,
                    grid=both,
                    grid style={line width=.1pt, draw=gray!10},
                    major grid style={line width=.2pt,draw=gray!50},
                    xtick={0, 0.2, 0.4, 0.6, 0.8, 1.0},
                    ytick={0, 0.2, 0.4, 0.6, 0.8, 1.0},
                    legend pos=north west,
                    label style ={at={(ticklabel cs:1.1)}},
                    legend cell align={left},
                ]

                \addplot[domain=0:1, samples=100, smooth, thick,  red]
                {1-(1-x)^(1/27)+(1-x)^(1/27)*(1-(1-x)^(26/27))^3};
                \addlegendentry{$y=f_{3,9}(x)$};

                \addplot[domain=0:1, samples=100, smooth, thick, blue]
                {x * (x + (1-x)*0.942809)^2};
                \addlegendentry{$y=g_{3,9}(x)$};

                \addplot[
                    smooth, 
                    thick,
                    teal
                ]
                coordinates {
                    (0.000000,0.240195)
                    (0.050000,0.265437)
                    (0.100000,0.294032)
                    (0.150000,0.326511)
                    (0.200000,0.363491)
                    (0.250000,0.405691)
                    (0.300000,0.453947)
                    (0.350000,0.509230)
                    (0.400000,0.572667)
                    (0.450000,0.645563)
                    (0.500000,0.729432)
                    (0.550000,0.826022)
                    (0.600000,0.937353)
                };
                \addlegendentry{$y=h_{3,9}(x)$};

                \addplot[domain=0:1, samples=100, smooth, thick,
                black, dotted] {x^3};
                \addlegendentry{$y=x^3$};

            \end{axis}
        \end{tikzpicture}
    \end{minipage}
    \hfill
    \begin{minipage}{0.3\textwidth}
        \centering
        \begin{tikzpicture}[scale=0.6]
            \begin{axis}[
                    title={$k=4$ and $r=27$},
                    xlabel={$x$},
                    ylabel={$y$},
                    xmin=0, xmax=1,
                    ymin=0, ymax=1,
                    axis lines=middle,
                    grid=both,
                    grid style={line width=.1pt, draw=gray!10},
                    major grid style={line width=.2pt,draw=gray!50},
                    xtick={0, 0.2, 0.4, 0.6, 0.8, 1.0},
                    ytick={0, 0.2, 0.4, 0.6, 0.8, 1.0},
                    legend pos=north west,
                    label style ={at={(ticklabel cs:1.1)}},
                    legend cell align={left},
                ]

                \addplot[domain=0:1, samples=100, smooth, thick, red]
                {1-(1-x)^(1/108)+(1-x)^(1/108)*(1-(1-x)^(107/108))^4};
                \addlegendentry{$y=f_{4,27}(x)$};

                \addplot[domain=0:1, samples=100, smooth, thick, blue]
                {x * (x + (1-x)*0.9428090)^3};
                \addlegendentry{$y=g_{4,27}(x)$};

                \addplot[
                    smooth, 
                    thick,
                    teal
                ]
                coordinates {
                    (0.000000,0.081423)
                    (0.030000,0.088437)
                    (0.060000,0.096230)
                    (0.090000,0.104904)
                    (0.120000,0.114570)
                    (0.150000,0.125359)
                    (0.180000,0.137414)
                    (0.210000,0.150900)
                    (0.240000,0.166002)
                    (0.270000,0.182930)
                    (0.300000,0.201920)
                    (0.330000,0.223239)
                    (0.360000,0.247188)
                    (0.390000,0.274105)
                    (0.420000,0.304374)
                    (0.450000,0.338425)
                    (0.480000,0.376743)
                    (0.510000,0.419876)
                    (0.540000,0.468438)
                    (0.570000,0.523124)
                    (0.600000,0.584714)
                    (0.630000,0.654087)
                    (0.660000,0.732233)
                    (0.690000,0.820265)
                    (0.720000,0.919438)
                    (0.74, 0.99243220)
                };
                \addlegendentry{$y=h_{4,27}(x)$};

                \addplot[domain=0:1, samples=100, smooth, thick,
                black, dotted] {x^4};
                \addlegendentry{$y=x^4$};

            \end{axis}
        \end{tikzpicture}
    \end{minipage}
    \hfill
    \begin{minipage}{0.3\textwidth}
        \centering
        \begin{tikzpicture}[scale=0.6]
            \begin{axis}[
                    title={$k=4$ and $r=64$},
                    xlabel={$x$},
                    ylabel={$y$},
                    xmin=0, xmax=1,
                    ymin=0, ymax=1,
                    axis lines=middle,
                    grid=both,
                    grid style={line width=.1pt, draw=gray!10},
                    major grid style={line width=.2pt,draw=gray!50},
                    xtick={0, 0.2, 0.4, 0.6, 0.8, 1.0},
                    ytick={0, 0.2, 0.4, 0.6, 0.8, 1.0},
                    legend pos=north west,
                    label style ={at={(ticklabel cs:1.1)}},
                    legend cell align={left},
                ]

                \addplot[domain=0:1, samples=100, smooth, thick, red]
                {1-(1-x)^(1/256)+(1-x)^(1/256)*(1-(1-x)^(255/256))^4};
                \addlegendentry{$y=f_{4,64}(x)$};

                \addplot[domain=0:1, samples=100, smooth, thick, blue]
                {x * (x + (1-x)*0.8660254)^3};
                \addlegendentry{$y=g_{4,64}(x)$};

                \addplot[
                    smooth, 
                    thick,
                    teal
                ]
                coordinates {
                    (0.000000,0.064446)
                    (0.030000,0.071012)
                    (0.060000,0.078371)
                    (0.090000,0.086626)
                    (0.120000,0.095892)
                    (0.150000,0.106302)
                    (0.180000,0.118001)
                    (0.210000,0.131157)
                    (0.240000,0.145958)
                    (0.270000,0.162614)
                    (0.300000,0.181365)
                    (0.330000,0.202478)
                    (0.360000,0.226258)
                    (0.390000,0.253046)
                    (0.420000,0.283226)
                    (0.450000,0.317232)
                    (0.480000,0.355553)
                    (0.510000,0.398740)
                    (0.540000,0.447413)
                    (0.570000,0.502271)
                    (0.600000,0.564103)
                    (0.630000,0.633796)
                    (0.660000,0.712351)
                    (0.690000,0.800895)
                    (0.720000,0.900699)
                    (0.745,0.9934963)
                };
                \addlegendentry{$y=h_{4,64}(x)$};

                \addplot[domain=0:1, samples=100, smooth, thick,
                black, dotted] {x^4};
                \addlegendentry{$y=x^4$};

            \end{axis}
        \end{tikzpicture}
    \end{minipage}

    \caption{Comparison between our lower bounds on existence of families
        of $r$-sparse $k$-uniform
        \hyperref[eq:eps_confiner]{$(\epsilon,p)$-confiners}
        ($\epsilon\geq  f_{k,r}(p)$ always) to the upper bounds implied
        by the \ref{eq:hsl} ($\epsilon\leq g_{k,r}(p)$ is possible) and
        known  random constructions
        ($\epsilon\leq h_{k,r}(p)$ is possible)
    for some  values of $k$ and $r$.}
    \label{FG:hsl_and_lower_bound}
\end{figure}

\subsubsection{The Ideas Behind Our Proofs}

Let $H=(V,E)$ be a $k$-uniform hypergraph on $n$ vertices and $rn$ edges.
For a set of vertices $A\subseteq V$, we denote by $E_r(A)$
the set of hyperedges $e\in H$ for which $|e\cap A|=r$, and by $E(A)$
set of hyperedges contained in $A$ (so that $E(A)=E_k(A)$). Similarly, we will
write $V(B) = \bigcup_{e \in B} e$ for the subset of vertices touched by $B
\subseteq E$. Fix $p\in (0,1)$ and let $\rv A$ be chosen uniformly at
random among all
density-$p$ subsets of $V$.

Our \cref{cor:fraction-show} about the typical behavior of sampling in $H$
are obtained by determining the expectation of $\frac{|E_r(\rv
A)|}{|E|}=\Pp\sqbr*{|\rv e\cap \rv A|=r}$ for $0\leq r\leq k$,
which turns out to be close to ${k\choose r}p^r(1-p)^{k-r}$,\footnote{More
    accurately, it turns out to be   probability
    that to draw the number $r$ according to the hypergeometric
    distribution $\Hyp(k,pn,n)$,
    which is   $\frac{\binom{k}{r}
\binom{n-k}{pn-r}}{\binom{n}{pn}}$.}
and showing that its variance is very small; see \cref{lem:hypg} for
a precise statement.
With this at hand, Chebyshev's Inequality implies that
$\frac{|E_r(\rv A)|}{|E|}$ must be very close to ${k\choose r}p^r(1-p)^{1-r}$
almost surely.
\cref{cor:fraction-show}
is then proven following some algebraic arguments together with this
fact. \cref{TH:worst_error_confine} follows just from our computation of
the expectation
of $\frac{|E_r(\rv A)|}{|E|}$.

%

Our results about the worst-case confinement probability
are derived from \cref{TH:lower_bound_strong_show},
whose proof is much more involved.
%
In order to prove it, we need
to find a set $A\subseteq V$ of density $p$ or less such that
$E(A)$ is as large as possible.
One way to guarantee that $E(A)$ will contain many edges is to start
with a set of hyperedges $B\subseteq E$ and take $A=V(B)$ as it will
guarantee that $E(A)\supseteq B$.
A key result that we establish shows that if $ B\subseteq E$ is
chosen at random with some fixed density, then \emph{on average},
the density of $E(A)=E(V(B))$ in $E$ will be strictly larger than that of $B$.
This is stated formally in the following proposition, in which we restricted
to regular hypergraphs for simplicity.

\begin{proposition}[Simplified Version of
    \cref{PR:edges-fraction-avg}]\label{PR:edge_fraction_show}
    Let $H = (V, E)$ be a $d$-regular hypergraph ($d \ge 2$) with
    hyperedges-to-vertices ratio $r := |E|/|V|$ such that $|e| (d-1) <
    |E|$ for all $e \in E$.
    Suppose $\gamma \in (0, 1)$ and let $\rv B \subseteq E$ be
    sampled at random by independently including each edge  $e \in E$
    in $\rv B$ with probability $\gamma$. Then,
    \begin{equation}\label{EQ:extra_density}
        \Exp\sqbr*{ \frac{|E(V(\rv B))| }{|E|} } \ge \gamma + (1 -
            (1-\gamma)^{d -
        1})^{\frac{d}{r}}.
    \end{equation}
\end{proposition}

In the right hand side of \cref{EQ:extra_density}, the term $\gamma$
comes from the fact that $\rv B\subseteq E(V(\rv B))$,
and the extra term $(1 - (1-\gamma)^{d - 1})^{\frac{d}{r}}$ is a
lower bound on the
average density of set of  hyperedges
that are unintentionally covered by the edges $\rv B$.
It is precisely this extra term that ultimately allows us to show
that the worst-case confinement probability for density-$p$ set of
vertices is $p^k+\Omega_p(r^{-1}k^{-1})$.

To utilize \cref{PR:edge_fraction_show} for our purposes,
we show that there is a particular value of $\gamma$
for which the density of $\rv A=V(\rv B)$ is almost surely
$p+o(1)$. Remarkably, we do this by applying our results about the
\emph{typical} behavior of sampling   to   the \emph{dual} hypergraph
of $H$  ---
indeed,  in the dual hypergraph of $H$, the complement $V-\rv A$ is the set
of \enquote{edges} which are confined to the set of
\enquote{vertices} $E-\rv B$, so
this is a question
about typical confinement probability (but in a possibly non-uniform
hypergraph).
Putting \cref{PR:edge_fraction_show} together with the
fact $\rv A$ almost surely has density $p+o(1)$ is enough
to show that for \emph{some} $B\subseteq E$ of density $\gamma$,
the set $A=E(B)$ will have density $p+o(1)$ in $V$
while at the same time the set $E(A)=E(V(B))$ will have density
at least $\gamma + (1 - (1-\gamma)^{d -
1})^{\frac{d}{r}})$, or a more complicated expression in the
non-regular case. The proof of \cref{TH:lower_bound_strong_show} then
concludes by showing that the lower bound
on $\frac{|E(V(B))|}{|E|}$ is at least $f_{r,k}(p)$.
This last step is immediate when $H$ is $d$-regular,
because $\gamma + (1 - (1-\gamma)^{d -
1})^{\frac{d}{r}})$ turns out to be exactly $f_{r,k}(p)$,
but in the general case, this is the content of
the very technical \cref{TH:hard-optimization}.

We hope that our methods, and in particular
\cref{PR:edge_fraction_show},  will find further applications.

%


It is also worth nothing that our proof of
the general version of \cref{PR:edge_fraction_show}
indicates that the density of $E(V(\rv B))$ is the smallest
when every two  hyperedges share at most one vertex. In particular,
this suggests
that \emph{simplicial} or \emph{cell complexes} should be far from being
optimal confiners, q.v.~\cref{RM:not-simplicial}.

\subsubsection{Implications to Dispersers}

A disperser is a basic object used in error-reduction for one-sided
randomized algorithms,
cf.~\cite{NisanT99, Shaltiel04, Goldreich11, Vadhan12}.
Following
\cite[Dfn.~6.19]{Vadhan12}
(for instance),
a $(k, \epsilon)$-disperser is a function
$F:\{0,1\}^n\times\{0,1\}^d\to \{0,1\}^m$
where for every $ 2^k$-element subset $A\subseteq \{0,1\}^n$,
the image of $A\times\{0,1\}^d$ under $F$ has at least
$2^m(1-\epsilon)$ elements.
The rationale behind this is that given a random string $w\in \{0,1\}^n$
having a possibly-non-uniform distribution, but which, loosely speaking,
has $k$ bits of entropy, we could compute $F(w,v)$ with a uniformly
random $v\in\{0,1\}^d$
to get an \enquote{almost} random string of $m$ bits.
One is interested in keeping $d$ as small as possible while having $m$
as large as possible.
The latter is equivalent to asking that the \emph{entropy loss}
$\ell:=k+d-m$ will be as small as possible.
Put $N=2^n$, $M=2^m$, $D=2^d$ and $K=2^k$.
Then we can  think of $F$ as representing the incidences in
a $D$-regular hypergraph on $N$ vertices and $M$ edges,
in which case the requirement of being a $(K,\epsilon)$-disperser
is equivalent to being an
\hyperref[eq:eps_confiner]{$(\epsilon,(1-\frac{K}{N}))$-confiner}.
Combining this observation with our \cref{TH:worst_error_confine} and
\cref{TH:hitting_bound_show},
we get the following new bounds on dispersers,
which best-apply in the regime where $k=n-O(1)$ and the entropy loss
$\ell=k+d-m$ is $O(1)$.

In the language of dispersers our results say that:

\begin{theorem}\label{TH:disp-lower-bound-show}
Fix $u\in\NN$.
Let $F:\{0,1\}^n\times\{0,1\}^d\to \{0,1\}^m$
be an $( {n-u},\epsilon)$-disperser
with entropy loss $\ell=(n-u)+d-m$, and put $L=2^\ell$ and $U=2^u$.
Then,
\begin{enumerate}[(i)]
\item $\epsilon\geq (1-\frac{1}{U})^{LU}- \frac{2LU}{N} \approx
    e^{-L}- \frac{2LU}{N} $,
\item
    putting $\ee:=\epsilon-(1-\frac{1}{U})^{LU}$
    and assuming $n\geq \Theta_u(\ee^{-1}+\ell)$,
    we have
    $d \geq -\lg [\epsilon-(1-\frac{1}{U})^{LU}]-O_u(1)$.
\end{enumerate}
Here, the hidden constants depend only on $u$.
\end{theorem}

\begin{proof}
Let $H$ be the corresponding $D$-regular hypergraph with $N$ vertices and $M$
vertices as above. Since
$\frac{2^{n-u}}{N}=U^{-1}$, it is an
\hyperref[eq:eps_confiner]{$(\epsilon,1-\frac{1}{U})$-confiner}.
Moreover, the sparsity of $H$ is $R:=\frac{M}{N}$
and its average uniformity  is $D/R=2^{n+d-m}=LU$
(\cref{PR:degree-to-unif-ratio}).
Thus, by \cref{TH:worst_error_confine},
$\epsilon\geq (1-U^{-1})^{L/U}- \frac{2LU}{N} $.
Furthermore, for $\ee$ as in the theorem,
we have by \cref{TH:hitting_bound_show} that
$R\geq \Omega_u((LU)^{-1}\ee^{-1})$, provided
that $N\geq\poly_u(\ee^{-1},UL)$, which holds if $n\geq
\Theta_u(\ee^{-1}+u+\ell)=
\Theta_u(\ee^{-1} +\ell)$ (because $u$ is fixed).
Multiplying both sides of $R\geq \Omega_u((LU)^{-1}\ee^{-1})$ by $LU$
and recalling that $D=R\cdot LU$,
gives $D\geq  \Omega_u([\epsilon-(1-\frac{1}{U})^{LU}]^{-1})$,
and taking logarithms gives (ii).
\end{proof}

A graph-theoretic analogue of this result also holds, but we refrain
from stating it here to avoid
duplication of the presented material.  Part (ii) of
\cref{TH:disp-lower-bound-show}  improves on the lower bound $d\geq
-\lg \epsilon +\lg u-O(1)$ of \cite[Thm.~1.5(a)]{RadhakrishnanTS00},
although the latter
holds without assuming $u$ is constant and without restrictions on $n$.
Part (i) of \cref{TH:disp-lower-bound-show} appears to be new.

%

\subsubsection{Implications to Vertex Expanders}

A $(K,A)$-vertex expander is a hypergraph\footnote{
Most texts describe it as a two-sided graph.
} $H=(V,E)$
in which  every set $S\subseteq V$ of $K$ or less
vertices touches at least $A\cdot|S|$ edges, cf.~\cite[Dfn.~4.1]{Vadhan12},
for instance.
In  \cite[Prob.~4.7]{Vadhan12}, a slightly
weaker notion is introduced: $H$
is a $(=K,A)$-vertex expander if every set
$S\subseteq V$ of exactly $K$ vertices touches at least
$AK$ hyperedges.
By looking at the complement of $S$, this is the same
as saying that $H$ is a
\hyperref[eq:eps_confiner]{$(1-\frac{KA}{|E|},1-\frac{A}{|V|})$-confiner}.
Thus, our \hyperref[TH:worst_error_confine]{Theorems
\ref{TH:worst_error_confine}}
and \ref{TH:hitting_bound_show} imply the following new bounds.

\begin{corollary}
Let $\ee, \delta\in (0,1)$ and let $H=(V,E)$
be a hypergraph with $n$ vertices, $r\cdot n$
hyperedges and average uniformity $k$. Suppose that every set of
$\lceil \delta n\rceil$
vertices in $H$ meets at least $\lceil \ee rn \rceil$ hyperedges.
Then:
\begin{enumerate}[(i)]
\item $\ee \leq 1-(1-\delta)^k+ \frac{2k}{n} $, and
\item $\ee \leq
    1-(1-\delta)^k-\Omega(\frac{\delta^2(1-\delta)}{rk})$, provided
    $n\geq\poly_\delta(r,k)$.
\end{enumerate}
\end{corollary}

\subsection{Relation to Other Works}
\label{subsec:relation}


\paragraph{Comparison to Expander-Walk-Based Samplers.}


There are many works which study how length-$(k-1)$ random walks on expander
graphs
perform as samplers, Chernoff samplers and confiners.
In more detail, let $G = (V, E)$ be a $d$-regular graph
with spectral expansion $\lambda := \lambda(G)$, and let $H=(V,E')$
be the hypergraph whose hyperedges are the paths of length $k-1$ in
$G$.\footnote{
As in page~\pageref{footnote:repetitions}, we suppress
the technicality that random walks may repeat the same vertex several times.
See \cref{footnote:repetitions}.
}
Then the celebrated \hyperref[eq:hsl]{(Expander) Hitting Set Lemma}
states that
\begin{equation}\tag{Hitting Set Lemma}\label{eq:hsl}
\Pp\sqbr*{ \rv{\vec v} = (v_1,\ldots, v_k) \subseteq A} \le p \cdot
\parens*{(1-\lambda) p + \lambda}^{k-1}\qquad\forall~A \subseteq
V~\textrm{of density $p$},
\end{equation}
implying that $H'$ is an $(p((1-\lambda)p+\lambda)^{k-1},p)$-confiner.
(In this case,  lower bounds on  $\Pp\sqbr*{ \rv{\vec v}   \subseteq
A}$ are also known, cf.~\cite{AlonFW95, Ta-ShmaZ24}.)
Moreover, by a sequence of works \cite{Kahale97, Gillman98, Healy08,
GolowichV22}, $\rv N_H(A)$ satisfies a Chernoff-like tail bound
\begin{equation}\label{eq:exp_chernoff}\tag{Expander Chernoff
Bound}\Pp\sqbr*{ \Abs*{\rv N_H(A) - kp} \ge k \eta } \le \exp\parens*{ -
\Omega(1 - \lambda) \cdot k \eta^2 }
\,\quad\forall~A \subseteq V~\textrm{of density $p$}
\end{equation}
for all $\eta > 0$,\footnote{See \cite[Cor.~2]{GolowichV22} for a more
quantitative statement.} and the number-of-hits distributes nearly
binomially, namely,
\begin{equation}\label{eq:eps-estimation}
\dist_{\TV}( \rv N_H(A), \Bin(k, p))\leq \Theta(\lambda)
\qquad
\text{for every
$A\subseteq V$
of density $p$,}
\end{equation}
i.e.,~$H$ is a $\Theta(\lambda)$-sampler.

There are three main differences between our \cref{cor:fraction-show}
and these results.
First,  these results address \emph{all} subsets $A\subseteq V$ of density
$p$ while the conclusions of \cref{cor:fraction-show} apply only to
\emph{almost all}
such sets. Second, the above results address only (hypergraphs
arising from) \emph{expander walks},
while our result applies to \emph{all uniform hypergraphs}.
Third, when the expander graph $G$  is required to be sparse
(i.e.,\ $d=O(1)$), our \cref{cor:fraction-show} gives much better
bounds on the confinement probability,
tail bound, and sampling, albeit only for '99.99\%' of $A$-s of density $p$.
Indeed,
if the degree $d$ of $G$ is fixed, then $\lambda\geq
\frac{2\sqrt{d-1}}{d}=\Theta(1)$ by
the Alon--Boppana bound \cite{Alon86}, in which case the bounds
\cref{eq:dist-showcase},
\cref{eq:chernoff-showcase}, \cref{eq:confinement-showcase} from
\cref{cor:fraction-show} are eventually
better as $n\to\infty$ (and $k$ is not too large relative to\ $n$)
than their counterparts for expander walks recalled above.
This is particularly notable in the case of distance from the
binomial distribution,
where \cref{eq:eps-estimation} gives
$\dist_{\TV}( \rv N_H(A), \Bin(k, p))\leq \Theta(1)$ for \emph{all} $A$
of density $p$, while \cref{cor:fraction-show} say
that $\dist_{\TV}( \rv N_H(A), \Bin(k, p))\leq o(1)$ for \emph{almost
all} such $A$.
The latter holds despite the fact that by \cite[Thm.~5]{GolowichV22},
it is impossible to improve upon $\dist_{\TV}( \rv N_H(A), \Bin(k,
p))\leq \Theta(\lambda)$ for  all  $A$  in general expander walks.
To conclude the last point, every sparse hypergraph sampler with
\emph{not-too-outlandish} vertex
degrees performs better  than a sparse-expander-walk based sampler
\enquote{99.99\% of the time}.


\begin{remark}
The work of \cite{GolowichV22} studies the convergence of $\rv N_H(A)$ to
$\Bin(k, p)$ in the \emph{total variation} distance. Precursor works have
studied convergence in weaker distances such as the \emph{Kolmogorov
distance}\footnote{Let $\rv X, \rv Y$ be distributions supported on the reals.
This distance $\dist_K(\rv X, \rv Y)$ is defined by: $\dist_K(\rv X,
\rv Y ) = \sup_{t \in \RR} \Abs*{ \Pp\sqbr*{\rv X \le t} -
\Pp\sqbr*{\rv Y \le t}}$}
between distributions, see for example \cite{Lezaud01, Kloeckner17,
Kloeckner19}. Works of \cite{CohenPT21, CohenMPPT22}
generalized these precursor works by giving bounds in \emph{total
variation} distance
which are weaker than that of \cite{GolowichV22}. These works also studied
the advantage that IID samples give over expander walk based
samples for certain classes of circuits.

The subsequent work   \cite{Golowich23} (specifically Thm.~14
therein), which evolved from
\cite{GuruswamiK21, CohenPT21, CohenMPPT22, GolowichV22}, proves a
Berry-Esseen type theorem which shows that
for a  hypergraph $H  $ arising from a length-$(k-1)$ expander
walk, $\rv N_H(A)$ is at most $O(\lambda/k^{1/2 - o(1)})$ away
in the \emph{total variation} distance from a discrete Gaussian-like
distribution (which need  not be $\Bin(k,p)$ in general, as implied
by \cite[Thm.~5]{GolowichV22}, recalled in our earlier discussion). However,
we note that the \emph{limiting} Gaussian-like distribution
has a non-trivial dependency on the expander graph $G$ and the set $A
\subseteq V$. In this light, the discrepancy between the lowerbound
\cite[Thm.~5]{GolowichV22},
our typical case bound \cref{cor:fraction-show},
and the upperbound of \cite[Thm.~14]{Golowich23} can be interpreted by
stating that for a typical $A \subseteq V$ of density $p$, the
limiting distribution should
look \emph{binomial}, but for a \emph{vanishing} fraction of $p$-dense sets it
can be \emph{far} from binomial.\footnote{As permitted by the upperbound
of $O_p(\lambda)$ proven in \cite[Cor.~4]{GolowichV22}.}
\end{remark}


We compared  confiners coming from expander random walks with our results
about the worst-case confinement probability in
\hyperref[subsec:worst-case]{Subsections~\ref{subsec:worst-case}}
and~\ref{subsec:fine-scale}.

\paragraph{Comparison to Chernoff Bounds for High-Dimensional
Expanders.}  Chernoff-type (sub-Gaussian) tail bounds for hypergraphs
arising from  simplicial complexes (rather than expander walks) were
proven in \cite{DiksteinH24} under suitable assumptions of
\emph{high-dimensional expansion}. However, the assumptions of
\emph{high-dimensional expansion} for these
bounds to kick in are significantly
stronger than the assumption of \emph{expansion} of expander walks.
Therefore, these bounds fall short of recovering the \ref{eq:exp_chernoff}
when the graph $G$ has constant degree.
Consequently, in the case of expander walks, the typical behavior of
tail-bounds implied by \cref{cor:fraction-show}(ii) is better.

\paragraph{Other Related Work.}
Our lower-bound for the hitting set property can be thought as a
lower-bound for the expansion of hypergraphs. The work
\cite{LouisPR22} is close to our work in spirit and it compares the
expansion of hypergraphs to several notions of high-dimensional
expansion, in the sense of \cite{AlevJT19, DiksteinD19, GurLL22,
AlevP23, DiksteinH24}.

Quite recently,   a breakthrough result \cite{HsiehLMRZ25}
constructed explicit \emph{lossless} vertex expanders. These objects
are highly related to
the \hyperref[eq:eps_confiner]{$\epsilon$-confiners} we have studied, see the
discussion in \cite[Prob.~4.7]{Vadhan12}.

\subsection{Organization of the Paper}
In \cref{sec:Preliminaries}, we recall some general
definitions and results which will be used throughout the paper. In
\cref{sec:almost-every-set},
we present and  prove our typical-case results. In
\cref{sec:lowerbounds},
we   present our results about worst-case sampling.
The proofs of these results are
given in   in
\cref{sec:lower-bound-proof},
except for some
technical analytic arguments that are differed to \cref{sec:analytic}.
In~\hyperref[apx:dispersers]{Appendix \ref{apx:dispersers}}, we
restate some known results about
disperseres in the language
of confiners.

\subsection{Acknowledgements}

We thank Max Hopkins and Dor Minzer for useful and insightful
conversations and Amit Levi for his feedback on an earlier version of
this manuscript. The
second named author
thanks the MIT Mathematics Department for
their hospitality
while writing this manuscript.

This research was supported by an ISF Grant No.~721/2024.

\section{Preliminaries}\label{sec:Preliminaries}
\subsection{Discrete Probability Theory}
Throughout, we will use the boldface font, e.g. $\rv X$, to refer to
random variables. Let $\rv X$ be a discrete random variable taking
values over a set $\Omega$.  We will write $\Exp \rv X$ and $\Var(\rv
X)$ for the expectation and variance of the random variable $\rv X$,
respectively.

We recall the following well-known inequalities due to Markov and
Chebyshev, see for example \cite[Thms.~4.1, 4.2]{Wasserman13}:
\begin{fact}[Markov's Inequality]\label{fac:andrei}
Let $\rv X$ be a discrete random variable taking values over the
non-negative reals. Then, for any $t > 0$, we have:
\[ \Pp\sqbr*{ \rv X > t } \le \frac{\Exp \rv X}{t}. \]
\end{fact}
\begin{fact}[Chebyshev's Inequality]\label{fac:pafnuty}
Let $\rv X$ be a discrete random variable taking values over the
reals. Then, for any $t > 0$, we have:
\[ \Pp\sqbr*{ \Abs*{\rv X - \Exp \rv X} > t } \le \frac{\Var(\rv X)}{t^2}. \]
\end{fact}

We recall that a random variable taking values in $\{0, \ldots, k\}$
is said to have the \ref{eq:bind} with parameters $(k, p)$, denoted
$\rv X \sim \Bino(k, p)$, if its probability mass function satisfies:
\begin{equation}\tag{binomial distribution}\label{eq:bind}
\Pp\sqbr*{\rv X = \ell} = \binom{k}{\ell} \cdot p^\ell (1-p)^{k - \ell}
\end{equation}
As an urn experiment, $\Bin(k, p)$ is the distribution of the number
of red balls
obtained by pulling $k$ balls out of an urn containing $n$ balls, a
$p$-fraction of which is red and the others are blue, with replacement.

Recall further that a random variable $\rv X$ taking values in $\{0,
\ldots, k\}$ is said to have the \ref{eq:hypgd} with parameters $(k,
m, n)$, denoted $\rv X \sim \Hyp(k, m,n)$, if its probability mass
function satisfies:
\begin{equation}\tag{hypergeometric distribution}\label{eq:hypgd}
\Pp\sqbr*{\rv X = \ell} = \frac{\binom{k}{\ell} \cdot \binom{n - k}{m
- \ell}}{\binom{n}{m}}\quad\forall \ell = 0,\ldots, k.
\end{equation}
Similarly to the binomial distribution,   $\Hyp(k, m,n)$ is the
distribution describing the number of red balls obtained by pulling
$k$ balls out of on urn containing $m$ red and $n-m$ blue balls
without replacement.

When $n$ is large,   binomial and hypergeometric random variables
tend to concentrate around their mean. This is known as the
Chernoff bound
\cite[Thm.~4.5]{Wasserman13}.

\begin{fact}[Chernoff Bound]\label{fac:chernoff}
Let $\rv X \sim \Bin(k, p)$ or $\rv X \sim \Hyp(k, pn, n)$. Then for all
$\eta > 0$, we have:
\[ \Pp\sqbr*{ \Abs*{\rv X - pk } > k \cdot \eta } \le 2\exp\parens*{
- 2 k \eta^2 }. \]
\end{fact}
\begin{remark}
The Chernoff bound is typically stated for sums of independent random
variables.
For binomial random variables $\rv X \sim \Bin(k, p)$, we can think of $\rv X$
as $\sum_{i = 1}^k \rv Z_i$
where the $\rv Z_i$ are independent random variables with
distribution $\Ber(p)$.
The proof extends without much difficulty to hypergeometric random
variables, e.g.~see~\cite[Thm.~5.2]{Mulzer18}.
\end{remark}

Next, we recall that the \ref{eq:totv} distance between two random
variables $\rv X$ and $\rv Y$ taking values over the same set
$\Omega$ is given by
\begin{equation}\tag{total variation}\label{eq:totv}
\dist_{\TV}(\rv X, \rv Y) = \frac{1}{2} \sum_{x \in \Omega} \Abs*{
\Pp\sqbr*{\rv X = x} - \Pp\sqbr*{\rv Y = y} } = \max_{A \subseteq
\Omega} \Abs*{ \Pp[\rv X \in A] - \Pp[\rv Y \in A]}.
\end{equation}
By an abuse of notation, when we know that $\rv Y$ is distributed
with law $\mu$, we will simply write $\dist_{\TV}(\rv X, \mu)$ in
place of $\dist_{\TV}(\rv X, \rv Y)$.

It is well-known that the hypergeometric distribution $\Hyp(k,m,n)$ is
close in total variation to a binomial distribution when $k$
is small relative to $n$:
\begin{theorem}[\cite{Kunsch94}]\label{thm:binhyp}
Let $k, m,n \in \NN$ be given such that $k \le m \le n$. Writing $p =
m/n$ and $q = 1 - p$, if $n p q \ge 1$, we have:
\[
\frac{1}{28} \cdot \frac{k-1}{n- 1}\le \dist_{\TV}(\Hyp(k, pn,
n),\Bin(k, p)) \le \frac{k-1}{n- 1}
\]
\end{theorem}
The intuitive explanation of \cref{thm:binhyp} is that as $n$ grows
relative to $k$, the impact of not replacing the balls in the urn
experiment producing
the hypergeometric distribution becomes negligible.
\subsection{Hypergraphs}
\label{subsec:hypergraphs}

For our purposes, a \emph{hypergraph} $H = (V, E)$ (possibly with
repeated hyperedges)
consists of a set of vertices $V=V(H)$, a set of hyperedges $E=E(H)$
and an incidence relation  specifying which vertices
are included in which hyperedges.
We write $v<e$ to indicate that the vertex $v$ is in the hyperedge $e$,
and denote by $\Vrt(e)$ or $\Vrt_H(e)$ the set of vertices of $e$, i.e.,\
$\{v\in V\,:\,v<e\}$.
However, when there is no risk of confusion, we shall identify
$e$ with its set of vertices $\Vrt_H(e)$. Unless specified otherwise,
we assume that all hypergraphs are finite and
have at least one vertex and at least
one hyperedge.

Suppose that $H=(V,E)$ is a hypergraph.
Given
set of vertices $A\subseteq V$, we write $E(A)$ or $E_H(A)$ for the
set of hyperedges having
all their vertices in $A$, i.e.,
\[
E(A)=E_H(A)=\set*{e\in E ~:~ \Vrt(e)\subseteq A}.
\]
We also say that the edges in $E(A)$ are confined to $A$.
Similarly, given a set of
edges $B\subseteq E$,
we write $V(B)$ or $V_H(B)$ for the set of vertices covered
by the
edges from $B$, i.e.,
\[
V(B)=V_H(B)=\bigcup_{e\in B} \Vrt(e).
\]

The uniformity of a hyperedge $e \in E$ is the number of vertices
contained in it and is denoted $|e|$ or $|e|_H$. We say that the
hypergraph $H$ is $k$-uniform if every edge has uniformity $k$.
As usual,
the average uniformity $\bar k$ and maximum uniformity $k^\star$ of
$H = (V, E)$ are defined by
\[\bar k = \frac{1}{|E|} \cdot \sum_{e \in E}
|e|\quad\textrm{and}\quad k^\star = \max_{e \in E} |e|.\]
We also recall that degree of a vertex $v \in V$ in $H$ is the number
of hyperedges $e \in V$ containing it, i.e.,
\[
\ddeg_H(v) = \set*{ e \in V~:~v < e~~(\textrm{or equivalently}~v
\in \Vrt(e))}.
\]
When all vertices have the same positive degree $d$, we say that $H$
is $d$-regular. The hypergraph $H$ is called regular if it is
$d$-regular for some $d>0$ (but when all vertices have degree $0$,
equiv.\ all the hyperedges are empty, the hypergraph is not considered
to be regular).
In general, the average degree and maximum   degree of $H$ are
$\frac{1}{|V|}\sum_{v\in V}\ddeg v$
and
$\max\{\ddeg v : v\in V\}$, respectively. We recall the following result, which
is just a hypergraph analogue of the \emph{handshaking lemma}:
\begin{proposition}\label{PR:degree-to-unif-ratio}
Let $H=(V,E)$ be a hypergraph of average uniformity $\bar{k}$ and average degree
$\bar d $. Then:
\[
\frac{|E|}{|V|}=\frac{\bar{d}}{\bar{k}}.
\]
\end{proposition}

\begin{proof}
This is equivalent to showing $|E|\bar{k}= |V|\bar{d}$.
This holds because
\[
|E|\bar{k} = \sum_{e\in E}\sum_{v:v<e} 1
=\sum_{v\in V}\sum_{e:e>v} 1 = |V| \bar{d}.
\qedhere
\]
\end{proof}

Given a hypergraph $H=(V,E)$ and a subset $A\subseteq V$, we write
$\rv N_H(A)$ for the   random variable
obtained by choosing an edge $\rv e\in E$ uniformly at random
and counting how many of its vertices are in $A$,
i.e.,
$\rv N_H(A)=|\rv e\cap A|$ (here we identified $\rv e$ its set of
vertices). We call $\rv N_H(A)$ the \emph{number of hits} in the set
$A$.
%
Finally, recall that the dual hypergraph  of $H = (V, E)$ is the hypergraph
$H^* = (V^*, E^*)$ obtained
by replacing the roles of the vertices and the hyperedges.
That is,  $V^* = E$ and $E^* = V$, and we have $e<v$ in $H^*$
if and only if $v<e$ in $H$.
Observe $H$ is $k$-uniform if and only if $H^*$ is $k$-regular.
We also have:

\begin{proposition}\label{PR:conf_duality}
Let $H=(V,E)$ be an $(\epsilon,p)$-confiner.
Then $H^*=(E,V)$ is a $(1-p,1-\epsilon')$-confiner for every
$\epsilon'>\epsilon$.
\end{proposition}

\begin{proof}
Write $n=|V|$ and $m=|E|$.
Let $B\subseteq E$ be a subset of density $1-\epsilon'$
(i.e.\ $|B|=\lfloor(1-\epsilon')m\rfloor$),
and let $A$ be the set of vertices in $H$ that are incident only to
edges in $B$,
equiv., $A=E_{H^*}(B)$. We need to show that $|A|\leq (1-p)n$.
It is easy to see that $E_H(V-A)\supseteq E-B$. Since
$|E-B|\geq  \epsilon' m>\epsilon m$ and $H$
is an $(\epsilon,p)$-confiner, we must have $|V-A|\geq pn$, and it follows
that $|A|\leq (1-p)n$.
\end{proof}

\subsection{Some Inequalities}

We recall several inequalities that will be needed several times in the sequel.
Their proofs are easy  and are included for completeness.

\begin{lemma}\label{LM:diff_trick}
Let $I\subseteq \RR$ be an interval and let
$f:I\to \RR$ be a differentiable function
such that $A\leq f'(x)\leq B$ for all $x\in I$.
Then for every $x<y$ in $I$, we have
$A(y-x)\leq f(y)-f(x)\leq B(y-x)$.
\end{lemma}

\begin{proof}
$A(y-x)=\int_x^y A\,\dif t\leq \int_x^y f'(t)\dif t=f(y)-f(x)$.
The other inequality is shown similarly.
\end{proof}

\begin{lemma}\label{LM:power-dist}
Let $x,y,k,\varepsilon  $ be non-negative real numbers such that $x<y$.
\begin{enumerate}[(i)]
\item If $k\geq 1$ or $k=0$, then
    $y^k-x^k\leq k y^{k-1} (y-x)$.
    In particular, $y^k-x^k\leq k(y-x)$ when $x,y\in [0,1]$.
\item If $0<k\leq  1$ and $x<1$, then $y^k-x^k\leq
    \frac{y-x}{ex(1-x)}$.
\end{enumerate}
\end{lemma}

\begin{proof}
(i) This clear for $k=0$, so assume $k\geq 1$.
Let $f(t)=t^k$. Then for $t\in [x,y]$, we have $f'(t)=kt^{k-1}\leq
ky^{k-1}$, so by \cref{LM:diff_trick}, $y^k-x^k\leq ky^{k-1}(y-x)$.


(ii)
Since $0<k\leq 1$,
the function $f(t)=t^{k}$ is concave down
in $(0,\infty)$. Thus, for every $a,\ee\geq 0$, we have
\[(a+\ee)^{k}=f(a+\ee)\leq f(a)+\ee\cdot
f'(a)=a^{k}+\ee \cdot k a^{k-1}.\]
When $a<1$, elementary analysis implies that the function
$k\mapsto ka^{k-1}$
has a global maximum at $\parens*{\frac{1}{\ln
(a^{-1})},\frac{1}{ea\ln(a^{-1})}}$,
so
\[(a+\ee)^{k}\geq a^{k}+\ee \cdot \frac{1}{ea\ln(a^{-1})}
\leq a^{k}+\ee \cdot \frac{1}{ea(1-a)} \]
(in the second inequality we used that $\ln a\leq a-1$).
Applying this with $a=x$ and $\ee=y-x$ gives what we want.
\end{proof}


\begin{lemma}\label{LM:power-ineq}
For every $x,t\in (0,1)$, we have
$tx< 1-(1-t)^x\leq -\ln(1-t)x$.
The right inequality actually holds for every real $x$.
\end{lemma}


\begin{proof}
The left inequality is equivalent to
$(1-t)^x< 1-tx$, which is a variant of Bernoulli's inequality
that can be proved as follows: The function $f(x)=(1-t)^x$ is
strictly convex, so the line segment connecting
$(0,f(0))=(0,1)$ and $(1,f(1))=(1,1-t)$ is above the graph of $f(x)$.
The equation of that line segment is $y=1-tx$, hence the inequality.
The right inequality holds because
$1-(1-t)^x=1-e^{x\ln(1-t)}\leq 1-(1+\ln(1-t)x)=-\ln(1-t)x$.
(Here, we used the fact that $e^a\geq 1+a$ for all real $a$.)
%
%
\end{proof}

\begin{lemma}\label{LM:ln-approximation}
For every $x\in (0,1)$, we have
$
\ln(1-x)> \frac{-x+0.5x^2}{1-x}
$.
\end{lemma}

\begin{proof}
\begin{align*}
\ln(1-x)
&=-x-\frac{x^2}{2}-\frac{x^3}{3}-\frac{x^4}{4}-\dots
> -x-\frac{x^2}{2}-\frac{x^3}{2}-\frac{x^4}{2}-\cdots
\\
&= -x-\frac{x^2}{2}(1+x+x^2+\dots)
=-x-\frac{x^2}{2}\cdot\frac{1}{1-x}=\frac{-x(1-x)-0.5x^2}{1-x}=\frac{-x+0.5x^2}{1-x}.
\qedhere
\end{align*}
\end{proof}

\begin{lemma}\label{LM:log-vs-poly}
Let $a>0$. Then for every $x\geq (\frac{2}{a})^{\frac{2}{a}}$, we
have $\ln x\leq x^a$.
\end{lemma}

\begin{proof}
This is clear for $a\geq 1$, because $\ln x\leq x\leq x^a$ for all
$x>0$, so assume
$a\in (0,1)$. Then,
by elementary analysis, the function $f(x)= x^a-\ln x$ is increasing
when $x\geq (\frac{1}{a})^{\frac{1}{a}}$. It is therefore enough to
check that $f((\frac{2}{a})^{\frac{2}{a}})\geq 0$. Indeed,
\[
f\parens*{\parens*{\frac{2}{a}}^{\frac{2}{a}}}
=\parens*{\frac{2}{a}}^2-\frac{2}{a}\ln\parens*{\frac{2}{a}}
\geq \parens*{\frac{2}{a}}^2-\frac{2}{a}\cdot\frac{2}{a}\geq 0.
\qedhere
\]
\end{proof}

\section{Performance of Hypergraph Samplers in the Typical Case}
\label{sec:almost-every-set}


In this section we prove several results
sharing a common theme: The sampling properties of \emph{every}
$k$-uniform hypergraph $H=(V,E)$
with not-too-wild vertex degrees are \emph{almost} always nearly as
good as those of the complete $k$-uniform
hypergraph on $V$.

Recall (\cref{subsec:hypergraphs}) that given $A\subseteq V$,
the number of hits $\rv N_H(A)$ is the number of times a uniformly
random edge in $E$ hits the set $A$.
The first main result of this section   states that if the set $A
\subseteq V$  we are interested in sampling
is chosen uniformly at random among all subsets of $V$ of density $p$,
then almost surely, the distribution of the number of hits
is close to the hypergeometric
distribution $\Hyp(k, pn, n)$, where $n=|V|$.
Formally:


\begin{theorem}[Typical Hypergeometric Convergence]\label{thm:dist}
Let $H = (V, E)$ be a $k$-uniform hypergraph on $n$ vertices in which
every vertex has
degree at most $D$.
Let $p\in\{\frac{0}{n},\frac{1}{n},\dots,\frac{n}{n}\}$ be given
and let $\rv A \subseteq V$ be a uniformly random
subset of cardinality $|\rv A| = p \cdot n$.
Then provided $n\geq \max\{4k^2,\frac{2k}{p(1-p)}\}$, we have
\[
\Pp_{\rv A}\sqbr*{ \dist_{\TV}(\rv N_H(\rv A), \Hyp(k, pn, n)) \ge
\ee } \le \frac{1}{\ee^2}
\cdot \parens*{ \frac{a_{\ref{thm:dist}}(k, p)}{n} + \frac{D-1}{|E|}
\cdot b_{\ref{thm:dist}}(k) }
\]
where the functions $a_{\ref{thm:dist}}(k, p)$ and $b_{\ref{thm:dist}}(k)$ are
defined by:
\[ a_{\ref{thm:dist}}(k, p) = \frac{e}{2} \cdot k \cdot (k+1)^2
\parens*{ \frac{k-1}{p(1-p)} + 4k}\quad\textrm{and}\quad
b_{\ref{thm:dist}}(k) = 0.325 \cdot k \cdot (k+1)^2.
\]
\end{theorem}
We   prove \cref{thm:dist} in
\cref{ss:dist_proof}.

It is well-known that the hypergeometric distribution enjoys a Chernoff bound
(\cref{fac:chernoff}).
Our next result states
that for almost every density-$p$ subset $A\subseteq V$, the random
variable $\rv N_H(A)$
also satisfies a Chernoff bound and is therefore
strongly concentrated around the mean $k p$. Moreover, the fraction
of $A$-s for which this fails is \emph{exponentially} small in $k$ (in addition
to decaying polynomially  in $n=|V|$).


\begin{theorem}[Typical Chernoff Bound]
\label{thm:chernoff}
Let $H = (V, E)$ be a $k$-uniform hypergraph on $n$ vertices in which
every vertex has
degree at most $D$.
Let $p\in\{\frac{0}{n},\frac{1}{n},\dots,\frac{n}{n}\}$ be given and
$\rv A \subseteq V$ be a uniformly random
subset of cardinality $|\rv A| = p \cdot n$.
Suppose $n\geq \max\{4k^2,\frac{2k}{p(1-p)}\}$.
Then for every $\eta>0$, we have\footnote{More accurately the $2
\exp(-2k \eta^2)$ in the LHS can be replaced by any tail estimate for
the binomial/hypergeometric distributions, carrying out analogous
changes in $a_{\ref{thm:chernoff}}$ and $b_{\ref{thm:chernoff}}$.}
\[ \Pp_{\rv A}\sqbr*{ \Pp\sqbr*{ \Abs*{\rv N_H(\rv A) - kp } \ge k
    \cdot \eta} \ge
2\exp(-2 k \eta^2) + \ee} \le \frac{1}{\ee^2} \parens*{
    \frac{a_{\ref{thm:chernoff}}(k, p, \eta)}n +
\frac{D-1}{|E|} \cdot b_{\ref{thm:chernoff}}(k, \eta)},
\]
where
\begin{align*}a_{\ref{thm:chernoff}}(k, p, \eta)
&~ = 16 a_{\ref{thm:dist}}(k,p) \exp(-4k \eta^2)
=~8ek(k+1)^2\parens*{
\frac{k-1}{p(1-p)} + 4k} \exp(-4k\eta^2)\\
b_{\ref{thm:chernoff}}(k, \eta) &~= 8
b_{\ref{thm:dist}} (k)  \exp(-2 k \eta^2)
=~ 2.6 k(k+1)^2 \exp(- 2 k \eta^2)
\end{align*}
and with the functions $a_{\ref{thm:dist}}(k, p)$ and
$b_{\ref{thm:dist}}(k)$ defined as in \cref{thm:dist}.
\end{theorem}
We will prove   \cref{thm:chernoff} in
\cref{sub:proof_of_the_cite_thm_chernoff_}.

The analogous statement to \cref{thm:chernoff}
with $a_{\ref{thm:chernoff}}$ and $b_{\ref{thm:chernoff}}$ replaced
by $a_{\ref{thm:dist}}$ and $b_{\ref{thm:dist}}$
follows immediately follows from \cref{thm:dist}  simply by appealing
to the definition of the \ref{eq:totv} distance.
\cref{thm:chernoff} provides an exponential improvement (in
$k$) over this. Indeed, when $k\eta^2 = \Omega(\log k)$,
the exponential term $\exp(-k\eta^2)$ shrinks much faster than
the growth of the polynomial
prefactors in $k$ occuring in $a_{\ref{thm:dist}},
b_{\ref{thm:dist}}$.


An immediate consequence of the above results and \cref{thm:binhyp}
is the following corollary, which implies parts (i) and (ii)
of \cref{cor:fraction-show}.

\begin{corollary}\label{cor:fraction}
Let $H = (V, E)$ be a $k$-uniform hypergraph on $n$ vertices with
maximum degree at most $D$.
Let $p\in\{\frac{0}{n},\frac{1}{n},\dots,\frac{n}{n}\}$ and  $\alpha
\in (0,1]$ be given
and suppose $n\geq \max\{4k^2,\frac{2k}{p(1-p)}\}$.
Then, writing $\widetilde n=|E|/D$ (e.g., $\widetilde n=\frac{n}{k}$
if $H$ is regular), we have
\[
\dist_{\TV}(\rv N_H(A), \Hyp(k, pn, n)) < \widetilde
n^{-\frac{1-\alpha}2}\qquad\textrm{and}\qquad\dist_{\TV}(\rv N_H(A),
\Bin(k, p)) < \widetilde n^{-\frac{1-\alpha}2} + \frac{k-1}{n-1}
\]
for all but $ (a_{\ref{thm:dist}}(k, p) + b_{\ref{thm:dist}}(k))
\cdot \widetilde n^{-\alpha}$ fraction of the sets $A \subseteq V$
with $|A| = p \cdot n$.

Furthermore, for any $\eta > 0$ and $\alpha'\in [0,1]$, we have
\begin{align*} \Pp\sqbr*{ \Abs*{\rv N_H(A) - kp } > k \eta}
& \le
2\exp(-2 k\eta^2)+\exp\parens*{-\frac{(1-\alpha') k \eta^2}{2}} \cdot
\widetilde n^{-\frac{1-\alpha}2 }
\\
& \leq
(2+\widetilde n^{-\frac{1-\alpha}2 })\exp\parens*{-\frac{(1-\alpha')
k \eta^2}{2}}
\end{align*}
for all but
$ 16(a_{\ref{thm:dist}}(k, p) +
b_{\ref{thm:dist}}(k)) e^{-\alpha'k\eta^2} \widetilde n^{-\alpha}$
fraction of the sets $A \subseteq V$ with $|A| = p \cdot n$.
\end{corollary}

\begin{proof}
For the first statement,
take $\ee =\widetilde n^{-\frac{1-\alpha}{2}}$ in \cref{thm:dist}
and note that $n\geq \frac{n}{k}\geq \frac{|E|}{D}$ by
\cref{PR:degree-to-unif-ratio}.
The second statement follows from the first and \cref{thm:binhyp}.
For the third statement,
take $\ee = \widetilde n^{-\frac{1-\alpha}{2}}
\exp(-\frac{(1-\alpha')k\eta^2}{2})$ in \cref{thm:chernoff}.
\end{proof}

Our proofs of \cref{thm:dist} and \cref{thm:chernoff} rely on the
following lemma and its
following corollaries, which   show that so long as $k$ grows slowly
in $n$, the variance of the fraction of edges which have an
intersection of size exactly $r$ with a uniformly random subset $\rv
A$ of density $p$ is vanishingly small. In other words, we show that
the fraction of edges that hit  $\rv A$ exactly $r$ times  converges
in \emph{quadratic mean} to the hypergeometric probability
$\binom{k}{r} \binom{n-k}{pn -r} {\binom{n}{pn}}^{-1}$.\footnote{For
more information on convergence between probability distributions we
refer to \cite[Ch.~5]{Wasserman13}.}  Formally:

\begin{lemma}[Main Lemma]\label{lem:hypg}
Let $H = (V,  E)$ be a $k$-uniform hypergraph on $n$ vertices  where
every vertex has degree at most $D$.
Let $p\in\{\frac{0}{n},\frac{1}{n},\dots,\frac{n}{n}\}$ be given
and $\rv A \subseteq V$ be a uniformly random subset of cardinality
$|\rv A| = p \cdot n$.
For an integer $r$, we write $  E_r(\rv A)$ for the set of edges in
$H$ that have an intersection of size $r$ with the random set $\rv A$.
Then\footnote{
Here, as usual, a binomial coefficient ${n \choose m}$ is defined to be
$0$ if $m\notin [0,n]$ or $n<0$.
}
\begin{equation}\label{eq:expectation}
\Exp \sqbr*{  \frac{ |E_r(\rv A)| }{|E|}} = \frac{ \binom{k}{r}
    \binom{n-k}{pn-
r}}{\binom{n}{pn}}.
\end{equation}
If moreover $n \ge \max\set*{4k^2,\frac{2r}{p},\frac{2(k-r)}{1-p}}$, then
\begin{equation}\label{eq:tighter}
\Var\sqbr*{  \frac{ |E_r(\rv A)| }{|E|}} \le
\parens*{\frac{C_{\ref{lem:hypg}}(k, r,p)}{n} +
C'_{\ref{lem:hypg}}(k,r) \cdot \frac{D-1}{|E|}},
\end{equation}
where
\[
C_{\ref{lem:hypg}}(k, r, p) = 2e\binom{k}{r}^2 p^{2r} (1-p)^{2(k-r)} k
\parens*{\frac{k-1}{pq} + 4k}\quad\textrm{and}\quad
C'_{\ref{lem:hypg}}(k, r) = p^r (1-p)^{k-r} \cdot 1.3k \cdot \binom{k}{r}
\]
\end{lemma}

We prove \cref{lem:hypg} in \cref{ss:main_lemma_proof}.

By using \cref{lem:hypg} and  appealing to Chebyshev's Inequality
(\cref{fac:pafnuty}), we  show that for \emph{almost} all density-$p$
subsets $A\subseteq V$, the probability
that
$\rv N_H( A)=r$ is very close to the probability that a
hypergeometric random variable $\Hyp(k, pn, n)$ takes the value $r$. Formally:


\begin{corollary}\label{cor:event-bd}
With the  notation and assumptions of \cref{lem:hypg}, for all $\ee>0$, we have
\[ \Pp_{\rv A}\sqbr*{ \Abs*{ \Pp\sqbr*{\rv N_H(\rv A) = r } - \frac{\binom{k}{r}
\binom{n-k}{pn-r}}{\binom{n}{pn}} } \ge \ee}
~\le~\frac{1}{\ee^2} \cdot \parens*{ \frac{ C_{\ref{lem:hypg}}(k,r,
    p) }{n} + \frac{D
- 1}{|E|} \cdot C'_{\ref{lem:hypg}}(k, r)}.
\]
\end{corollary}

\begin{proof}
Let $A \subseteq V$ be a fixed subset. By the definition of $\rv
N_H(A)$, we have
\[
\Pp\sqbr*{ \rv N_H(A) = r } =
\frac{\Abs*{\{e\in E\,:\, |e\cap A|=r\}}}{|E|}=
\frac{|E_r(A)|}{|E|}
\]
Replacing $A$ with the random set $\rv A$
then gives $\Pp_{\rv A}\sqbr*{ \rv N_H(A) = r }=\frac{|E_r(\rv
A)|}{|E|}$. Now, by \cref{lem:hypg} and Chebyshev's Inequality
(\cref{fac:pafnuty}),
%
\begin{align*}
\Pp_{\rv A}\sqbr*{ \Abs*{ \Pp\sqbr*{\rv N_H(\rv A) = r } -
        \frac{\binom{k}{r}
\binom{n-k}{pn-r}}{\binom{n}{pn}} } \ge \ee}
&~=~\Pp_{\rv A}\sqbr*{ \Abs*{ \frac{|E_r(\rv A)|}{|E|} - \Exp \sqbr*{
\frac{ |E_r(\rv A)| }{|E|}} } \ge \ee},\\
&~\le~\frac{\Var\parens*{ \rv E_r(\rv A) }}{\ee^2},\\
&~\le~\frac{1}{\ee^2} \cdot \parens*{ \frac{ C_{\ref{lem:hypg}}(k,r,
    p) }{n} + \frac{D
- 1}{|E|} \cdot C'_{\ref{lem:hypg}}(k, r)}.
\qedhere
\end{align*}
\end{proof}

The next corollary follows by taking $r=k$ and $r=0$ in \cref{cor:event-bd}.
It implies \cref{cor:fraction-show}(iii),
and will also be needed to prove our lower bounds on the worst-case
confinement probability in \cref{sec:lowerbounds}.



\begin{corollary}\label{cor:thm0}
Let $H = (V, E)$ be a $k$-uniform hypergraph on $n $ vertices with
maximum degree $D$.
Let $p\in\{\frac{0}{n},\frac{1}{n},\dots,\frac{n}{n}\}$ and
$\alpha\in [0,1]$ be given,
$\rv A \subseteq V$ be a uniformly random subset of cardinality $|\rv
A| = p \cdot n$, and   $\rv e$ denote a uniformly random  hyperedge in $E $.
Suppose
$n\geq \max\{4k^2,\frac{2k}{p(1-p)}\}$.
Then,
setting $\widetilde n =\frac{|E|}{D}$,
we have
\begin{align*}\Pp_{\rv A}\sqbr*{ \Abs*{ \Pp_{\rv e}\sqbr*{\rv e \cap
    \rv A = \emptyset} - (1-p)^k}
\ge  n^{-\frac{1-\alpha}{2}}}
&~\le~C_{\ref{cor:thm0}} \cdot \frac{(1-p)^k \cdot k^2 \cdot
\widetilde n^{-\alpha}}{p}
\leq \frac{4C_{\ref{cor:thm0}}}{e^2p(\ln(1-p))^2}
\cdot \widetilde n^{-\alpha}
\\
\Pp_{\rv A}\sqbr*{ \Abs*{ \Pp_{\rv e}\sqbr*{\rv e \subseteq \rv A } - p^k}
\ge  n^{-\frac{1-\alpha}{2}}}
&~\le~C_{\ref{cor:thm0}} \cdot \frac{p^k \cdot k^2 \cdot \widetilde
n^{-\alpha}}{1-p}
\leq \frac{4C_{\ref{cor:thm0}}}{e^2(1-p)(\ln p)^2}
\cdot \widetilde n^{-\alpha}
\end{align*}
for some absolute constant $C_{\ref{cor:thm0}} > 0$ independent of
$n$, $k$, $p$, $D$. In fact, we can take $C_{\ref{cor:thm0}}=114$.
\end{corollary}

\begin{proof}
Since $\rv e \subseteq \rv A $ if and only
if $\rv e\cap(V-\rv A)=  \emptyset$, and since $V-\rv A$
distributes uniformly on the $(1-p)$-dense subsets of $V$,
the statement
about $ \Pp_{\rv e}\sqbr*{\rv e \subseteq \rv A }$ follows
from the statement about $\Pp_{\rv e}\sqbr*{\rv e \cap \rv A = \emptyset}$
once $p$ is replaced with $1-p$. It is therefore enough
to prove the statement about $\Pp_{\rv e}\sqbr*{\rv e \cap \rv A =
\emptyset}$.

The event $\set*{\rv e \cap \rv A = \emptyset}$ holds precisely
when $\rv N_H(\rv A) = 0$. By \cref{cor:event-bd}, we have
\begin{align*}
\Pp_{\rv A}\sqbr*{ \Abs*{ \Pp_{\rv e}\sqbr*{\rv e \cap \rv A =
\emptyset} - \frac{ \binom{n-k}{pn} }{\binom{n}{pn}}} \ge \ee}
&
=
\frac{1}{\ee^2} \cdot \parens*{
    \frac{  C_{\ref{lem:hypg}}(k, 0,
    p)}{n}+\frac{D-1}{|E|}C'_{\ref{lem:hypg}}(k,0)
}
\\
&
\le \frac{1}{\ee^2} \cdot \parens*{ 2e \cdot \frac{(1-p)^{2k-1} \cdot
k(5k-1)}{p \cdot n} + (1-p)^k \frac{D-1}{|E| } \cdot  1.3 k} \\
&\leq  \frac{1}{\varepsilon^2}\cdot \frac{(10e+1.3)(1-p)^k
k^2}{p\cdot \tilde{n}},
\end{align*}
where in the last inequality, we used the fact that $n\geq
\frac{|E|}{D}$ shown in proof of \cref{cor:fraction}.
By \cref{thm:binhyp}, noting that $\binom{n-k}{pn}/\binom{n}{pn}$ and
$(1-p)^k$ are the probabilities of producing $k$ hits from the
hypergeometric and binomial distributions respectively, we have:
\[ \Abs*{\frac{\binom{n-k}{pn}}{\binom{n}{pn}} - (1-p)^k} \le
\frac{k-1}{n-1} \le \frac{k}{n} := \delta.\]
Thus,  we get:
\[ \Pp_{\rv A}\sqbr*{ \Abs*{\Pp_{\rv e}\sqbr*{ \rv e \cap \rv A =
\emptyset} - (1-p)^k} \ge \ee} \le \frac{1}{(\ee - \delta)^2} \cdot
\frac{(10e+1.3) \cdot (1-p)^k \cdot k^2}{p\cdot \tilde{n}}.\]
Take $\varepsilon=n^{-\frac{1-\alpha}{2}}$.
Our assumption
$n\geq 4k^2$
implies that
\[
\varepsilon=
n^{-\frac{1-\alpha}{2}}=\frac{n^\frac{1+\alpha}{2}}{n}\geq
\frac{2k}{n}=2\delta,
\]
so $(\varepsilon-\delta)^{-2}\leq (\varepsilon/2)^{-2}=4n^{-(1-\alpha)}$.
We therefore have
\begin{align*}
\Pp_{\rv A}\sqbr*{ \Abs*{\Pp_{\rv e}\sqbr*{ \rv e \cap \rv A =
\emptyset} - (1-p)^k} \ge   n^{-\frac{1-\alpha}{2}}}
& \le
\frac{4}{n^{-(1-\alpha)}} \cdot \frac{(10e+1.3)(1-p)^k \cdot
k^2}{p\cdot \tilde{n}}
\\
& \leq
\widetilde n^{-\alpha} \cdot \frac{4(10e+1.3) (1-p)^k\cdot k^2}{p}.
\end{align*}
We finish by taking $C_{\ref{cor:thm0}}=114\geq 4(10e+1.3)$ and
noting that by elementary analysis, $(1-p)^k\cdot k^2\leq
\frac{4}{e^2(\ln (1-p))^2}$.
\end{proof}

Suppose now that $H = (V, E)$ is a possibly \emph{non-uniform}
hypergraph with average uniformity $k$ that is also an
$(\epsilon,p)$-confiner.
We finally note that \cref{lem:hypg} and a convexity argument imply
the following lower bound on $\epsilon$, whose proof we present in
\cref{ss:conf_lowerbound_main}.
We will give a better lower bound for sparse hypergraphs in
Section~\ref{sec:lowerbounds}.
Contrary to  the other main results of this section, this is not a
result about the typical behavior of the confinement probability of
$p$-dense subsets of $V$,
but rather a result about their worst-case (i.e., largest possible)
confinement probability.

\begin{corollary}[Worst-Case Confinement Probability Lowerbound
{[General]}] \label{cor:conf_lowerbound_main}
Let $H = (V, E)$ be a (possibly non-uniform) hypergraph with average
uniformity $k$
and $n$ vertices.
If $H$
is an \hyperref[eq:eps_confiner]{$(\epsilon, p)$-confiner} and $n\geq
\frac{1}{p(1-p)}$, then
$\epsilon \ge p^k - 2k/n$.

Moreover, when $H = (V, E)$ is  $k$-uniform and $pn$ is integral, we
have the improved bound $\epsilon \ge
\binom{n}{k}^{-1} \binom{n-k}{pn- k} \ge p^k - (k-1)/(n-1)$.
\end{corollary}

We present the proofs of our
results \cref{thm:dist} , \cref{thm:chernoff}, \cref{lem:hypg}, and
\cref{cor:conf_lowerbound_main} in the following subsections.

\subsection{Proof of Typical Hypergeometric Convergence
(\cref{thm:dist})}\label{ss:dist_proof}
\begin{proof}[Proof of \cref{thm:dist}]
Let $A \subseteq V$ be a fixed set of cardinality $p \cdot n$.
Recall that by the definition of the \ref{eq:totv} distance, we have:
\begin{align}
\dist_{\TV}(\rv N_H(A), \Hyp(k, pn, n))
&~=~\frac{1}{2} \cdot \sum_{r = 0}^k \Abs*{ \Pp\sqbr*{\rv N_H(A) = r}
- \Pp\sqbr*{\Hyp(k, pn, n) = r}},\notag\\
&~=~\frac{1}{2} \cdot \sum_{r = 0}^k \Abs*{ \Pp[\rv N_H(A) = r] -
    \frac{\binom{k}{r} \binom{n - k}{pn -
r}}{\binom{n}{pn}}}\label{eq:ftra}
\end{align}
where the second equality is obtained by the definition of the
\ref{eq:hypgd}. We observe that by the pigeon hole principle, if
$\dist_{\TV}(\rv N_H(A), \Hyp(k, pn, n)) \ge \ee$, there should
necessarily exist some $r \in \set*{0, \ldots, k}$ for which we have
\[ \Abs*{ \Pp\sqbr*{ \rv N_H(A) = r} - \frac{\binom{k}{r}
\binom{n-k}{pn -r}}{\binom{n}{pn}}} \ge \frac{2\ee}{k+1}.\]

Randomizing over the choice of $\rv A$ and using the above observation:
\begin{align*}
\Pp_{\rv A}\sqbr*{ \dist_{\TV}( \rv N_H(\rv A), \Hyp(k, pn, n)) \ge \ee }
&~\le~\Pp_{\rv A}\sqbr*{\exists~r~\in~\set*{0, \ldots,
    k}~\tst~\Abs*{\Pp\sqbr*{\rv N_H(\rv A) = r} - \frac{\binom{k}{r}
\binom{n-k}{np - r}}{\binom{n}{pn}}} \ge \frac{2\ee}{k+1}},\\
&~\le~\sum_{r = 0}^k \Pp_{\rv A}\sqbr*{ \Abs*{ \Pp\sqbr*{\rv N_H(\rv A)
    = r} - \frac{\binom{k}{r}\binom{n-k}{np - r}}{\binom{n}{pn}} } \ge
\frac{2\ee}{k+1} },\\
&~\le~\frac{(k+1)^2}{4\ee^2} \sum_{r = 0}^k
\parens*{\frac{C_{\ref{lem:hypg}}(k, r, p)}{n} +
\frac{D - 1}{|E|} \cdot C'_{\ref{lem:hypg}}(k, r)}.
\end{align*}
where the second inequality follows by a simple union bound and the
third by \cref{cor:event-bd}.

We observe,
\[
(k+1)^2\sum_{r = 0}^k C'_{\ref{lem:hypg}}(k, r) = 1.3k(k+1)^2 \sum_{r
= 0}^k p^r(1-p)^{k-r}\binom{k}{r} = 1.3k (k+1)^2 = 4 b_{\ref{thm:dist}}(k)
\]
by the binomial formula.
Further, we have
\[
\textstyle{
    \sum_{r = 0}^k C_{\ref{lem:hypg}}(k, r, p) = 2e \cdot k \cdot
    \parens*{\frac{k-1}{p(1-p)} +
    4k} \cdot \sum_{r = 0}^k \binom{k}{r}^2 p^{2r} (1-p)^{2(k-r)} \le
    2e \cdot k\parens*{ \frac{k-1}{p(1-p)} + 4k} = \frac{ 4
    a_{\ref{thm:dist}}(k, p)}{(k+1)^2},
}
\]
where we used that the summands are squared probabilities in
upperbounding the sum by 1.

Putting everything together, we have
\[
\Pp\sqbr*{ \dist_{\TV}(\rv N_H(\rv A), \Hyp(k, pn, n)) \ge \ee } \le
\frac{1}{\ee^2} \cdot \parens*{ \frac{a_{\ref{thm:dist}}(k, p)}n +
\frac{D-1}{|E|} \cdot b_{\ref{thm:dist}}(k)},
\]
as we wanted to prove.
\end{proof}
\subsection{Proof of the Typical Chernoff Bound
(\cref{thm:chernoff})}\label{sub:proof_of_the_cite_thm_chernoff_} 

\begin{proof}[Proof of \cref{thm:chernoff}]
Let $A \subseteq V$ be fixed and let $\rv Z \sim \Hyp(k, pn, n)$ be a
hypergeometric random variable. Notice that for every $\eta>0$,
\begin{align}
\Pp\sqbr*{ \Abs*{\rv N_H(A) - kp } \ge k \cdot \eta}
&~=~\sum_{j \in [0, k],\atop |j - kp| > k \eta} \Pp\sqbr*{\rv N_H(A) =
j}\notag\\
&~\le~\sum_{j \in [0, k], \atop |j- kp| > k \eta} \parens*{\Pp\sqbr*{
    \rv Z = j} +
\Abs*{ \Pp\sqbr*{\rv N_H(A) = j} - \Pp\sqbr*{\rv Z = j} }}\notag\\
&~=~\sum_{j \in [0, k],\atop |j - kp| > k \eta} \Pp\sqbr*{\rv Z = j} +
\sum_{j \in [0, k], \atop |j - kp| > k \eta}\Abs*{ \Pp\sqbr*{\rv
N_H(A) = j} - \Pp\sqbr*{\rv Z = j} }\notag\\
&~=~\Pp\sqbr*{ \Abs*{\rv Z - kp } > k \eta} + \sum_{j \in [0, k],\atop |j -
kp| > k \eta} \Abs*{ \Pp\sqbr*{\rv N_H(A) = j} - \Pp\sqbr*{\rv Z = j}
}\notag\\
&~\le~2\exp(-2 k \eta^2)  + \sum_{j \in [0, k], \atop |j - kp| > k \eta}
\Abs*{\Pp\sqbr*{\rv N_H(A) = j} - \Pp\sqbr*{\rv Z = j}},\label{eq:yestagh}
\end{align}
where we have appealed to the Chernoff bound
for hypergeometric random
variables   (\cref{fac:chernoff}) in the final
inequality. As a result, if we have
\[ \Pp\sqbr*{ \Abs*{\rv N_H(A) - kp} > k \cdot \eta} \ge 2\exp(-2 k \eta^2) +
\ee,\]
then by the pidgeon-hole principle, at least one of the terms
occuring in the sum
in \cref{eq:yestagh} should be greater than  $\frac{\ee}{k+1}$.
Thus, writing $J = \set*{j \in [0, k] : |j - kp| > k \eta}$ and
letting $\rv A\subseteq V$
be chosen at random as in the theorem,   we have
\begin{align}
\Pp_{\rv A}\sqbr*{ \Pp\sqbr*{ \Abs*{\rv N_H(\rv A) - kp } \ge k
    \cdot \eta} \ge
2\exp(-2 k \eta^2) + \ee}
&~\le~\Pp\sqbr*{ \exists~j~\in J~\tst~\Abs*{\Pp\sqbr*{\rv N_H(\rv A)
    = j} - \Pp\sqbr*{\rv Z = j}} \ge
\frac{\ee}{k+1} }\notag\\
&~=~\sum_{j \in [0, k], \atop |j - kp| > k \eta}
\Pp_{\rv A}\sqbr*{ \Abs*{\Pp\sqbr*{ \rv N_H(\rv A) = j} - \Pp\sqbr*{\rv Z =
j} } \ge \frac{\ee}{k+1} } \notag\\
&~\le~\frac{(k+1)^2}{\ee^2 } \sum_{j \in [0, k], \atop |j - kp| > k
\eta}\parens*{
    \frac{C_{\ref{lem:hypg}}(k, j, p)}n +
\frac{D-1}{|E|} C'_{\ref{lem:hypg}}(k, j) }
\label{eq:yestaghh}
\end{align}
where we have appealed to \cref{cor:event-bd} to obtain the
final inequality \cref{eq:yestaghh}. Now, we bound the sums
occuring in \cref{eq:yestaghh}
individually. We have
\begin{align*}
\sum_{j \in [0, k],\atop |j - kp| > k \eta}
C_{\ref{lem:hypg}}(k, j, p)
&~=~ 2ek\parens*{\frac{k-1}{p(1-p)} + 4k}\sum_{j \in [k+1], \atop |j
- kp| > k \eta} \binom{k}{j}^2 p^{2j}
(1-p)^{2(k-j)},\\
&~\le~2ek\parens*{\frac{k-1}{p(1-p)} + 4k}\parens*{ \sum_{j \in
        [0, k],\atop |j
- kp| > k \eta} \binom{k}{j} p^j (1-p)^{k-j} }^2 .
\end{align*}
Writing $\rv B \sim \Bin(k, p)$ for a binomially distributed random
variable we see that
\[ \sum_{j \in [0, k],\atop |j - kp| > k\eta} \binom{k}{j} p^j (1-p)^{k-j} =
\Pp\sqbr*{ |\rv B - kp| > k \eta} \le 2\exp(-2 k \eta^2) , \]
where we have appealed to the Chernoff bound for binomial random variables in
the inequality (\cref{fac:chernoff}).
Thus,
\begin{equation}
\label{eq:cher_first_term}
\sum_{j \in [0, k],\atop |j - kp| > k \eta}
C_{\ref{lem:hypg}}(k, j, p)
\le 2ek\parens*{\frac{k-1}{p(1-p)} + 4k} \cdot 4\exp(-4k \eta^2)
\end{equation}
Along similar lines,
\begin{align}
\sum_{j \in [0, k], \atop |j - kp| > k \eta}
C'_{\ref{lem:hypg}}(k, j)
&~=~1.3k \sum_{j \in [0, k], \atop |j - kp| > k \eta} p^j (1-p)^{k-
j}\binom{k}{r}\notag\\
&~=~1.3k \Pp\sqbr*{ |\rv B - kp | > k \eta}\notag\\
&~\leq ~2.6k \exp(-2 k \eta^2).\label{eq:cher_second_term}
\end{align}
Combining the bounds from \cref{eq:cher_first_term}
and \cref{eq:cher_second_term} with \cref{eq:yestaghh}
we obtain
{
\begin{align*}\Pp_{\rv A}\big[ \Pp\sqbr*{ \Abs*{\rv N_H(\rv A) - kp }
        \ge k \cdot \eta} & \ge
    2\exp(-2 k \eta^2) + \ee\big]
    \\
    &~\le~\frac{(k+1)^2}{\ee^2}\parens*{ \frac{8ek}{n}
        \parens*{\frac{k-1}{p(1-p)} + 4k}
    \exp(-4k\eta^2) + \frac{D-1}{|E|} 2.6 k \exp(- 2 k\eta^2)}\\
    &~=~\frac{1}{\ee^2}
    \parens*{\frac{a_{\ref{thm:chernoff}}(k, p, \eta )}n +
    \frac{D-1}{|E|} \cdot b_{\ref{thm:chernoff}}(k, \eta) },
\end{align*}
}
which concludes our proof.
\end{proof}

\subsection{Proof of the Main \cref{lem:hypg}}\label{ss:main_lemma_proof}
\begin{proof}[Proof of \cref{lem:hypg}]
Suppose $\rv A\subseteq V$ is chosen at random as in the lemma.
For all $e \in E$, define
\[\rv X_r(e, \rv A) =
\begin{cases}
    1 & |e \cap \rv A| = r,\\
    0 & \textrm{otherwise.}
\end{cases}\]
Then $|E_r(\rv A)|=\sum_{e\in E}\rv X_r(e,\rv A)$.

Now, observe that for any hyperdege $e \in E$, the size $|e \cap \rv
A|$ of the intersection distributes according
to the hypergeometric distribution $\Hyp(k, pn, n)$. Thus, provided
\[
\Exp \rv X_r(e, \rv A) = \Pp_{\rv A}\sqbr*{ \Abs*{e \cap \rv A }
= r } = \frac{
\binom{k}{r} \cdot \binom{n - k}{pn - r}}{\binom{n}{pn}}.
\]
Since $\frac{| E_r(\rv A) |}{|E|}= \frac{1}{|E|} \cdot \sum_{e \in E}
\rv X_{r}(e, \rv
A)$ this establishes \cref{eq:expectation} by linearity of the expectation.

In order to prove \cref{eq:tighter}, we will first assume the following bound:
\begin{claim}\label{cl:tent}
We have
\[\Var\parens*{\frac{| E_r(\rv A) |}{|E|} } \le \binom{k}{r}^2
    \underbrace{\parens*{\frac{\binom{n-2k}{pn -
            2r}}{\binom{n}{pn}} - \frac{\binom{n- k}{pn
    -r}^2}{\binom{n}{pn}^2 }}}_{\mathsf{term}_1} + \,\frac{k
    \cdot (D - 1)}{|E|}
    \cdot \binom{k}{r}
\underbrace{\frac{\binom{n-k}{pn -r}}{\binom{n}{pn}}}_{\mathsf{term}_2}.\]
\end{claim}
The inequality stated in \cref{eq:tighter}, will now follow from a careful
analysis of   $\mathsf{term}_1$ and $\mathsf{term}_2$ occuring in
\cref{cl:tent}. We proceed with this analysis and postpone   the
proof of \cref{cl:tent} to the end of this subsection.
Suppose $n\geq\max\{4k^2,\frac{2r}{p},\frac{2(k-r)}{1-p}\}$.
Writing $q=1-p$ and $s=k-r$ and note that our assumption  on
$n$ implies $pn\geq 2r$ and $qn\geq 2s$. Thus, we have
\begin{align*}
\frac{\binom{n-k}{pn-r}}{\binom{n}{pn}}
&~=~\frac{(n-k)! \cdot (pn)! \cdot (qn)!}{n! \cdot (pn - r)!(qn -  s)!}\\
&~=~\frac{\prod_{i =0}^{r-1} (pn -i) \cdot \prod_{j = 0}^{s-1} (qn -
j)}{\prod_{\ell = 0}^{k-1} (n-\ell)}\\
&~=~p^rq^s \cdot \frac{\prod_{i = 0}^{r-1} (n - i/p) \cdot \prod_{j =
0}^{s-1} (n - j/q)}{
\prod_{\ell = 0}^{k-1} (n-\ell)}
\\
&~=~p^rq^s \cdot \prod_{i = 0}^{r-1} \frac{n - i/p}{n-i}
\cdot \prod_{j = 0}^{s-1} \frac{n - j/q}{n-(r+j)}.
\end{align*}
In this expression, the factors $ \frac{n - i/p}{n-i}$ and  $\frac{n
- j/q}{n-(r+j)}$
lie in the interval $\sqbr*{ 1 - \frac{c(p, r ,s)}{n},
1-\frac{k-1}{n-(k-1)}}$, where
\[ c(p,r,s) = \max\set*{\frac{r-1}{p}, \frac{s-1}{q}} \le \frac{k-1}{p q}.\]
Therefore,
\begin{equation}\label{eq:sandwich}
p^r q^s \cdot\parens*{1 - \frac{c(p, r, s)}{n}}^{k}  \le
\frac{\binom{n-k}{pn-r}}{\binom{n}{pn}} \le  p^r q^s \cdot \parens*{1
+ \frac{k-1}{n-(k-1)}}^{k}
\end{equation}
%
%
By an analogous argument
\begin{equation}\label{eq:sandwich2}
p^{2r} q^{2s} \cdot \parens*{1 - \frac{c(p, 2r, 2s)}{n}}^{2k} \le
\frac{\binom{n- 2k}{pn - 2r}}{\binom{n}{pn}} \le p^{2r} q^{2s} \cdot
\parens*{1+\frac{2k-1}{n - (2k - 1)}}^{2k}.
\end{equation}

Using  \cref{eq:sandwich} and \cref{eq:sandwich2} together gives
\[
\mathsf{term}_1=
\frac{\binom{n- 2k}{pn - 2r}}{\binom{n}{pn}} -
\frac{\binom{n-k}{pn-r}^2}{\binom{n}{pn}^2}
\le
p^{2r}q^{2s} \cdot \parens*{\parens*{ 1 + \frac{2k-1}{n-(2k-1)}}^{2k}
- \parens*{1 - \frac{c(p, r,s)}{n}}^{2k}},
\]
and \cref{LM:power-dist} implies
\[
{\textstyle
    \parens*{\parens*{ 1 + \frac{2k-1}{n-(2k-1)}}^{2k} - \parens*{1 -
    \frac{c(p, r,s)}{n}}^{2k}}
    \le 2k \parens*{ \frac{2k-1}{n - (2k-1)} + \frac{c(p, r,s)}n } \cdot
    \parens*{ 1 + \frac{2k-1}{n - (2k-1)}}^{2k-1}.
}
\]
Since
\[ \parens*{ 1 + \frac{2k - 1}{n- (2k-1)}}^{2k-1} \le
\parens*{\exp\parens*{ \frac{2k - 1}{n- (2k-1)}}}^{2k-1} =
\exp\parens*{\frac{(2k-1)^2}{n - (2k-1)}},\]
we obtain
\[
{\textstyle
    \parens*{\parens*{ 1 + \frac{2k-1}{n-(2k-1)}}^{2k} - \parens*{1 -
    \frac{c(p, r,s)}{n}}^{2k}}
    \le \parens*{ \frac{2k-1}{n- (2k-1)} + \frac{c(p, r, s)}{n}} \cdot
    2k \cdot \exp\parens*{\frac{(2k-1)^2}{n - (2k-1)}}
}\]
Thus, we have
\[
\mathsf{term}_1  = \frac{\binom{n-2k}{pn - 2r}}{\binom{n}{pn}} -
\frac{ \binom{n-k}{pn - r}^2 }{\binom{n}{pn}^2} \le p^{2r}q^{2s}
\cdot \parens*{ \frac{2k-1}{n -(2k-1)}+ \frac{c(p, r, s)}n} \cdot 2k
\cdot \exp\parens*{ \frac{(2k-1)^2}{n -(2k-1)}}.
\]
Since we assume $n \ge 2(2k-1)$,  we can further write
\[
\mathsf{term}_1 \le p^{2r}q^{2s} \cdot \frac{1}{n} \parens*{ c(r,p,s)
+ 2(2k-1)} \cdot 2k \cdot \exp\parens*{ \frac{(2k-1)^2}{n -(2k-1)}}.
\]

We can also use \cref{eq:sandwich} to obtain,
\[
\mathsf{term}_2 \le \frac{\binom{n-k}{pn - r}}{\binom{n}{pn}} \le
p^r q^s \cdot \exp\parens*{ \frac{k(k-1)}{n - (k-1)}}.
\]
Plugging the preceding into \cref{cl:tent}, we obtain:
{
\[
    \Var(\rv E_r(\rv A)) \le \frac{1}{n} \cdot \parens*{ p^{2r} q^{2s}
        \cdot (c(r, p, s) + 2(2k-1)) \cdot 2k \cdot
    \exp\parens*{\frac{(2k-1)^2}{n - (2k-1)}}}
    + \frac{(D-1)}{|E|} p^r q^s \cdot \parens*{k \cdot \binom{k}{r}
    \cdot \exp\parens*{ \frac{k(k-1)}{n - (k-1)}}}.
\]
}
By our assumption that $n \ge 4 k^2$,
we have $\frac{(2k-1)^2}{n - (2k-1)}\leq 1$ and
$\frac{k(k-1)}{n - (k-1)}\leq \frac{1}{4}$, so
the last inequality implies:
\[
\Var(\rv E_r(\rv A)) \le \frac{1}{n} \cdot \parens*{ 2e \cdot
    \binom{k}{r}^2 p^{2r} q^{2s}
\cdot k \cdot (c(r, p, s) + 4k)} + \frac{(D-1)}{|E|} \cdot p^r q^s \parens*{
\sqrt[4]{e} k \cdot \binom{k}{r}}.
\]
Finally, using that $c(r, p, s) \le (k-1)/pq$ and $\sqrt[4]{e}<1.3$, we get
\[
\Var(\rv E_r(\rv A)) \le \frac{1}{n} \cdot \parens*{ 2e \cdot
    \binom{k}{r}^2 p^{2r}
q^{2s} \cdot k \cdot \parens*{\frac{k-1}{pq} + 4k}} + \frac{(D-1)}{|E|}
\cdot p^r q^s \parens*{ 1.3\cdot k \cdot \binom{k}{r} },
\]
which proves assertion \cref{eq:tighter} of  lemma.

\medskip

\emph{Proof of \cref{cl:tent}.}
Write $\rv Y= \frac{| E_r(\rv A)|}{|E|}$.
Since $\rv Y = \frac{1}{|E|} \cdot \sum_{e\in E} \rv X_r(e, \rv A)$, we have
\begin{align*}
\Var\parens*{\rv Y}
&~=~\Exp[\rv Y^2] - \Exp[\rv Y]^2\\
&~=~\frac{1}{|E|^2}\parens*{\sum_{e, f \in E} \Exp[\rv X_r(e, \rv A) \rv
X_r(f, \rv A)]}  - \Exp[\rv Y]^2\\
&~=~\frac{1}{|E|^2}\sum_{e, f: e \cap f = \emptyset} \Exp[\rv X_r(e, \rv
A) \rv X_r(f, \rv A)] + \frac{1}{|E|^2}\sum_{e, f: e \cap f \ne
\emptyset} \Exp[\rv X_r(e, \rv
A) \rv X_r(f, \rv A)]] -
\Exp[\rv Y]^2.
\end{align*}
Now notice: Since $\rv X_r(e, \rv A)$ are Boolean random variables, we
have $\rv X_r(e, \rv A) \rv X_r(f, \rv A) \le \rv X_r(e, \rv A)$. We
use this observation to bound the second sum. As for each $e \in E$,
there are at most $k (D-1)$ hyperedges $f \in E$ that intersect $e$,
we proceed as follows:
\begin{align*}
\Var(\rv E_r(\rv A))
&~\le~\frac{1}{|E|^2} \sum_{e, f: e \cap f = \emptyset} \Exp[\rv X_r(e, \rv
A) \rv X_r(f, \rv A)] + \frac{k (D-1)}{|E|^2} \sum_{e \in E}\Exp[\rv
X_r(e, \rv A)] - \Exp[\rv Y]^2
\\
&~=~\frac{1}{|E|^2} \sum_{e, f: e \cap f = \emptyset} \Exp[\rv X_r(e, \rv
A) \rv X_r(f, \rv A)]
+ \frac{ k(D-1)}{|E| } \Exp[\rv Y]  - \Exp[\rv Y]^2
\\
&~=~\frac{1}{|E|^2} \sum_{e, f: e \cap f = \emptyset} \Exp[\rv X_r(e, \rv
A) \rv X_r(f, \rv A)]
+ \frac{ k(D-1)}{|E| } \frac{\binom{k}{r} \binom{n - k}{pn -
r}}{\binom{n}{pn}} - \frac{ \binom{k}{r}^2
\binom{n-k}{pn-r}^2}{\binom{n}{pn}^2},
\end{align*}
where we have appealed to \cref{eq:expectation} for the second equality.

To finish, note that for $e \cap f = \emptyset$, by appealing to analogous
arguments to those  we used in establishing \cref{eq:expectation}, we get
\[ \Exp[\rv X_r(e, \rv A) \rv X_r(f, \rv A)] = \frac{ \binom{k}{r}
\binom{k}{r} \binom{n-2k}{pn - 2r}}{\binom{n}{pn}}. \]
Since the the number of pairs $(e,f)\in E\times E$ with $e\cap f=\emptyset$
is at most $|E|^2$, we have
\[\frac{1}{|E|^2} \sum_{e, f: e \cap f = \emptyset} \Exp[\rv X_r(e, \rv
A) \rv X_r(f, \rv A)]\leq \frac{\binom{k}{r}^2 \binom{n-2k}{pn -
2r}}{\binom{n}{pn}}
\]
and \cref{cl:tent} follows.
This completes the proof  of \cref{lem:hypg}.
\end{proof}
\subsection{Proof of the Worst-Case Confinement Probability
Lowerbound (\cref{cor:conf_lowerbound_main})}\label{ss:conf_lowerbound_main}
\begin{proof}[Proof of \cref{cor:conf_lowerbound_main}]
Suppose $pn$ is an integer. If $H= (V, E)$ were a $k$-uniform
hypergraph, then the statement would follow from
\cref{lem:hypg}: Let $A \subseteq V$ be a fixed set
of density $p$. The probability of $\rv e$ being
confined to the set $A$, i.e.,~$\rv e \subseteq A$, is given by the
fraction of hyperedges having all $k$ elements from $A$,
i.e.,~$\Pp_{\rv e}[\rv e \subseteq A] = \frac{|E_k(A)|}{|E|}$. Further, note
that picking $\rv A$ uniformly random over all density-$p$ sets, we have
\begin{equation}\label{eq:average_integral}\Exp_{\rv A} \Pp_{\rv e
\sim E}\sqbr*{\rv e \subseteq A} = \Exp_{\rv A}\sqbr*{\frac{ |E_k(\rv
A)|}{|E|}} = \frac{\binom{n-k}{pn-k}}{\binom{n}{pn}} \ge p^k -
\frac{k-1}{n-1}
\end{equation}
where the equality is due to \cref{eq:expectation} and
the inequality due to \cref{thm:binhyp} (here we need $n\geq (p (1-p))^{-1}$).
Thus, there exists at least one subset $A \subseteq E$ of density $p$ such
that $\Pp_{\rv e}\sqbr*{\rv e \subseteq A} \ge p^k - (k-1)/(n-1)$.

For the case where $H$ is non-uniform, we decompose the edge set into
disjoint sets $E
= E_0 \sqcup  E_1 \sqcup \cdots \sqcup E_K$, where $K$ is the maximum
uniformity of
$H$. Then $H_j := (V, E_j)$ is a $j$-uniform graph for all $j =0,
\ldots, K$. Writing $f_j = |E_j|/|E|$ for the fraction of $j$-uniform edges
in $H$, we first observe that our average uniformity assumption implies
\begin{equation}\label{eq:average_uniformity_here}
\sum_{j = 0}^K f_j \cdot j=  k.
\end{equation}
Then, by the law of total expectation, we have
\begin{align*}
\Exp_{\rv A} \Pp_{\rv e \sim E}\sqbr*{\rv e \subseteq \rv A}
&~=~\Exp_{\rv A}\sqbr*{\sum_{j = 0}^K \Pp\sqbr*{\rv e \in E_j} \cdot \Pp_{\rv e
\sim E_j}\sqbr*{\rv e \subseteq \rv A}}\\
&~=~\sum_{j = 0}^K f_j \Exp_{\rv A}\sqbr*{ \Pp_{\rv e \sim E_j}\sqbr*{\rv e
\subseteq \rv A}}\\
&~\ge~\sum_{j=  0}^K f_j \cdot \parens*{p^j -
\frac{j-1}{n-1}} \\
&~=~\parens*{\sum_{j = 0}^K f_j \cdot  p^j } - \frac{k-1}{n-1},
\end{align*}
where we have used \cref{eq:average_integral} for the inequality
and \cref{eq:average_uniformity_here} for the last equality. Now, the
convextiy of $x \mapsto p^x$ and \cref{eq:average_uniformity_here}
yield
\[ \sum_{j = 0}^K f_j p^j \ge p^k\quad\textrm{and consequently}\quad
\Exp_{\rv A} \Pp_{\rv e \sim E}\sqbr*{\rv e \subseteq \rv A} \ge p^k
- \frac{k-1}{n-1}.\]
Thus, there still exists at least one $A \subseteq V$ of density $p$
satisfying $\Pp_{\rv e}\sqbr*{\rv e \subseteq A } \ge p^k -  (k-1)/(n-1)$,
which concludes the proof when $pn$ is integral.

When $pn$ is not integral, let $p' \in (0, 1)$ be such that $p'n =
\lfloor pn \rfloor$.
Then, the argument above shows that there exists a $p$-dense set $A
\subseteq V$ such
that
$\Pp_{\rv e}\sqbr*{\rv e \subseteq A } \ge (p')^k - (k-1)/(n-1)$. By
\cref{LM:power-dist}, we have that $(p')^k \ge p^k - k/n$
since $|p - p'| \le 1/n$. Thus, we have
\[ \Pp_{\rv e}\sqbr*{\rv e \subseteq A} \ge p^k - \frac{k-1}{n-1} -
\frac{k}{n} \ge p^k - \frac{2k}{n}\]
where for the final inequality we have used $(k-1)/(n-1) \le k/n$.
\end{proof}

\section{Performance of Hypergraph Confiners in the Worst
Case}\label{sec:lowerbounds}

We have already seen in \cref{cor:conf_lowerbound_main} that if we have a
family of \hyperref[eq:eps_confiner]{$(\epsilon,p)$-confiners}
each having  average uniformity at most $k$,
then we must have $\epsilon \geq p^k$. In this section,
we present better lower bounds, which hold
when the family  is sparse,
i.e.,~each hypergraph $H=(V,E)$ in the
family satisfies $\frac{|E|}{|V|}\leq r$ for some fixed $r$.
These bounds are summarized in \cref{CR:all_lower_bounds}.
Our first and main result of this kind says that
$\epsilon\geq p^k+\Omega_{p,k} (\frac{1}{r})$.
Formally:


\begin{theorem}\label{TH:lower-bound}
There
are   constants $a,b,u,v,w,C>0$ such that the following hold.
Let $k>1$, $r>0$ and $\delta\in (0,1)$ be real numbers.
Let $H=(V,E)$ be a hypergraph with $n$ vertices, $rn$ edges and
average uniformity $k$.
If
\[
n\geq C\,k^a r^b \delta^{-u}(1-\delta)^{-v} \min\{1,k-1\}^{-w}
=r^b\cdot
\poly(k,\textstyle{\frac{1}{k-1},\frac{1}{\delta},\frac{1}{1-\delta}}),
\]
then there is $A\subseteq V$ of size $\lfloor \delta n\rfloor$ such that
\[
\frac{|E(A)|}{|E|}\geq
\min\left\{\delta^k+
\frac{c_{\ref{TH:lower-bound}}(k,\delta)}{r },
1\right\},
\]
where
\begin{align*}
c_{\ref{TH:lower-bound}}(k,\delta)&=\textstyle{\frac{1}{30}\delta(1-\delta)^2\min\{\frac{1}{k},\frac{k-1}{2}\}}.
\end{align*}
Moreover, setting $t=7+\sqrt{33}$, we can take
$b=w=t\leq 12.75$,
$a=u=5t\leq 63.73$
and
$v=2t\leq 25.49$.
\end{theorem}


As a consequence, we get the following theorem which extends
\cref{CR:eps_view_show}
from the introduction.


\begin{theorem}\label{TH:lower-bound-epsilon-view}
Let $a,b,u,v,w,C$ be   as in \cref{TH:lower-bound}
(e.g., we can take $b=7+\sqrt{33}\leq 12.75$),
and let $k>1$, $\delta\in (0,1)$ and $\ee\in (0,1-\delta^k)$
be real numbers.
Let $H=(V,E)$ be a hypergraph with $n$ vertices
and average uniformity $k$
that is also a  \hyperref[eq:eps_confiner]{$(\delta^k+\ee,\delta)$-confiner}.
If $k\geq 2$, then
\[\frac{|E|}{|V|}\geq \min\left\{
\frac{\delta(1-\delta)^2}{30k}\cdot \ee^{-1}
\,,\,
\parens*{\frac{n}{Ck^a\delta^{-u}(1-\delta)^{-v}}}^{\frac{1}{b}}
\right\}
=\min\left\{\Omega_{k,\delta}(\epsilon^{-1})\,,\,\Omega_{k,\delta}(n^{\frac{1}{b}})\right\},\]
and if $1<k<2$, then
\[\frac{|E|}{|V|}\geq \min\left\{
\frac{\delta(1-\delta)^2(k-1)}{60}\cdot \ee^{-1}
\,,\,
\parens*{\frac{n}{Ck^a\delta^{-u}(1-\delta)^{-v}(k-1)^{-w}}}^{\frac{1}{b}}
\right\}
=\min\left\{\Omega_{k,\delta}(\epsilon^{-1})\,,\,\Omega_{k,\delta}(n^{\frac{1}{b}})\right\}.
\]
%
\end{theorem}

\begin{proof}[Proof of \cref{TH:lower-bound-epsilon-view} using
\cref{TH:lower-bound}]
Let $r=\frac{|E|}{|V|}$
be the sparsity of $H$.
If $n\leq Ck^a r^b \delta^{-u} (1-\delta)^{-v}\min\{1,k-1\}^{-w}$,
then $r\geq (\frac{n}{C k^a  \delta^{-u}
(1-\delta)^{-v}\min\{1,k-1\}^{-w}})^{1/b}$, and we are clearly done.
We may therefore assume that
$n\geq Ck^a r^b \delta^{-u} (1-\delta)^{-v}\min\{1,k-1\}^{-w}$.
Now, by \cref{TH:lower-bound}, there is $A\subseteq V$ of density $\delta$
such that
\[
\frac{|E(A)|}{|E|}\geq
\min\left\{\delta^k+
\frac{c_{\ref{TH:lower-bound}}(k,\delta)}{r },
1\right\}.
\]
On the other hand, since $H$ is a
\hyperref[eq:eps_confiner]{$(\delta^k+\ee,\delta)$-confiner},
\[
\frac{|E(A)|}{|E|}\leq \delta^k+\ee<  1 .
\]
Together, this implies that
$\delta^k+\frac{c_{\ref{TH:lower-bound}}(k,\delta)}{r }\leq
\delta^k+\ee$, and rearranging gives $r\geq
c_{\ref{TH:lower-bound}}(k,\delta)\ee^{-1}$.
\end{proof}

\cref{TH:lower-bound} is   a consequence of an even stronger
lower bound that applies when every vertex of $H$ is included in some edge.
To state it, note first that the extra assumption on $H$ implies that
its average degree,
which is $rk$
by \cref{PR:degree-to-unif-ratio}, is at least $1$, and hence
$r\geq \frac{1}{k}$.
With this in mind, for every real $k\geq 1$ and $r\geq \frac{1}{k}$,
define a function
%
\begin{equation}\label{EQ:f-k-r-delta-dfn}
f_{k,r}(x)=
1-(1-x)^{\frac{1}{rk}}
+(1-x)^{\frac{1}{rk}}
(1-(1-x)^{\frac{rk-1}{rk}})^k
\end{equation}
We will show in
\cref{subsec:properties-of-f-k-r-delta} that the functions
$f_{r,k}(x)$ have the following properties:
\begin{itemize}
\item $f_{k,r}(x)\geq f_{k,r'}(x)$ for all $r\leq r'$ and $x\in
[0,1]$ (\cref{LM:f-k-r-decrease-in-r}).
\item $\lim_{r\to\infty} f_{k,r}(x)=x^k$ (elementary calculus).
\item $f_{k,r}(x)\geq
x^k+\min\{\frac{1}{k},\frac{k-1}{2}\}\cdot\frac{x(1-x)^2}{r}$
(\cref{TH:distance-to-delta-to-k}).
\item $f'_{k,r}(x)\in [\frac{1}{rk},k]$ for all $x\in [0,1]$
(\cref{PR:derivative}).
In particular, $f_{k,r}(x)$ is strictly increasing.
\end{itemize}
It is also straightforward to see that $f_{1,r}(x)=x$.
The graph of $f_{k,r}(x)$ for some values of $k,r$ is compared to
that of $y=x^k$ in \cref{FG:f-k-r} in the introduction.

Our next theorem states that
if every vertex of an
$(\epsilon,\delta)$-confiner is included in some edge,
then we must have $\epsilon \geq f_{k, 3r}(\delta)-o(1)$, where
$k$ and $r$ are the sparsity and average uniformity of $H$, respectively.
Moreover, if $H$ is regular, or has minimal vertex degree $3$,
then the better lower bound $\epsilon \geq f_{k, r}(\delta)-o(1)$
holds.

\begin{theorem}\label{TH:lower-bound-strong}
For every $\sigma\in (0,\frac{1}{6})$, there
are positive constants $a,b,u,v,C>0$
such that the following hold:
Let   $k\geq 1$, $r>0$ and $\delta\in (0,1)$ be real numbers.
Let $H=(V,E)$ be a hypergraph with $n$ vertices, $r\cdot n$ edges and
average uniformity $k$.
Suppose further that every vertex of $H$ is included in some edge.
\begin{enumerate}[(i)]
\item If $n\geq Ck^a r^b \delta^{-u} (1-\delta)^{-v} $,  then there is
$A\subseteq V$ with $|A|=\lfloor \delta n \rfloor$ such that
\[\frac{|E(A)|}{|E|}\geq f_{k,3r}(\delta)-63200\,
    k^4\delta^{-3} n^{-\sigma}
=f_{k,3r}(\delta)-\poly(k,\delta^{-1})\cdot O(n^{-\sigma}).\]
\item If, in addition to the assumption in (i), $H$ is regular, or
$\deg(v)\geq 3$ for all $v\in V$, then there is
$A\subseteq V$ with $|A|=\lfloor \delta n \rfloor$ such that
\[\frac{|E(A)|}{|E|}\geq f_{k,r}(\delta)-63200\,
    k^4\delta^{-3} n^{-\sigma}
=f_{k,r}(\delta)-\poly(k,\delta^{-1})\cdot O(n^{-\sigma}).\]
\end{enumerate}
Moreover, for   (ii), we can take
\[
a=12,\quad b=\textstyle{\frac{24}{1-6\sigma}},\quad u=v=6,\quad
C=3\cdot 10^6,
\]
or we can take
$\sigma=\frac{7-\sqrt{33}}{16}=0.07846\dots$ and choose
\[
a=b=u=\sigma^{-1}=7+\sqrt{33}=12.7445\dots,
\quad
v=\textstyle{\frac{\sigma^{-1}}{2}}=6.3722\dots,
\quad
C=7.6\cdot 10^6.
\]
For (i), one can use the same $a,b,u,v$ from (ii) if one replaces
$C$ by $3^b C$.
\end{theorem}

\begin{proof}[Proof of \cref{TH:lower-bound} assuming
\cref{TH:lower-bound-strong}.]
Fix some $\sigma\in (0,\frac{1}{6})$
and let $a',b',u',v',C'$ be the constants $a,b,u,v,C$ promised by
\cref{TH:lower-bound-strong};
we will specify $\sigma$ and these constants later.
We write $c(k)=\min\{\frac{1}{k},\frac{k-1}{2}\}$.

Since $H$ may have vertices of degree $0$,
we cannot  directly appeal to \cref{TH:lower-bound-strong}.
We overcome this by removing these vertices.
Formally, let $V_1=\{v\in V\,:\,\deg_H v\geq 1\}$
and   $H_1=(V_1,E)$, i.e., $H_1$ is the hypergraph
obtained from $H$ by removing all vertices of degree $0$.
Put $\alpha=\frac{|V_1|}{|V|}$. Then $H_1$ has $\alpha n$ vertices,
$rn$ edges and average uniformity $k$. Thus, its sparsity is
$r_1:=\frac{r}{\alpha}$.
We now break into three cases.

\smallskip

{\it Case I. $\alpha\leq \delta$.}
In this case, $A:=V_1$ has density at most $\delta$ and satisfies
$\frac{|E_H(A)|}{|E|}=\frac{|E|}{|E|}=1$.

\smallskip

{\it Case II. $\alpha>\delta$ and $\alpha\geq \frac{1}{2}$.}
Suppose that
$
\alpha n\geq C'k^{a'}r_1^{b'}\delta^{-u'}(1-\delta)^{-v'}=
\alpha^{-b'} C'k^{a'}r^{b'}\delta^{-u'}(1-\delta)^{-v'} $.
Since $\alpha>\delta$, this condition is satisfied  when
\[
n\geq  C'k^{a'}r^{b'}\delta^{-u'-b'-1}(1-\delta)^{-v'}:=N_1.
\]
Now, by applying  \cref{TH:lower-bound-strong} to $H_1$, we obtain
a subset $A\subseteq V_1$ of density $\delta$ in $V_1$ such that
\[
\frac{|E_{H}(A)|}{|E|}=
\frac{|E_{H_1}(A)|}{|E|}
\geq
f_{k,3r_1}(\delta )-63200 k^4\delta^{-3} (\alpha n)^{-\sigma}=(\star).
\]
Using the facts
we recalled about $f_{k,r}(x)$ (\cref{TH:distance-to-delta-to-k})
and our assumptions $\alpha>\delta$ and $\alpha\geq\frac{1}{2}$
(which imply $r_1\leq 2r$), we further have
\begin{align*}
(\star)
&\geq
\parens*{\delta^k+\frac{c(k)\delta(1-\delta)^2}{3r_1}}-63200
k^4\delta^{-3} (\delta n)^{-\sigma}
\geq
\delta^k+\frac{ c(k)\delta(1-\delta)^2}{6r}-63200 k^4\delta^{-4}
n^{-\sigma}.
\end{align*}

\smallskip

{\it Case III. $\frac{1}{2}>\alpha>\delta$.}
Put $\delta_1=\frac{\delta}{2\alpha}$. Then
$\delta\leq \delta_1\leq \frac{1}{2}$, and hence $1-\delta_1\geq \frac{1}{2}$.
As a result, the condition
$\alpha n\geq C'k^{a'}r_1^{b'}\delta_1^{-u'}(1-\delta_1)^{-v'}
=\alpha^{-b'}C'k^{a'}r^{b'}\delta_1^{-u'}(1-\delta_1)^{-v'}$ holds
if
\[
n\geq
C'k^{a'}r^{b'}\delta^{-u'-b'-1}(0.5)^{-v'}=(2^{v'}C')k^{a'}r^{b'}\delta^{-u'-b'-1}=:N_2.
\]
Assuming this, \cref{TH:lower-bound-strong} tells us that there is
$A\subseteq V_1$ of density $\delta_1$ in $V_1$ such that
\[
\frac{|E_{H}(A)|}{|E|}=
\frac{|E_{H_1}(A)|}{|E|}
\geq
f_{k,3r_1}(\delta_1 )-63200 k^4\delta_1^{-3} (\alpha
n)^{-\sigma}=(\star\star),
\]
and similarly to Case II,
\begin{align*}
(\star\star) &\geq
\parens*{\delta_1^k+\frac{c(k)\delta_1(1-\delta_1)^2}{3r_1}}-63200
k^4\delta_1^{-3} (\delta n)^{-\sigma}
\\
&\geq
\delta^k + \frac{c(k)\cdot\frac{\delta}{2\alpha} \cdot
(\frac{1}{2})^2 }{\frac{3r}{\alpha}}
-63200 k^4\delta^{-4} n^{-\sigma}
\\
&=\delta^k   +\frac{ c(k)\delta(1-\delta)^2}{24 r}-63200
k^4\delta^{-4}   n^{-\sigma}.
\end{align*}
Note also that $|A|\leq \frac{\delta}{2\alpha}\cdot\alpha
n=\frac{\delta}{2}n$, so
$A$ has density smaller than $\delta$.

\smallskip

Now, since at least one of Case I--III holds, if $n\geq \max\{N_1,N_2\}$,
then there is $A\subseteq V$ of density at most $\delta$ such that
\[
\frac{|E_{H}(A)|}{|E|}\geq \min\left\{\delta^k   +\frac{
c(k)\delta(1-\delta)^2}{24r}-63200 k^4\delta^{-4}   n^{-\sigma},1\right\}.
\]
If we assume further that
\[
n\geq \parens*{120\cdot \frac{r}{ c(k)\delta(1-\delta)^2}\cdot 63200
k^4\delta^{-4}}^{\sigma^{-1}}=:N_3,
\]
then we would get
\[
\frac{|E_{H}(A)|}{|E|}\geq
\min\left\{
\delta^k   +\parens*{\frac{1}{24}-\frac{1}{120}}
\frac{ c(k)\delta(1-\delta)^2}{ r}
,
1
\right\}
=
\min\left\{
\delta^k +\frac{ c(k)\delta(1-\delta)^2}{30 r},
1\right\}.
\]
Moreover, we can add vertices to   $A$ until it has density $\delta$
in $V$ without affecting the validity of this conclusion.
It is therefore enough to choose $a,b,u,v,w,C$ such that
\[
C k^a r^b \delta^{-u} (1-\delta)^{-v} \min\{1,k-1\}^{-w}
\geq \max\{N_1,N_2,N_3\}.
\]
One readily checks that we can take
\begin{align*}
a &= \max\{a',5\sigma^{-1}\},
&
b &= \max\{b',\sigma^{-1}\},
\\
u &= \max\{u'+b'+1, 5\sigma^{-1}\},
&
v &=\max\{v',2\sigma^{-1}\},
\\
w &= \sigma^{-1},
&
C &= \max\{2^{v'}C', (2\cdot 120\cdot 63200)^{\sigma^{-1}}\}.
\end{align*}
We finally choose $\sigma=\frac{7-\sqrt{33}}{6}$
and $a'=b'=u'=\sigma'^{-1}=7+\sqrt{33}$,
$v'=\frac{1}{2}\sigma'=\frac{7+\sqrt{33}}{2}$ to obtain the $a,b,u,v,w$
specified in the theorem. The value of $C$ evaluates to be
slightly less than $10^{91.52}$.
\end{proof}

Next, we observe that part (i) of \cref{TH:lower-bound-strong}
follows from part (ii).

\begin{proof}[Proof of \cref{TH:lower-bound-strong}(i) assuming
\cref{TH:lower-bound-strong}(ii).]
Let $a,b,u,v,C$ be constants promised for part (ii).
Suppose $  n\geq  Ck^a (3r)^b \delta^{-u}(1-\delta)^{-v}
=(3^b C)k^a r^b \delta^{-u}(1-\delta)^{-v}$.
Let $H^{(3)}$ be the hypergraph obtained from $H$ by repeating each
hyperedge $3$ times.
Formally, $H^{(3)}=(V,E\times\{1,2,3\})$, where for a hyperedge
$(e,i)\in E\times \{1,2,3\}=E(H^{(3)})$, we set $\Vrt_{H^{(3)}}(e,i)=\Vrt_H(e)$.
Clearly, $|E(H^{(3)})|/|V(H^{(3)})|=3|E|/|V|=3r$ and
$\deg_{H^{(3)}}(v)=3\deg_H(v)$ for all $v\in V$. Since we assumed $H$
has no veritces of degree $0$, every   vertex in $H^{(3)}$ has
degree at least $3$.
Thus, by applying \cref{TH:lower-bound-strong}(ii) to $H^{(3)}$,
we see that there is $A\subseteq V$ with $\lfloor \delta n\rfloor$
elements such that
\[
\frac{|E_{H^{(3)}}(A)|}{|E(H^{(3)})|}\geq
f_{k,3r}(\delta)-63200k^4\delta^{-3}n^{-\sigma}.
\]
However, it is straighforward to see that
\[
\frac{|E_{H^{(3)}}(A)|}{|E(H^{(3)})|}=
\frac{|E_H(A)\times\{1,2,3\}|}{|E\times\{1,2,3\}|}=\frac{|E_H(A)|}{|E|},
\]
so $A$ is the set we are looking for.
\end{proof}

To conclude our previous discussion, we have reduced the proofs of
\cref{TH:lower-bound} and
\cref{TH:lower-bound-strong}(i)
into proving the stronger lower
bound on the worst-case confinement probability in
\cref{TH:lower-bound-strong}(ii).


The proof of \cref{TH:lower-bound-strong}(ii) is quite involved and
occupies all of \cref{sec:lower-bound-proof} and half of
\cref{sec:analytic}. We survey it now before diving into the details.


Let $k$, $r$ and $d$ denote the average uniformity, sparsity
and average vertex degree of $V$, and note that $d=kr$
(\cref{PR:degree-to-unif-ratio}).
Suppose   first that the maximum vertex degree    and the
maximum uniformity  of $H$ are not too
large in the sense that they do not exceed
$dn^\alpha$ and $kn^\alpha$, respectively, for some small $\alpha>0$.
Under this assumption, we prove \cref{TH:lower-bound-strong}(ii)
by using a random process: Fixing some $\gamma\in (0,1)$, we choose
a random subset $\rv B$ of $E$ by including each $e\in E$ in $B$ with
probability
$\gamma$ (independently of the other edges), and let $\rv A$ be the
set of vertices covered by the edges of $B$, i.e., $\rv
A=\bigcup_{e\in \rv B}\Vrt(e)$.
The rationale behind this choice is that $E(\rv A)$ is guaranteed
to include $\rv B$, which consists of approximately $\gamma$-fraction of
the edges in $H$.
We show that for a carefully chosen value of $\gamma$, there is a positive
probability 
that $\rv A$ will satisfy both $|\rv A|\leq \delta n$
and $\frac{|E(\rv A)|}{|\rv A|}\geq f_{k,r}(\delta)-o(1)$, and thus
the desired set $A$ from \cref{TH:lower-bound-strong}(ii) exists. We
establish this probabilistic statement  in three main steps
carried out in \cref{subsec:lower-bound-with-assumpt}:
\begin{enumerate}[(1)]
\item First, we show that the size of $\rv A$ is almost surely very close to
a certain number, which is determined by $\gamma$ and the  vertex
degrees appearing in $H$
(\cref{PR:vertex-upper-bound-flip}).
Interestingly, this is shown using results from
\cref{sec:almost-every-set}. For example, when $H$ regular, the dual
hypergraph $H^*=(E,V)$ (see \cref{subsec:hypergraphs}) is a uniform
hypergraph, and the set $\rv A$
is precisely the set of $H^*$-edges that meet the set of
$H^*$-vertices $\rv B$, so we can use   \cref{cor:thm0} to estimate
its size.
A similar but more sophisticated argument is used for non-regular
hypergraphs.

\item Next, we show that the expected number of edges in
$E(\rv A)$ is bounded from below by $|E|$ times a certain number
depending on $\gamma$ and the vertex degrees occurring in $H$
(\cref{PR:edges-fraction-avg}).
We remark that $E(\rv A)$ is guaranteed to contain $\rv B$, and the
almost all the work is dedicated to estimate the contribution of the
edges that are \emph{not} in $\rv B$ and are still covered by  $\rv A$. To
prove this step
we establish a general result
about the probability that a random subset of a set meets some other
subsets (\cref{PR:meeting-probability}),
which may be of independent interest.

\item Finally, we prove that if we choose $\gamma$ so that the
expected size of $\rv A$, as determined in (1), is $\delta n$, then
the lower bound on $\Exp(|E(\rv A)|)$ is at least $f_{k,r}(\delta)
\cdot |E|$. This is the content of the extermely technical multi-variable
ineqaulity established in \cref{TH:hard-optimization}. By settling
for a slightly smaller
$\delta$, this is enough to prove that  $|\rv A|\leq \delta n$
and $\frac{|E(\rv A)|}{|\rv A|}\geq f_{k,r}(\delta)-o(1)$ happen
simultaneously with positive probability
(\cref{TH:pre-lower-bound-i}, see also \cref{LM:Markov-modified}).
\end{enumerate}
For example, in the special case where $H$ is regular (the degree of
regularity is $d=rk$),
the expected size of $\rv A$ in (1) turns out to be $1-(1-\gamma)^{rk}$, so we
need to choose $\gamma=1-(1-\delta)^{\frac{1}{rk}}$ in order to have
$|\rv A|\approx \delta n$ almost surely. For this $\gamma$, the expected size of
$|E(\rv A)|$ in (2) turns out
to be $f_{k,r}(\delta)$, so step (3) is unecessary for this case. (In
fact, in this case, the term $1-(1-\delta)^{\frac{1}{rk}}$ in
$f_{r,k}(\delta)\approx \frac{|E(\rv A)|}{|\rv A|}$ is the
contribution of the edges in $\rv B$, whereas the remaining term
$(1-\delta)^{\frac{1}{rk}}
(1-(1-\delta)^{\frac{rk-1}{rk}})^k$ is the contribution of the edges
in $E(\rv A)-\rv B$.)

In the general case, i.e., when $H$ does have vertices of large
degree, or edges of large size, we first modify $H$ to eliminate these
vertices and edges. This affects only a negligible amount of
vertices and edges, and so we can    still assert
\cref{TH:lower-bound-strong}(ii). The details of this modifcation are
the subject matter of
\cref{subsec:removing-vrt-edges}, where the proof of
\cref{TH:lower-bound-strong}(ii) is completed.

\medskip

We conclude this section with some corollaries and remarks.
We begin with the following corollary which summarizes the implications
of \cref{TH:lower-bound} and \cref{TH:lower-bound-strong} (and also
\cref{cor:conf_lowerbound_main})
to confiners.

\begin{corollary}
\label{CR:all_lower_bounds}
Let $H=(V,E)$ be a hypergraph with $n$ vertices,
$rn$  edges ($r>0$) and average uniformity $k$. Suppose
that $H$ is also
an $(\epsilon,p)$-confiner with $\epsilon, p\in(0,1)$. Then:
\begin{enumerate}[(i)]
\item \label{item:simple} $\epsilon\geq p^{k}-\frac{2k}{n}$, provided $n\geq
\frac{1}{\delta(1-p)}$;
\item $\epsilon\geq
p^k+\min\{\frac{1}{k},\frac{k-1}{2}\}\cdot\frac{p(1-p)^2}{30}\cdot\frac{1}{r}$,
provided $n\geq r^{12.75}\cdot
\poly(k,\textstyle{ \frac{1}{p},\frac{1}{1-p}})$ and $k>1$;
\item \label{item:stronger} $\epsilon\geq f_{k,3r}(p)-\poly(k,p^{-1})\cdot
n^{-0.078}$,
provided $n\geq r^{12.75}\cdot
\poly(k,\textstyle{ \frac{1}{p},\frac{1}{1-p}})$, $k\geq 1$,
and $H$ has minimum vertex degree at least $1$;
\item \label{item:strongest} $\epsilon\geq f_{k,r}(p)-\poly(k,p^{-1})\cdot
n^{-0.078}$,
provided $n\geq r^{12.75}\cdot
\poly(k,\textstyle{\frac{1}{k-1},\frac{1}{p},\frac{1}{1-p}})$, $k\geq 1$,
and $H$ is regular or has minimum vertex degree at least $3$.
%
%
\end{enumerate}
\end{corollary}

Each lower bound in \cref{CR:all_lower_bounds}
is asymptotically stronger than its predecessor, namely,
$p^k\leq
p^k+\min\{\frac{1}{k},\frac{k-1}{2}\}\cdot\frac{p(1-p)^2}{30}\cdot\frac{1}{r}
\leq f_{k,3r}(p)\leq f_{k,r}(p)$, but on the other hand, it applies
under stronger assumptions.

Recall from \cref{PR:conf_duality} that the dual hypergraph of an
$(\epsilon,p)$-confiner is a $(1-p,1-\epsilon')$-confiner
for all $\epsilon'>\epsilon$.
By combining this with \cref{CR:all_lower_bounds}, we obtain the following:

\begin{corollary}
\label{CR:all_lower_bounds_dual}
With notation and assumptions as in  Corollary~\ref{CR:all_lower_bounds}:
\begin{enumerate}
\item[({\ref{item:simple}}$^*$)] $\epsilon\geq
1-(1-p)^{\frac{1}{rk}}-\frac{2k}{e p(1-p) n}$,
provided $n\geq \frac{4k}{p(1-p)}$ and $rk\geq 1$;
\item[({\ref{item:stronger}}$^*$)] $\epsilon\geq 1-
f_{rk,3/r}^{-1}(1-p)-\poly(r,k,\frac{1}{p},\frac{1}{1-p})n^{-0.078}$,
provided $n\geq \poly(r,k,\frac{1}{p},\frac{1}{1-p})$, $rk\geq 1$,
and $H$ has no empty hyperedges;
\item[({\ref{item:strongest}}$^*$)] $\epsilon\geq 1-
f_{rk,1/r}^{-1}(1-p)-\poly(r,k,\frac{1}{p},\frac{1}{1-p})n^{-0.078}$,
provided $n\geq \poly(r,k,\frac{1}{p},\frac{1}{1-p})$, $rk\geq 1$,
and $H$ has minimum uniformity at least $3$.
\end{enumerate}
\end{corollary}

Here, $f^{-1}_{k,r}(x)$ is the inverse function of $f_{k,r}:[0,1]\to [0,1]$.
As in \cref{CR:all_lower_bounds}, each item
of \cref{CR:all_lower_bounds_dual} asserts an asymptotically
stronger lower bound than the previous item, but also makes
stronger assumptions.
When $p$ is small, all three lower bounds of \cref{CR:all_lower_bounds_dual}
are very close to the lower bound $f_{k,r}(p)$ in
\cref{CR:all_lower_bounds}(\ref{item:strongest}). Indeed, by
checking the derivatives at $p=0$, one finds that all of these bounds
become $\epsilon \geq \frac{p}{rk}(1+o(1))$ as $p\to 0$.
Note, however, that (\ref{item:simple}*) holds under very mild
assumptions, and its proof is very simple compared to (\ref{item:strongest}).
By contrast, when $p$ approches $1$, the bounds in (\ref{item:simple}*)
and (\ref{item:stronger}*) are weaker than the bound
$\epsilon\geq p^k-o(1)$ of \cref{CR:all_lower_bounds}(\ref{item:simple});
this can be seen by  checking that the slopes
the graphs of $1-(1-x)^{\frac{1}{rk}}$, $1-f^{-1}_{rk,3/r}(1-x)$
and $x^k$ at the point $(1,1)$ are $\infty$, $3k$ and $k$, respectively.
Finally, regarding the two asymptotically strongest bounds
(\ref{item:strongest}) and (\ref{item:strongest}*),
experiments suggest that there is $a=a(r,k)\in (0,1)$
such that (\ref{item:strongest}*) is stronger when $0< p<a$ while
(\ref{item:strongest}) is stronger when $a<p<1$. It also seems
that using (\ref{item:strongest}*) instead of (\ref{item:strongest})
does not lead to notable improvement in \cref{TH:lower-bound}.

\begin{proof}[Proof of \cref{CR:all_lower_bounds_dual}]
Throughout,  $H$ is assumed to be an $(\epsilon,p)$-confiner.

({\ref{item:simple}}$^*$)
Set $c_1=1-(1-p)^{\frac{1}{rk}}$
and $c_0=1-(1-\frac{p}{2})^{\frac{1}{rk}}$.
Suppose that $n$ is large enough so that
$\epsilon_0:=1-(1-p+\frac{2k}{n})^{\frac{1}{rk}}> c_0$
and $r n\geq
\max\{\frac{1}{\epsilon'(1-\epsilon')}\,|\,\epsilon'\in [c_0,c_1]\}$.
Under this assumption, we claim
that $\epsilon\geq \epsilon_0$. Indeed, suppose it is not
the case. Then there exists $\epsilon'\in (\epsilon,\epsilon_0)\cap
[c_0,c_1]$ (here we need $c_0<\epsilon_0$). By
\cref{PR:conf_duality}, $H^*$ is a $(1-p,1-\epsilon')$-confiner.
Note also that $H^*$
has $rn$ vertices, $n$ edges, average uniformity $rk$
(this is the average degree of $H$),
and sparsity $\frac{1}{r}$. Moreover, our assumption on $n$
implies that $rn\geq\frac{1}{\epsilon'(1-\epsilon')}$.
Thus, by  \cref{CR:all_lower_bounds}(\ref{item:simple}),
$1-p\geq (1-\epsilon')^{rk}-\frac{2rk}{rn}$.
This rearranges to $\epsilon'\geq
1-(1-p+\frac{2k}{n})^{\frac{1}{rk}}=\epsilon_0$, and contradicts our
choice of $\epsilon'$.
We conclude that it must be the case that $\epsilon\geq \epsilon_0$.

Now, in order to prove ({\ref{item:simple}}$^*$),
it is enough to check
that $\epsilon_0\geq 1-(1-p)^{\frac{1}{rk}}-\frac{2k}{e p(1-p) n}$
and that our assumptions $\epsilon_0>c_0$ and  $r n\geq
\max\{\frac{1}{\epsilon'(1-\epsilon')}\,|\,\epsilon'\in [c_0,c_1]\}$
are satisfied when $n\geq \frac{4k}{p(1-p)}$.
The inequality $\epsilon_0\geq 1-(1-p)^{\frac{1}{rk}}-\frac{2k}{e
p(1-p) n}$ is equivalent to
$(1-p+\frac{2k}{n})^{\frac{1}{rk}}-(1-p)^{\frac{1}{rk}}
\leq \frac{2k}{e p(1-p) n}$,  which
is a consequence of \cref{LM:power-dist}(ii) and our assumption $rk\geq 1$.
The condition $\epsilon_0>c_0$
rearranges to
$(1-p+\frac{2k}{n})^{\frac{1}{rk}}<(1-\frac{p}{2})^{\frac{1}{rk}}$,
and holds because $n\geq\frac{4k}{p(1-p)}>\frac{4k}{p}$. Finally, if
$\epsilon'\in [c_0,c_1]$,
then
\[
\frac{1}{\epsilon'(1-\epsilon')}
\leq \frac{1}{c_0(1-c_1)}
=\frac{1}{(1-(1-\frac{p}{2})^{\frac{1}{rk}})(1-p)^{\frac{1}{rk}}}
\leq \frac{1}{\frac{p}{2}\cdot \frac{1}{rk}\cdot (1-p)}=
r\cdot \frac{2k}{p(1-p)}\leq rn,
\]
where here, we used \cref{LM:power-ineq} and the fact that $rk\geq 1$.

({\ref{item:stronger}}$^*$)
Let $N(k,r,p)$ and
$C(k,p)$ denote the terms $\poly(k,r,\frac{1}{p},\frac{1}{1-p})$
and $\poly(k,\frac{1}{p},\frac{1}{1-p})$
appearing in the statement of
\cref{CR:all_lower_bounds}(\ref{item:stronger}), and put $\sigma=0.078$.
Then, according to that corollary, $\epsilon\geq f_{k,r}(p)-C(k,p)n^{-\sigma}$
when $n\geq N(k,r,p)$.

Put $c_1=1-f^{-1}_{kr,3/r}(1-p)$,
$c_0=1-f^{-1}_{kr,3/r}(1-\frac{p}{2})$, and define
\begin{align*}
& M(k,r,p)=\max_{\epsilon'\in [c_0,c_1]}
N(kr,{\textstyle\frac{1}{r}},1-\epsilon'),
& &
D(k,r,p)=\max_{\epsilon'\in [c_0,c_1]} C(kr,1-\epsilon'),
\\
&
\epsilon_0= 1-f^{-1}_{kr,3/r}(1-p+D(k,r,p)n^{-\sigma}).
\end{align*}
Suppose first that  $n$ is large enough
so that  $n\geq \frac{1}{r}M(k,r,p)$ and $\epsilon_0>c_0$.
We claim that
$\epsilon\geq \epsilon_0$.
Indeed, for the sake of contradiction, suppose that
$\epsilon< \epsilon_0$. Then there exists $\epsilon'\in
(\epsilon,\epsilon_0)\cap [c_0,c_1]$ (here we need $\epsilon_0>c_0$).
By \cref{PR:conf_duality}, $H^*$ is an $(1-p,1-\epsilon')$-confiner,
and by our assumption on $n$, we have $rn\geq M(k,r,p)\geq
N(kr,\frac{1}{r}, 1-\epsilon')$.
Thus, by \cref{CR:all_lower_bounds}(\ref{item:stronger}) (applied to $H^*$),
we have
\[1-p\geq f_{kr,3/r}(1-\epsilon')-C(k,1-\epsilon')n^{-\sigma}\geq
f_{kr,3/r}(1-\epsilon')-D(k,r,p)n^{-\sigma}.\]
This rearranges to $\epsilon'\geq
1-f_{kr,3/r}^{-1}(1-p+D(k,r,p)n^{-\sigma})=\epsilon_0 $,
which contradicts the choice of $\epsilon'$. We conclude that our
claim must be correct.

Now, in order to prove ({\ref{item:stronger}}$^*$), it is enough
to show that
when $n\geq\poly(k,r,\frac{1}{p},\frac{1}{1-p})$,
we have,
$ \epsilon_0\geq
1-f^{-1}_{kr,3/r}(1-p)-\poly(k,r,\frac{1}{p},\frac{1}{1-p})n^{-\sigma}$,
$\epsilon_0>c_0$ and
$n\geq \frac{1}{r}M(k,r,p)$.
Observe first that by \cref{PR:derivative},
on $[0,1]$, the derivative of $f_{kr,3/r}(x)$ takes values in
$[\frac{1}{3k},kr]$,
so the derivative of $f_{kr,3/r}^{-1}(x)$ takes values in $[\frac{1}{kr},3k]$.
Thus, by  \cref{LM:diff_trick},
\[
\epsilon_0=1-f^{-1}_{kr,3/r}(1-p+D(k,r,p)n^{-\sigma})\geq 1-
f^{-1}_{kr,3/r}(1-p)-3k D(k,r,p)n^{-\sigma}.
\]
Furthermore, it is straightforward
to see that $\epsilon_0>c_0$ if and only if $n^\sigma\geq
\frac{2}{p}\cdot D(k,r,p)$. It is therefore enough to show that
$\frac{1}{r}M(k,r,p)$ and $D(k,r,p)$
are in the class $\poly(k,r,\frac{1}{p},\frac{1}{1-p}) $.
By definition, both functions are bounded by
$\poly(k,r,\frac{1}{r},\frac{1}{\epsilon'},\frac{1}{1-\epsilon'})$
for some $\epsilon'\in [c_0,c_1]$, so it is enough
to show that $\frac{1}{r}\leq \poly(k)$
and $\frac{1}{\epsilon'}, \frac{1}{1-\epsilon'}\leq
\poly(r,k,\frac{1}{p},\frac{1}{1-p})$.
The first statement follows from our assumption $rk\geq 1$,
which implies $\frac{1}{r}\leq k$. For the second statement,
since $\frac{1}{\epsilon'}\leq\frac{1}{c_0}$ and
$\frac{1}{1-\epsilon'}\leq \frac{1}{1-c_1}$,
it is enough to show that $\frac{1}{c_0},\frac{1}{1-c_1}\leq
\poly(r,k,\frac{1}{p},\frac{1}{1-p})$.
Indeed, by \cref{LM:diff_trick} and our ealier conclusion
on the derivative of $f^{-1}_{kr,3/r}$,
\[
c_0=f^{-1}_{kr,3/r}(1)-f^{-1}_{kr,3/r}(1-{\textstyle\frac{p}{2}})
\geq \frac{1}{rk}\cdot \frac{p}{2},
\]
so $\frac{1}{c_0}\leq \frac{2rk}{p}$, and similarly,
\[
1-c_1=f^{-1}_{kr,3/r}(1-p)=f^{-1}_{kr,3/r}(1-p)-f_{kr,3/r}(0)\geq
\frac{1-p}{rk},
\]
so $\frac{1}{1-c_1}\leq \frac{kr}{1-p}$. This is exactly what we want.

%

({\ref{item:strongest}}$^*$) This is completely analogous to proof of
({\ref{item:stronger}}$^*$).
%
%
\end{proof}

\begin{remark}
We conjecture that \cref{CR:all_lower_bounds}(\ref{item:strongest})
holds for general hypergraphs, possibly with different constants.
\cref{CR:all_lower_bounds_dual}(\ref{item:simple}*),
which gives a very similar bound for small $p$ and holds under very
mild assumptions, may serve as evidence for that.
%
\end{remark}

\begin{remark}
The proof of \cref{TH:lower-bound-strong}(ii) suggests that
the best confiners are regular and have the property
that every two edges meet at no more than one vertex;
see \cref{RM:not-simplicial} and \cref{RM:regular_is_best} for details.
In fact, for non-regular hypergraphs of minimum degree
at least $3$, the lower bound of \cref{TH:lower-bound-strong}(ii) can
be slightly improved; see  \cref{RM:regular_is_best}.
\end{remark}

\begin{remark}
(i)
There are more options for choosing the parameters
$a,b,u,v,C$ of
\cref{TH:lower-bound-strong}
than is stated in the theorem.
The possible values that these parameters
can take are described \cref{CR:lower_bound_indep_of_r},
which is a special case of the more general
\cref{TH:lower-bound-strong-detailed}.
The values of $a,b,u,v,C$ that we specified
as the second option in \cref{TH:lower-bound-strong}
are meant to minimize the value of the parameter
$b$ (the exponent of $r$ or $\ee^{-1}$)
in \cref{TH:lower-bound} and \cref{TH:lower-bound-epsilon-view}.
This in turn increases the constant $c$ from \cref{CR:random_algs},
because $c=\frac{1}{b}$.

(ii)
When $H$ is regular and uniform, or even when the maximal degree and
maximal uniformity of $H$
are not too large,
one can lower the smallest value
that $n=|V|$ can take in \cref{TH:lower-bound-strong}(ii)
and also the error term
$63200 k^3\delta^{-3}n^{-\sigma}$.
See \cref{TH:pre-lower-bound-i} and
\cref{CR:lower_bound_for_reg_unif}.

(iii) In all the results of this section and \cref{sec:lower-bound-proof},
we have made no attempt in optimizing constants,
save for minimizing the exponent of $r$ (resp.\ $\ee^{-1}$)
in \cref{TH:lower-bound} (resp.\ \cref{TH:lower-bound-epsilon-view}).
The constants appearing
in  these results can very
likely be improved.
\end{remark}

\begin{remark}
The bound $f_{k,r}(x)-x^k\geq \min\{\frac{1}{k},\frac{k-1}{2}\}\cdot
\frac{x(1-x)^2}{r}$
(\cref{TH:distance-to-delta-to-k}) that is used
in deriving \cref{TH:lower-bound}
from \cref{TH:lower-bound-strong} is
likely optimal in order of magnitude, except possibly
when $x$ approaches $1$, where it is \emph{nearly} optimal.
Indeed, setting $D_{k,r}(x)=f_{k,r}(x)-x^k$,
experiments suggest that  the limits
$\lim_{r\to\infty} \frac{D_{k,r}(x)}{r}$,
$\lim_{k\to\infty} \frac{D_{k,r}(x)}{k}$
and
$\lim_{k\to1^-} \frac{D_{k,r}(x)}{1/(k-1)}$
exist and are positive for all $x\in (0,1)$.
Furthermore,  $D'_{k,r}(0)=\frac{1}{rk}$,
which means that $D_{k,r}(x)\approx \Theta_{k,r}(x)$
as $x$ approaches $0$,
and  $D'_{k,r}(1)=0$, which means
that $D'_{k,r}(x)=o (1-x)$ when $x$ approaches $1$ from the left.
Experiments also suggest that when $x$ approaches $1$ from the left,
$D_{k,r}(x)=\Omega_{k,r}((1-x)^{2})$
and $D_{k,r}(x)=o((1-x)^{2-\ee})$
for all $\ee>0$.
\end{remark}

\begin{remark}\label{RM:k_less_than_1}
(i)
There is no hope in proving a result similar to
\cref{TH:lower-bound} for $k=1$. That is,
in a hypergraph $H=(V,E)$ with average uniformity $1$ and sparsity $r$,
the maximum value of $\frac{|E(A)|}{|E|}$
as $A$ ranges over the $\delta$-dense subsets of $V$
is not always bounded away from $\delta$.
This can be seen, for example, by fixing some $r\in\NN$ and
considering the hypergraph
$H=(V,E)$ with $V=\{1,\dots,n\}$ and with hyperedges
$\{1\},\dots,\{n\}$ each repeated
$r$ times. Indeed, for this hypergraph, any $A\subseteq V$ of size
$\lfloor\delta n\rfloor$ satisfies $\frac{|E(A)|}{|E|}\leq \delta$.

(ii) A variant of \cref{TH:lower-bound} is also true for $0<k<1$,
although this case
is not so interesting because $H$ must have edges having no vertices.
This case can
be reduced to the case $k\geq 1$ by removing the edges with no
vertices. We omit the details.
\end{remark}

\section{Proof of Lower Bounds on Sparse Confiners}
\label{sec:lower-bound-proof}

In this section we prove \cref{TH:lower-bound-strong}(ii).
First, in \cref{subsec:prob-result}, we prove a general probabilistic
result that is needed for the proof,
but may be of independent interest.
Then, in \cref{subsec:lower-bound-with-assumpt}, we prove a variant
of \cref{TH:lower-bound-strong}(ii) under
the assumption that
the maximum degree and maximum uniformity are not too large with respect
to the average degree and average uniformity, respectively.
These assumptions are then removed at the expense
of changing the constants in \cref{subsec:removing-vrt-edges}. During
the proof, we use a few technical analytic results
whose proofs are postponed to \cref{sec:analytic}.

\subsection{A Probabilistic Result}
\label{subsec:prob-result}

\begin{proposition}\label{PR:meeting-probability}
Let $X$ be a finite nonempty set
and let $A_1,\dots,A_t,A'_1,\dots,A'_t$ be nonempty subsets of $X$
such that $|A_i|=|A'_i|$ for all $i \in [t]$.
Assume further that $A'_1,\dots,A'_t$ are pairwise disjoint.
Let $\rv B$ be a subset of $X$ chosen at random in one of the following ways:
\begin{enumerate}
\item[(1)] $\rv B$ is chosen uniformly at random among all
$m$-element subsets of $X$ ($0\leq m\leq |X|$);
\item[(2)] every $x\in X$ is included in $\rv B$ with probability $p$
independently of the other elements of $X$ ($0\leq p\leq 1$).
\end{enumerate}
Then, in both cases (1) and (2),
\[
\Pp\sqbr*{\text{$\rv B$ \emph{meets} $A_1,\dots,A_t$}}\geq
\Pp\sqbr*{\text{$\rv B$ \emph{meets} $A'_1,\dots,A'_t$}}.
\]
\end{proposition}

Here, by saying that at set $B$ meets $A_1,\dots,A_t$, we mean that
$B\cap A_i\neq \emptyset$
for all $i$. Informally, the proposition says that having overlaps
between the sets
$A_1,\dots,A_t$ can only increase the probability
that $\rv B$ meets them. This is useful because
computing $\Pp\sqbr*{\text{$\rv B$ meets $A_1,\dots,A_t$}}$ is easier
when $A_1,\dots,A_t$
are pairwise disjoint, especially in case (2).

We first prove the following lemma.

\begin{lemma}\label{LM:meeting-count}
Let $X,A_1,\dots,A_t $ be as in \cref{PR:meeting-probability}.
Fix $m\in\{0,1,\dots,|X|\}$ and let
$M(A_1,\dots,A_t)$ denote the collection of subsets $B\subseteq X$
having $m$ elements and meeting $A_1,\dots,A_t$.
Let $A'_t$ be another subset of $X$ with $|A'_t|=|A_t|$
that does not meet $A_1,\dots,A_{t-1}$.
Then
\[
\Abs*{M(A_1,\dots,A_t)} \geq \Abs*{M(A_1,\dots,A_{t-1},A'_t)}.
\]
\end{lemma}

\begin{proof}
We   prove the lemma by
constructing an injection $\Psi:M(A_1,\dots,A_{t-1},A'_t)\to M(A_1,\dots,A_t)$.
Since $|A_t|=|A'_t|$, we have $|A_t-A'_t|=|A'_t-A_t|$.
Choose a bijection $\sigma:A'_t-A_t\to A_t-A'_t$
and define $\Psi:M(A_1,\dots,A_{t-1},A'_t)\to M(A_1,\dots,A_t)$
by
\[
\Psi(B)=\left\{
\begin{array}{ll}
    B & B\cap A_t\neq\emptyset
    \\
    (B-(A'_t-A_t))\cup \sigma(B\cap (A'_t-A_t))  & B\cap A_t=\emptyset .
\end{array}
\right.
\]
Let us check first that $\Psi(B)$ is indeed in $M(A_1,\dots,A_t)$.
This is clear if $B$ meets $A_t$. On the other hand, if $B\cap
A_t=\emptyset$, then
$B\cap (A'_t-A_t)\neq\emptyset$ because $B$ meets $A'_t$,
and so $\Psi(B)\cap A_t\supseteq \sigma (B\cap (A'_t-A_t))\neq\emptyset$.
In addition, since $A'_t$ is disjoint from $\bigcup_{i=1}^{t-1} A_i$
and $B$ meets
$A_1,\dots,A_{t-1}$,
the set $B-A'_t$ meets $A_1,\dots,A_{t-1}$, and therefore so is $\Psi(B)$.
Finally,   $|\Psi(B)|=|B|=m$ because $\Psi(B)$ is obtained
from $B$ by removing $|B\cap(A'_t-A_t)|$ elements and then adding the
same amount
of elements. We conclude that $B\in M(A_1,\dots,A_t)$ when $B\cap
A_t=\emptyset$.

Before proving that $\Psi$ is injective, we first show that for every
$B\in M(A_1,\dots,A_{t-1},A'_t)$,  we have $\Psi(B)\cap A'_t\neq\emptyset$
if and only if $B\cap A_t\neq \emptyset$. Indeed, if $B\cap
A_t\neq\emptyset$,
then $\Psi(B)=B$ and thus $\Psi(B)\cap A'_t=B\cap A'_t\neq\emptyset$.
On the other hand, if $B\cap A_t=\emptyset$, then
$B-(A'_t-A_t)=B-(A'_t\cup A_t)$
and thus
$\Psi(B)=(B-(A'_t\cup A_t))\cup\sigma((B\cap (A'_t-A_t)))
\subseteq (B-(A'_t\cup A_t))\cup (A_t-A'_t)\subseteq B-A'_t$,
and it follows that $\Psi(B)\cap A'_t=\emptyset$.

We now prove that $\Psi$ is injective.
Suppose that $B,C\in M(A_1,\dots,A_{t-1},A'_t)$
and $\Psi(B)=\Psi(C)$.
We need to show that $B=C$.
If $\Psi(B) $ meets $A'_t$, then by the previous paragraph
both $B$ and $C$ meet $A_t$, and thus $B=\Psi(B)=\Psi(C)=C$.
On the other hand, if $\Psi(B)\cap A'_t=\emptyset$,
then both $B$ and $C$ do not meet $A_t$.
This means that $\Psi(B)\cap A_t=\sigma(B\cap (A'_t-A_t))$,
so
\begin{align*}
(\Psi(B)-(A_t-A'_t))\cup \sigma^{-1}(\Psi(B)\cap A_t)
&=
(B-(A'_t-A_t))\cup \sigma^{-1}(\sigma(B\cap (A'_t-A_t)))
\\
&=
(B-(A'_t-A_t))\cup  (B\cap (A'_t-A_t))
\\
&=B.
\end{align*}
The same holds if we replace $B$ with $C$, and it follows that
\[
B=
(\Psi(B)-(A_t-A'_t))\cup \sigma^{-1}(\Psi(B)\cap A_t)
=
(\Psi(C)-(A_t-A'_t))\cup \sigma^{-1}(\Psi(C)\cap A_t)=C,
\]
as required.
This proves that $\Psi$ is injective and completes the proof.
\end{proof}

\begin{proof}[Proof of \cref{PR:meeting-probability}]
Suppose first that $\rv B$ is chosen as in (1).
Then, in the notation of \cref{LM:meeting-count}, it is enough
to prove that
\begin{equation*}
|M(A_1,\dots,A_t)|\geq |M(A'_1,\dots,A'_t)|.
\end{equation*}

To that end, observe first that we may assume that $A'_t$ does not meet
$A_1,\dots,A_{t-1}$. Indeed, since $A'_1,\dots,A'_t$
are pairwise-disjoint, we have
$\sum_{i=1}^{t-1}|A_i|=\sum_{i=1}^{t-1}|A'_i|\leq |X|-|A'_t|$,
so there is a permutation $\sigma:X\to X$ such that
$\sigma(A'_t)$ is disjoint from $\bigcup_{i=1}^{t-1} A_i$.
Replacing $A'_1,\dots,A'_t$ with $\sigma(A'_1),\dots,\sigma(A'_t)$
does not affect the value of $|M(A'_1,\dots,A'_t)|$, and after
that replacement,  $A'_t$ does not meet
$A_1,\dots,A_{t-1}$.

Now that $A'_{t}$ does not meet $A_1,\dots,A_{t-1}$,
the sets $A_1,\dots,A_{t-1},A'_{t-1}$ are all contained in
$X':= X-A'_t$.
Thus, repeating the argument of the last paragraph with $X'$ in place of $X$
allows us to further assume that $A'_{t-1}$ does not meet
$A_1,\dots,A_{t-2}$ (provided $t\geq 2$).
Continuing in this manner, we may assume that
each $A'_i$ ($1\leq i\leq t$) does not meet $A_1,\dots,A_{i-1}$.

Fix some $i\in \{1,\dots,t\}$.
Then $A'_i$ is  disjoint from both $A_1,\dots,A_{i-1}$ and
$A'_{i+1},\dots,A'_{t}$.
Thus,  \cref{LM:meeting-count} tells us that
\[
\Abs*{M(A_1,\dots,A_{i-1},A_i,A'_{i+1},\dots,A'_t)}\geq
\Abs*{M(A_1,\dots,A_{i-1},A'_i,A'_{i+1},\dots,A'_t)}.
\]
As this holds for every $i$, we get that
\[
|M(A_1,\dots,A_t)|\geq |M(A_1,\dots,A_{t-1},A'_t)|
\geq |M(A_1,\dots,A_{t-2},A'_{t-1},A'_t)|\geq\dots\geq
|M(A'_1,\dots,A'_t)|.
\]
This proves the proposition   in case (1).

Next, suppose that $\rv B$ is chosen as in (2). In this case, if
we condition on the event that $\Abs*{\rv B} =m$ ($0\leq m\leq |X|$),
then $\rv B$ is distributed uniformly among the $m$-element subsets of $X$.
Our proof of case (1) therefore implies
\[
\Pp\sqbr*{\text{$\rv B$ meets $A_1,\dots,A_t$}\,:\,\Abs*{\rv B}=m}\geq
\Pp\sqbr*{\text{$\rv B$ meets $A'_1,\dots,A'_t$}\,:\,\Abs*{\rv B}=m}.
\]
As this holds for all $m\in\{0,1,\dots,|X|\}$, we get
\begin{align*}
\Pp\sqbr*{\text{$\rv B$ meets $A_1,\dots,A_t$}}
&=\sum_{m=0}^{|X|} \Pp\sqbr*{\text{$\rv B$ meets
$A_1,\dots,A_t$}\,:\,\Abs*{\rv B}=m}\Pp\sqbr*{\Abs*{\rv B}=m}
\\
&\geq \sum_{m=0}^{|X|} \Pp\sqbr*{\text{$\rv B$ meets
$A'_1,\dots,A'_t$}\,:\,\Abs*{\rv B}=m}\Pp[\Abs*{\rv B}=m]
\\
&=\Pp[\text{$\rv B$ meets $A'_1,\dots,A'_t$}],
\end{align*}
so the proposition holds in case (2).
\end{proof}

\subsection{Proof of \cref{TH:lower-bound-strong}(ii) When The
Maximum Degree Uniformity are Not Large}
\label{subsec:lower-bound-with-assumpt}

\begin{boxedc}{0.95\textwidth}{Notation}
The following general notation  will be assumed throughout this subsection:
\begin{itemize}
\item $H=(V,E)$ is a hypergraph
    with $n$ vertices,
\item $r=|E|/|V|$ is the sparsity (hyperedges-to-vertices ratio) of $H$,
\item $k$ is the average uniformity of $H$,
\item $d_1,\dots,d_m$
    are the vertex degrees that occur in $H$, i.e.,~$\{\deg(v)\,~:~\,
    v\in V\}=\{d_1,\dots,d_m\}$,
\item $V_i=\{v\in V~:~\deg (v)=d_i\}$ for $1\leq i\leq m$ is
    the set of vertices of degree $d_i$,
\item $u_i=\frac{|V_i|}{|V|}$ is the fraction of vertices of
    degree $d_i$
    for all $i = 1,\ldots, m$,
\item $d$ is the average degree of $H$.
\end{itemize}
Note that $d=\sum_{i=1}^m u_i d_i=rk$, where the second equality holds
by \cref{PR:degree-to-unif-ratio}.
\end{boxedc}

In order to construct the set $A\subseteq V$ promised by
\cref{TH:lower-bound-strong}(ii),
we will sample a  random set of edges $\rv B\subseteq E$ with density
$\gamma$ and set $\rv A=V(\rv B)$, i.e.,~$\rv A$ is the set of
vertices which are met by some edge in $\rv B$. We will show that for
a certain choice
of $\gamma$  (depending on $d_1,\dots,d_m$, $u_1,\dots,u_m$ and
$\delta$), there is a positive probability that $|\rv A|\leq \delta n$
and at the same time $|E(\rv A)|$ includes at least a $f_{k,r}(\delta)$-fraction
of the edges, where $f_{k,r}$ is as defined in \cref{EQ:f-k-r-delta-dfn}.

%

\begin{proposition}\label{PR:vertex-upper-bound-i}
Let $\alpha\in [0,\frac{1}{2})$, $\beta\in [0,1)$, $\gamma\in(0,1)$,
$\eta\in (0,\frac{1}{2})$ be real numbers such that
\[\mu:=(1-\alpha-\beta-\eta)(1-2\eta)-
\alpha>0\quad\quad\tand\quad\quad \gamma |E| \in \NN.
\]
Then there  are   constants
$A_{\ref{PR:vertex-upper-bound-i}}(k,r) ,
B_{\ref{PR:vertex-upper-bound-i}}(k,r,\gamma) ,
N_{\ref{PR:vertex-upper-bound-i}}(k,r,\alpha,\gamma)>0$
(depending only on the indicated parameters)
such that the following hold:
With notation as above, suppose that the hypergraph
$H=(V,E)$ satisfies    $d=rk\geq 1$ and
\begin{equation}\label{eq:extra_assumptions}\deg(v) \leq
d n^\alpha~\textrm{for all}~v \in V\quad\tand\quad |e| \le
k n^{\beta}~\textrm{for all}~e \in E.
\end{equation}
Let be $\rv B\subseteq E$ be chosen uniformly at
random among all  $\gamma|E|$-element subsets of $E$.
If   $n\geq N_{\ref{PR:vertex-upper-bound-i}}(k,r,\alpha,\gamma)$,  then
\[
\Pp_{\rv B}\sqbr*{\frac{|V(\rv B)|}{n}\leq
    \sum_{i=1}^m
    u_i(1-(1-\gamma)^{d_i})+A_{\ref{PR:vertex-upper-bound-i}}(k,r)
n^{-\eta}}
\geq 1-B_{\ref{PR:vertex-upper-bound-i}}(k,r,\gamma) n^{-\mu}.
\]
Moreover, we can take
\begin{align*}
A_{\ref{PR:vertex-upper-bound-i}}(k,r)\qquad&~=~2kd=O(k^2r),
\\
B_{\ref{PR:vertex-upper-bound-i}}(k,r,\gamma)~\quad&~=~\frac{4
C_{\ref{cor:thm0}} d}
{e^2\gamma(\ln(1-\gamma))^2} = O(rk\gamma^{-3}),
\\
N_{\ref{PR:vertex-upper-bound-i}}(k,r,\alpha,\gamma)&~=~
\max\left\{4dk,\frac{2k}{\gamma(1-\gamma)}\right\}^{(1-2\alpha)^{-1}}=\poly_{\alpha}(k,r,\gamma^{-1},(1-\gamma)^{-1})
.
\end{align*}
\end{proposition}

\begin{proof}
We may assume that $d_1=0$ by setting $u_1=0$ if there are no
vertices of degree $0$.
Under this assumption, $V(\rv B)$ never meets $V_1=\{v\in
V\,:\,\deg(v)=0\}$.

%
Our assumption   $\mu>0$ implies that $\alpha<1-\beta$.
Let $\tau\in (\alpha,1-\beta)$; we will specify
$\tau$ later.  We choose
$A_{\ref{PR:vertex-upper-bound-i}},
B_{\ref{PR:vertex-upper-bound-i}}, N_{\ref{PR:vertex-upper-bound-i}}$
as in the proposition.
In addition,
we define
\[
I=\{i\in \{2,\dots,m\}\,:\, u_i \geq  kn^{-\tau}\}
\qquad\text{and}\qquad
J=\{i\in \{2,\dots,m\}\,:\, u_i < kn^{-\tau} \},
\]
and note that $|I|,|J|\leq m-1\leq d n^\alpha$ by the first assumption  in
\cref{eq:extra_assumptions}.

Suppose $n\geq N_{\ref{PR:vertex-upper-bound-i}} $.
For every $i\in I$,
let $H_i$ denote the hypergraph obtained from $H$
by deleting the vertices that are not in $V_i$.
Formally, $H_i=(V_i,E)$, i.e.,~$H_i$ has $V_i$ as its set of vertices
and the same edge set as $H$,
but we only consider the edge-to-vertex incidences to vertices in
$V_i$ so that $ \Vrt_{H_i}(e)=\Vrt_H(e)\cap V_i$.
Consequently, $H_i$ is a $d_i$-regular hypergraph with $u_in$
vertices and $rn$ edges. By  \cref{eq:extra_assumptions},
its maximum uniformity is at most $k n^\beta$.
This means that the dual hypergraph $H_i^*$ (q.v.~\cref{subsec:hypergraphs})
is $d_i$-uniform with $u_in$ edges and $rn$ vertices and its
maximum vertex degree $D$ is at most $k n^\beta$.

We now apply
\cref{cor:thm0} to $H_i^*$ ($i\in I$) and the random subset
$\rv B\subseteq V(H_i^*)$, which has density $\gamma$.
Note first that
in order to apply it we need to have
$rn=|V(H^*_i)|\geq \max\{4d_i^2,\frac{2d_i}{\gamma(1-\gamma)}\}$.
Since $d_i\leq d n^\alpha$, it is enough to check
that $rn\geq n^{2\alpha}\max\{4d^2,\frac{2d }{\gamma(1-\gamma)}\}$,
which rearranges to
$n^{1-2\alpha}\geq \max\{4dk,\frac{2k}{\gamma(1-\gamma)}\}$.
This prerequisite holds by our assumption $n\geq
N_{\ref{PR:vertex-upper-bound-i}}$.
Next, observe that the number $\tilde{n}=\tilde{n}(H_i^*)$
associated with $H_i^*$ in \cref{cor:thm0} is
\[\tilde{n}(H_i^*)=\frac{u_i n}{D}\geq
\frac{kn^{1-\tau}}{kn^\beta}=n^{1-\beta-\tau}.\]
Thus, the corollary (applied with $1-2\eta$ in place of $\alpha$)
implies that
\begin{align*}
p_i& :=\Pp_{\rv B}\sqbr*{\Pp_{\rv v\sim
    E(H_i^*)}\sqbr*{\Vrt_{H_i^*}(\rv v)\cap \rv B=\emptyset}
\geq (1-\gamma)^{d_i}-(rn)^{-\eta}}
\\
& \geq 1-
\frac{4C_{\ref{cor:thm0}}}{e^2\gamma(\ln (1-\gamma))^2}\cdot
\tilde{n}(H_i^*)^{2\eta-1}
\\
& \geq 1-
\frac{C_{\ref{PR:vertex-upper-bound-i}}}{\gamma(\ln (1-\gamma))^2}
\cdot
n^{ (1-\beta-\tau)(2\eta-1)}.
\end{align*}

On the other hand, since $\Pp_{\rv v\sim
E(H_i^*)}\sqbr*{\Vrt_{H_i^*}(\rv v)\cap \rv B=\emptyset}
=\Pp_{\rv v\sim V_i}\sqbr*{\rv v\notin V_{H_i}(\rv B)}$, we have
\begin{align*}
p_i &= \Pp_{\rv B}\sqbr*{\Pp_{\rv v\sim V_i}\sqbr*{\rv v\notin
V_{H_i}(\rv B)}\geq (1-\gamma)^i-( rn)^{-\eta}}
\\
&=
\Pp_{\rv B}\sqbr*{\Pp_{\rv v\sim V_i}\sqbr*{\rv v\in V_{H_i}(\rv
B)}\leq 1-(1-\gamma)^i+(rn)^{-\eta}}
\\
&=\Pp_{\rv B}\sqbr*{|V_{H}(\rv B)\cap V_i|\leq
u_in(1-(1-\gamma)^i+(rn)^{-\eta})}.
\end{align*}
Hence, by the union bound,
\begin{align}\label{EQ:union-bound}
\Pp_{\rv B}\sqbr*{\text{$|V_{H}(\rv B)\cap V_i|\leq
        u_in\parens*{1-(1-\gamma)^i+(rn)^{-\eta}}$ for
all $i\in I$}}
&\geq 1-|I| \cdot
\textstyle\frac{C_{\ref{PR:vertex-upper-bound-i}}}{\gamma(\ln
(1-\gamma))^2}  n^{ (1-\beta-\tau)(2\eta-1)}
\\
&\geq
1-\textstyle\frac{C_{\ref{PR:vertex-upper-bound-i}}d}{\gamma(\ln
(1-\gamma))^2}  n^{ \alpha+(1-\beta-\tau)(2\eta-1)}
\nonumber
\end{align}
Observe now that since $V(\rv B)\cap V_1=\emptyset$,
when the event in \cref{EQ:union-bound} occurs, we have
\begin{align*}
|V(\rv B)|
&\leq
\sum_{i\in I}
\parens*{u_in\parens*{1-(1-\gamma)^i+(rn)^{-\eta}}}+\sum_{i\in J} u_in
\\
&\leq
n\sum_{i\in I}  u_i (1-(1-\gamma)^i) +
\sum_{i=1}^m(u_in)(rn)^{-\eta}+\sum_{i\in J} k n^{1-\tau}
\\
&\leq
n\sum_{i\in I}  u_i (1-(1-\gamma)^i) +
r^{-\eta}n^{1-\eta}+k|J|n^{1 -\tau}
\\
&\leq
n\sum_{i\in I}  u_i (1-(1-\gamma)^i) +k^\eta n^{1-\eta}+k d
n^{1+\alpha-\tau}
\\
&\leq
n\sum_{i=1}^m  u_i (1-(1-\gamma)^i) +k d n^{1-\eta}+k d
n^{1+\alpha-\tau}
\end{align*}
where in the fourth inequality we used that $r\geq \frac{1}{k}$.

We now choose $\tau$ so that $1-\eta =1+\alpha-\tau$, i.e.,
\[
\tau = \alpha+\eta.
\]
Note that $\tau<1-\beta$ by our assumption that $\mu>0$.
For this choice of $\tau$, we get that
\[
\Pp_{\rv B}\sqbr*{|V(\rv B)|\leq n\sum_{i\in I}  u_i (1-(1-\gamma)^i)
+2kd n^{1-\eta}}
\geq
1-\frac{C_{\ref{PR:vertex-upper-bound-i}}d}{\gamma(\ln (1-\gamma))^2}
n^{ \alpha+(1-\alpha-\beta-\eta)(2\eta-1)},
\]
which is   exactly what we want.
\end{proof}

\begin{proposition}\label{PR:vertex-upper-bound-flip}
Let $\alpha\in [0,\frac{1}{2})$, $\beta\in [0,1)$, $\gamma\in(0,1)$,
$\eta\in (0,\frac{1}{2} )$  be real numbers satisfying
\[
\mu:=( 1-\alpha  -\beta-\eta)(1-2\eta)-\alpha>0.
\]
Then there are constants $
A_{\ref{PR:vertex-upper-bound-flip}}(k,r),
B_{\ref{PR:vertex-upper-bound-flip}}(k,r,\gamma),
N_{\ref{PR:vertex-upper-bound-flip}}(k,r,\alpha,\beta,\gamma, \eta)>0$
(depending on the indicated parameters)
such that the following holds:
Suppose that
\begin{equation}\label{eq:extra_assumptions2}
\deg(v) \leq
dn^\alpha~\textrm{for all}~v \in V\quad\tand\quad |e| \le
kn^{\beta}~\textrm{for all}~e \in E.
\end{equation}
Let $\rv B\subseteq E$ be sampled at random by
including each $e\in E$ in $\rv B$ with probability $\gamma$
(independently of the other edges).
If   $n\geq
N_{\ref{PR:vertex-upper-bound-flip}}(k,r,\alpha,\beta,\gamma,\eta)$
and $d=rk\geq 1$, then
\[
\Pp_{\rv B}\sqbr*{\frac{|V(\rv B)|}{n}\leq  \sum_{i=1}^m
    u_i(1-(1-\gamma)^{d_i})+A_{\ref{PR:vertex-upper-bound-flip}}(k,r)
n^{-\eta }}
\geq
1-B_{\ref{PR:vertex-upper-bound-flip}}(k,r,\gamma)
n^{-\mu}.
\]
Moreover, we can take
\begin{align*}
A_{\ref{PR:vertex-upper-bound-flip}}(k,r)\qquad\qquad~ &~=~
A_{\ref{PR:vertex-upper-bound-i}}(k,r)+d=2kd+d=O(k^2r),
\\
B_{\ref{PR:vertex-upper-bound-flip}}(k,r,\gamma)\qquad\quad~~ &~=~
\max_{\lambda\in[\frac{\gamma}{2},\frac{1}{2}+\frac{\gamma}{2}]}
B_{\ref{PR:vertex-upper-bound-i}}(k,r,\lambda) + 2 =
\frac{8C_{\ref{cor:thm0}}d}{e^2\gamma(\ln(1-0.5\gamma))^2}+2
=
O(kr\gamma^{-3}),
\\
N_{\ref{PR:vertex-upper-bound-flip}}(k,r,\alpha,\beta,\gamma,\eta)
&~=~\max_{\lambda\in[\frac{\gamma}{2},\frac{1}{2}+\frac{\gamma}{2}]}
\left\{
    N_{\ref{PR:vertex-upper-bound-i}}(k,r,\alpha,\lambda),
    \parens*{\frac{2}{\gamma(1-\gamma)}}^{\eta^{-1}},
    \max\left\{
        \frac{16}{(1 -2\eta)^2},\mu k
    \right\}^{2(1 -2\eta)^{-1}}
\right\}
\\
&
\le
~\quad~\max
\left\{
    (4dk)^{(1-2\alpha)^{-1}},
    \parens*{\frac{8k}{\gamma(1-\gamma)}}^{(1-2\alpha)^{-1}},
    \parens*{\frac{2}{\gamma(1-\gamma)}}^{\eta^{-1}},
    \max\left\{
        \frac{16}{(1 -2\eta)^2},\mu k
    \right\}^{2(1 -2\eta)^{-1}}
\right\}
.
\end{align*}
%
%
\end{proposition}

\begin{proof}
Again, we choose
$A_{\ref{PR:vertex-upper-bound-flip}},
B_{\ref{PR:vertex-upper-bound-flip}},
N_{\ref{PR:vertex-upper-bound-flip}}$
as in the proposition.  Suppose $n\geq
N_{\ref{PR:vertex-upper-bound-flip}}$.
This means in particular that
$n\geq N_{\ref{PR:vertex-upper-bound-i}}(k,\alpha,\lambda,\eta)$
for every $\lambda\in [\frac{\gamma}{2},\frac{1}{2}+\frac{\gamma}{2}]$.

Fix some $\lambda\in [\frac{\gamma}{2},\frac{1}{2}+\frac{\gamma}{2}]$ such
that $\lambda|E|$ is an integer. Then once conditioning on the event
$\Abs*{ \rv B}=\lambda|E|$,
the set $\rv B$ distributes uniformly on the subsets of $E$ with
density $\lambda$.
Therefore, by \cref{PR:vertex-upper-bound-i},
\begin{equation}\label{EQ:PR:vertex-upper-bound-flip:ineq-ii}
\Pp_{\rv B}\sqbr*{\frac{|V(\rv B)|}{n}\leq  \sum_{i=1}^m
    u_i(1-(1-\lambda)^{d_i})+A_{\ref{PR:vertex-upper-bound-i}}(r,k)
    n^{-\eta }
\,\Big|\, \Abs*{ \rv B}=\lambda |E|}
\geq 1-B_{\ref{PR:vertex-upper-bound-i}}(r,k,\lambda) n^{-\mu}.
\end{equation}
In addition, \cref{LM:power-dist}    implies
that for all $i\in \{1,\dots,m\}$,
\[
|(1-(1-\lambda)^{d_i})-(1-(1-\gamma)^{d_i})|=|(1-\lambda)^{d_i}-(1-\gamma)^{d_i}|
\leq d_i |(1-\lambda)-(1-\gamma)| \leq d_i |\lambda-\gamma|.
\]
Thus,
\begin{equation}\label{EQ:PR:vertex-upper-bound-flip:ineq-i}
\sum_{i=1}^m u_i(1-(1-\lambda)^{d_i})\leq
\sum_{i=1}^m u_i(1-(1-\gamma)^{d_i})+\sum_{i=1}^m u_i
d_i|\lambda-\gamma|
=
\sum_{i=1}^m u_i(1-(1-\gamma)^{d_i})+  d|\lambda-\gamma|
.
\end{equation}

Fix some $\tau\in (0,\frac{1}{2})$, to be chosen later, and put
$\varepsilon=n^{-\tau}$.
We will assume that $n$ is big enough so that
\begin{equation}\label{EQ:PR:vertex-upper-bound-flip:cond-on-n}
\varepsilon= n^{-\tau}\leq \frac{\gamma(1-\gamma)}{2}
\qquad\text{and}\qquad
\exp\left( \frac{ n^{1-2\tau}}{k } \right)\geq    n^\mu.
\end{equation}
This first inequality implies that $\ee\leq
\min\{\frac{\gamma}{2},\frac{1-\gamma}{2}\}$,
and therefore $[\gamma-\ee,\gamma+\ee]\subseteq
[\frac{\gamma}{2},\frac{1}{2}+\frac{\gamma}{2}]$.

By the Chernoff bound (\cref{fac:chernoff}),
we have
\begin{align*}
\Pp_{\rv B}\sqbr*{|\Abs*{\rv B} -\gamma rn|> \varepsilon rn}
\leq
2 \exp\left(-rn\cdot\ee^2 \right)
=
2 \exp\left(-rn^{1-2\tau} \right)
\leq 2\exp\left(-\frac{ n^{1-2\tau}}{ k } \right)
\leq  2  n^{-\mu}.
\end{align*}
Combining this with \cref{EQ:PR:vertex-upper-bound-flip:ineq-ii}
and \cref{EQ:PR:vertex-upper-bound-flip:ineq-i} and noting
that $I:=\frac{1}{|E|} \mathbb{N} \cap [\gamma-\varepsilon ,
\gamma+\varepsilon]\subseteq
[\frac{\gamma}{2},\frac{1}{2}+\frac{\gamma}{2}]$, we get
\begin{align*}
\Pp_{\rv B} & \sqbr*{\frac{|V(\rv B)|}{n}\leq  \sum_{i=1}^m
    u_i(1-(1-\gamma)^{d_i})+A_{\ref{PR:vertex-upper-bound-i}}(r,k)
    n^{-\eta}
+dn^{ -\tau}}
\\
&\geq
\sum_{\lambda\in I} \Pp_{\rv B}\sqbr*{\frac{|V(\rv B)|}{n}\leq
    \sum_{i=1}^m
    u_i(1-(1-\gamma)^{d_i})+A_{\ref{PR:vertex-upper-bound-i}}(r,k)
    n^{-\eta }
+d\varepsilon  \,\Big|\, |\rv B|=\lambda|E|}\Pp_{\rv
B}\sqbr*{|\rv B|=\lambda|E|}
\\
&\geq
\sum_{\lambda\in I} \Pp_{\rv B}\sqbr*{\frac{|V(\rv B)|}{n}\leq
    \sum_{i=1}^m
    u_i(1-(1-\lambda)^{d_i})+A_{\ref{PR:vertex-upper-bound-i}}(r,k)
    n^{-\eta}\,\Big|\,
|\rv B|=\lambda|E|} \cdot \Pp_{\rv B}\sqbr*{|\rv B|=\lambda|E|}
\\
&\geq
\sum_{\lambda\in I}
\left(
    1-B_{\ref{PR:vertex-upper-bound-i}}(r,k,\lambda)
n^{-\mu}\right)
\Pp_{\rv B}\sqbr*{|\rv B|=\lambda|E|}
\\
&\geq
\sum_{\lambda\in I}
\left(
    1-(B_{\ref{PR:vertex-upper-bound-flip}}(r,k,\gamma)-2)
n^{-\mu}\right)
\Pp_{\rv B}\sqbr*{|\rv B|=\lambda|E|}
\\
&=
\left(
    1-(B_{\ref{PR:vertex-upper-bound-flip}}(r,k,\gamma)-2)
n^{-\mu}\right)
\parens*{1-\Pp_{\rv B}\sqbr*{|\Abs*{\rv B} -\gamma rn|> \varepsilon r n)}}
\\
&\geq
\left(
1-(B_{\ref{PR:vertex-upper-bound-flip}}(r,k,\gamma)-2) n^{-\mu}\right)
(1-  2n^{-\mu})
\\
&\geq 1-
B_{\ref{PR:vertex-upper-bound-flip}}(r,k,\gamma)  n^{-\mu}
\end{align*}

At this point, we choose
$\tau=\eta$;
note that   $\tau\in (0,\frac{1}{2})$ because we assume
that $\eta\in (0,\frac{1}{2} )$.
What we have shown in the last paragraph now simplifies to
\begin{equation}\label{EQ:PR:vertex-upper-bound-flip:final}
\Pp_{\rv B}   \sqbr*{\frac{|V(\rv B)|}{n}\leq  \sum_{i=1}^m
u_i(1-(1-\gamma)^{d_i})+(A_{\ref{PR:vertex-upper-bound-i}}(r,k
)+d)n^{-\eta}
}\geq 1-B_{\ref{PR:vertex-upper-bound-flip}}(r,k,\gamma) n^{-\mu}.
\end{equation}
This proves the theorem, provided that
\eqref{EQ:PR:vertex-upper-bound-flip:cond-on-n} holds
for $\tau= \eta$.

We finish by verifying that
\cref{EQ:PR:vertex-upper-bound-flip:cond-on-n} is true
when $\tau= \eta$ and $n\geq N_{\ref{PR:vertex-upper-bound-flip}}$.
The first statement of \cref{EQ:PR:vertex-upper-bound-flip:cond-on-n}
follows from the fact that $N_{\ref{PR:vertex-upper-bound-flip}}\geq
\parens*{\frac{2}{\gamma(1-\gamma)}}^{ \eta^{-1}}$. The
second statement
is equivalent to $\frac{n^{1 -2\eta}}{k}-\mu\ln n\geq 0$.
By Lemma~\ref{LM:log-vs-poly} and the fact that
\[N_{\ref{PR:vertex-upper-bound-flip}}\geq
\max\left\{
\frac{16}{(1 -2\eta)^2},\mu k
\right\}^{2(1 -2\eta)^{-1}}=
\max\left\{
\parens*{\frac{2}{0.5(1 -2\eta)}}^{\frac{2}{0.5(1 -2\eta)}},
(\mu k)^{\frac{2}{ 1 -2\eta }}
\right\},\]
we have
\[
\frac{n^{1 -2\eta}}{k}-\mu\ln n
\geq
\frac{n^{1 -2\eta}}{k} -\mu n^{0.5(1 -2\eta)}
=\frac{n^{0.5(1 -2\eta)}}{k}
\parens*{n^{0.5(1 -2\eta)}-k\mu}\geq 0,
\]
as required.
%
\end{proof}

Next, we bound the average number of edges in $E(V(\rv B))$ when
$\rv B$ is chosen as in \cref{PR:vertex-upper-bound-flip}.
It is here that we shall need \cref{PR:meeting-probability}.

\begin{proposition}\label{PR:edges-fraction-avg}
With notation as above, suppose that
the hypergraph $H=(V,E)$ satisfies
\begin{equation}\label{eq:extra_assumptions-edges-avg}\deg(v) \geq
2~~\textrm{for all}~~v \in V\quad\tand\quad
\sum_{v\in \Vrt(e)}(\deg(v)-1) <|E|~~\textrm{for all}~~e \in E.
\end{equation}
Let $\gamma\in (0,1)$ and let
$\rv B\subseteq E$ be sampled at random by
including each $e\in E$ in $\rv B$ with probability $\gamma$
(independently of the other edges).
Then
\[\Exp\sqbr*{\frac{|E(V(\rv B))|}{|E|}}
\geq
\gamma + (1-\gamma)\prod_{i=1}^m(1-(1-\gamma)^{d_i-1})^{\frac{u_i d_i}{r}}.
\]
\end{proposition}

\begin{proof}
For every $e\in E$, define a random variable $\rv X_e$ by
\[
\rv X_e=\left\{
\begin{array}{ll}
    1 & e\in E(V(\rv B)) \\
    0 & e\notin E(V(\rv B)).
\end{array}
\right.
\]
Then
\begin{equation}\label{EQ:PR:edges-fraction-avg:eq1}
\Exp\sqbr*{\frac{|E(V(\rv B))|}{|E|}}=\frac{1}{|E|}\sum_{e\in
E}\Exp \rv X_e.
\end{equation}
For every $e\in E$, we have
\begin{align*}
\Exp \rv X_e &=
\Exp[ \rv X_e \,|\,e\in B]\Pp[e\in \rv B]+\Exp[\rv X_e \,|\,e\notin
\rv B]\Pp[e\notin \rv B]
\\
&=1\cdot \gamma+\Pp[e\in E(V(\rv B))\,|\,e\notin B]\cdot (1-\gamma).
\end{align*}
By substituting this into \cref{EQ:PR:edges-fraction-avg:eq1}, we get
\begin{align*}
\Exp\sqbr*{\frac{|E(V(\rv B))|}{|E|}}
=\gamma +(1-\gamma)\frac{1}{|E|}\sum_{e\in E} \Pp[e\in E(V(\rv
B))\,|\,e\notin \rv B].
\end{align*}
Observe that by Jensen's inequality,
\[
\ln \parens*{\frac{1}{|E|}\sum_{e\in E} \Pp[e\in E(V(\rv
B))\,|\,e\notin \rv B]}
\geq \frac{1}{|E|}\sum_{e\in E} \ln \Pp[e\in E(V(\rv
B))\,|\,e\notin \rv B].
\]
Thus, in order to prove the proposition, it is enough to show that
\begin{align}\label{EQ:PR:edges-fraction-avg:what-we-need}
\frac{1}{|E|}\sum_{e\in E} \ln \Pp[e\in E(V(\rv B))\,|\,e\notin \rv B]
& \geq \ln
\parens*{\prod_{i=1}^m(1-(1-\gamma)^{d_i-1})^{\frac{u_i d_i}{r}}}
\\
&= \frac{1}{r}\sum_{i=1}^m u_id_i\ln (1-(1-\gamma)^{d_i-1}). \nonumber
\end{align}
(Notice that we need $d_i>1$  for the right hand side to be defined,
and this holds
by the first part of \cref{eq:extra_assumptions-edges-avg}.)

Fix some $e\in E$. For every $v\in \Vrt(e)$, let $E_v$ denote the set
of edges in $H$ having $v$ as a vertex. Then $|E_v|=\deg(v)$.
Moreover, provided that $e\notin B$,
we have $e\in E(V(B))$ if and only if $B$ meets the set
$E_v-\{e\}$ for every $v<e$. This means that
\[
\Pp[e\in E(V(\rv B))\,|\,e\notin \rv B] = \Pp[\text{$\rv B$ meets
$E_v-\{e\}$ for all $v<e$} \,|\, e\notin \rv B]
=:(\star).
\]
Conditioning on the event $e\notin\rv B$ simply means
that $\rv B$
is chosen as a random subset of $E-\{e\}$
by adding each edge in $E-\{e\}$ to $\rv B$ with probability $\gamma$.
In addition, by the second part of \cref{eq:extra_assumptions-edges-avg},
$\sum_{v<e}|E_v-\{e\}|=\sum_{v<e} (\deg(v)-1)\leq |E|-1=|E-\{e\}|$.
We may therefore choose pairwise disjoint subsets
$\{A'_v\}_{v<e}$ of $E-\{e\}$
such that $|A'_v|=|E_v-\{e\}|=\deg(v)-1$ for all $v<e$.
We are now in position to apply \cref{PR:meeting-probability}
to the set $X=E-\{e\}$ and the collections of subsets $\{E_v-\{e\}\}_{v<e}$
and $\{A'_v\}_{v<e}$. According to that proposition,
\begin{align}\label{EQ:PR:edges-fraction-avg:star-bound}
(\star)&\geq
\Pp[\text{$\rv B$ meets $A'_v$ for all $v<e$}\,|\,e\notin\rv B]
=\prod_{v<e}\Pp[\rv B\cap A'_v=\emptyset]
\\
&=\prod_{v<e}(1-(1-\gamma)^{|A'_v|})=\prod_{v<e}(1-(1-\gamma)^{\deg(v)-1}),
\nonumber
\end{align}
where the first equality holds because the sets $\{A'_v\}_{v<e}$
are pairwise disjoint. Thus,
\begin{align*}
\frac{1}{|E|}\sum_{e\in E} \ln \Pp[e\in E(V(\rv B))\,|\,e\notin \rv B]
&\geq
\frac{1}{|E|}\sum_{e\in E} \ln
\parens*{\prod_{v<e}(1-(1-\gamma)^{\deg(v)-1})}
\\
&
=
\frac{1}{r n}\sum_{e\in E} \sum_{v<e} \ln  (1-(1-\gamma)^{\deg(v)-1})
\\
&
=
\frac{1}{rn}\sum_{v\in V}\sum_{e>v} \ln  (1-(1-\gamma)^{\deg(v)-1})
\\
&
=
\frac{1}{rn}\sum_{v\in V}\deg(v) \ln  (1-(1-\gamma)^{\deg(v)-1})
\\
&=
\frac{1}{rn}\sum_{i=1}^m (u_i n)d_i \ln(1-(1-\gamma)^{d_i-1})
\\
&=\frac{1}{r} \sum_{i=1}^m  u_i d_i \ln(1-(1-\gamma)^{d_i-1}).
\end{align*}
This proves \cref{EQ:PR:edges-fraction-avg:what-we-need} and hence
the proposition.
\end{proof}

\begin{remark}\label{RM:not-simplicial}
The proof of \cref{PR:edges-fraction-avg}
suggests that $\Exp\sqbr*{\frac{|E(V(\rv B))|}{|E|}}$ would be the smallest
when every two hyperedges in $H$ share at most one vertex.
Indeed, in this situation, for every $e\in E$, the sets
$\{E_v-\{e\}\}_{v<e}$
are pairwise disjoint, meaning that the lower bound on $(\star)$, see
\cref{EQ:PR:edges-fraction-avg:star-bound}, is actually an equality.
Since \cref{PR:edges-fraction-avg} is a step in proving our lower bounds on
confiners, this suggests that
the best confiners
are those in which every two hyperedges share
at most one vertex.
\end{remark}

\begin{remark}\label{RM:edges-frac-avg-regular}
When $H$ is a regular, the conclusion
of \cref{PR:edges-fraction-avg} simplifies
to
\[\Exp\sqbr*{\frac{|E(V(\rv B))|}{|E|}}
\geq \gamma+(1-\gamma)(1-(1-\gamma)^{d-1})^k.\]
Unfortunately, this simplified conclusion,
which is equivalent
to $\sum_{i=1}^m u_id_i\ln(1-(1-\gamma)^{d_i-1})\geq
d\ln((1-(1-\gamma)^{d-1})$, may  fail for general $H$,
because the function  $x\mapsto x\ln(1-(1-\gamma)^{x-1})$ need
not be concave up in an interval containing $d_1,\dots,d_m$.
\end{remark}

We shall use the following lemma to apply
\cref{PR:vertex-upper-bound-flip}
and~\cref{PR:edges-fraction-avg} together.

\begin{lemma}\label{LM:Markov-modified}
Let $\rv X$ be a random variable taking  values
in $[0,1]$. Put $a=\Exp \rv X$ and let $S$ be an event
with $\Pp(S)\geq 1-b$, for some $0\leq  b \leq 1$.
Then $\Pr\sqbr*{(\rv X\geq  a-2b) \wedge S}>0$.
\end{lemma}

\begin{proof}
This can be deduced from
Markov's Inequality, but
giving a direct proof is shorter.

The lemma holds trivially if $b\geq \frac{1}{2}$ or $b=0$,
so we may assume that $0<b<\frac{1}{2}$.
Moreover, decreasing $b$ can only
decrease  $\Pr\sqbr*{(\rv X > a-2b) \wedge S}$, so it is enough
to prove the lemma for $b=1-\Pp(S)$.
Suppose for the sake of contradiction that
$\Pr\sqbr*{(\rv X \geq a-2b) \wedge S}=0$.
Then
\begin{align*}
    a&=\Exp \rv X =
    \Exp\sqbr*{\rv X\,|\,S}\Pp(S)+
    \Exp\sqbr*{\rv X\,|\,\neg S}\Pp(\neg S)
    \\
    &\leq (a-2b)(1-b)+1\cdot b =a-2b-ab+2b^2+b
    =a-ab +b(2b-1).
\end{align*}
Since $0<b< \frac{1}{2}$, the final expression is strictly smaller than
$a$, which is absurd.
\end{proof}


Fix some $\delta\in (0,1)$.
Recall that our goal is to apply \cref{PR:vertex-upper-bound-flip}
with the parameter $\gamma$ chosen so that
the density of $V(\rv B)$ is at most $\delta$,
i.e., we need $\delta= \sum_{i=1} u_i(1-(1-\gamma)^{d_i})$ (up to
additional negligible terms).
\cref{PR:edges-fraction-avg} then provides
us with a lower bound on $\Exp\sqbr*{\frac{|E(V(\rv B))|}{|E|}}$
given \emph{in terms of $\gamma$} and other parameters of the
hypergraph $H$.
For example, when $H$ is regular, we have $m=1$ and $d_1=d$, so we can solve
$\delta=1-(1-\gamma)^d$ for $\gamma$ and get
$\gamma=1-(1-\delta)^{\frac{1}{d}}$.
Substituting this into \cref{PR:edges-fraction-avg} gives
$\Exp\sqbr*{\frac{|E(V(\rv B))|}{|E|}}\geq
\gamma+(1-\gamma)(1-(1-\gamma)^{d-1})^k=f_{k,r}(\delta)$,
where $f_{k,r}(x)$ is as in \cref{EQ:f-k-r-delta-dfn}, and
this is actually enough to prove \cref{TH:lower-bound-strong}(ii)
in the regular case. To treat the case of a general hypergraph,
we use the following theorem, which says that
$f_{k,r}(\delta)$ is a lower bound for
$\gamma + (1-\gamma)\prod_{i=1}^m(1-(1-\gamma)^{d_i-1})^{\frac{u_i d_i}{r}}$
whenever the minimum vertex degree of $H$ is at least $3$.
The very technical proof of this result is given in
\cref{subsec:proof-of-hard-optimization}.

\begin{theorem}\label{TH:hard-optimization}
Suppose we are given
\begin{itemize}
    \item $u_1,\dots,u_m\in [0,1]$   such that $\sum_{i=1}^m u_i=1$,
    \item $k\in [1,\infty)$ and
    \item $\delta\in (0,1)$.
\end{itemize}
For every $x_1,\dots,x_m\in (1,\infty)$, let
\begin{itemize}
    \item $d=\sum_{i=1}^m x_iu_i$,
    \item $r=\frac{d}{k}$,
    \item $\gamma$ be the unique element
        of $(0,1)$ such that\footnote{
            The existence and uniqueness of $\gamma$
            is an easy consequence of applying the Mean Value Theorem
            to the increasing function
            $f(t) =u_i(1-(1-t)^{x_i})$.
        }
        \begin{equation}\label{EQ:gamma-dfn}
            \delta=\sum_{i=1}^m u_i(1-(1-\gamma)^{x_i}),
        \end{equation}
\end{itemize}
and put
\[
    F(x_1,\dots,x_m)=\gamma+(1-\gamma)\prod_{i=1}^m(1-(1-\gamma)^{x_i-1})^{\frac{u_i
    x_i}{r}}.
\]
Let
$f_{k,r}(x)=1-(1-x)^{\frac{1}{rk}}+(1-x)^{\frac{1}{rk}}(1-(1-x)^{\frac{rk-1}{rk}})^k$
be as in \cref{EQ:f-k-r-delta-dfn}.
If $x_1,\dots,x_m\geq 3$, then
\[
    F(x_1,\dots,x_m)\geq
    f_{k,r}(\delta)
\]
and
equality holds if and only if $x_1=\dots=x_m=d$.
\end{theorem}

\begin{remark}
(i)
By straightforward computation,
$F(d,\dots,d)=f_{k,r}(\delta)$ even without assuming
$x_1,\dots,x_m\geq 3$.

(ii) We conjecture that \cref{TH:hard-optimization} continues
to hold if we
replace the assumption $x_1,\dots,x_m\geq 3$ with  $x_1,\dots,x_n\geq
2$. If that were true,
then \cref{TH:lower-bound-strong}(ii) would also hold for
non-regular hypergraphs having minimum vertex degree $ 2$.
\end{remark}


We shall also need the following two lemmas.

\begin{lemma}\label{LM:gamma-bounds}
With notation as in \cref{TH:hard-optimization},
if $x_1,\dots,x_m\geq c$ for some $c\geq 1$, then
\[1-(1-\delta)^{\frac{1}{d}}\leq \gamma\leq
1-(1-\delta)^{\frac{1}{c}}.\]
\end{lemma}

\begin{proof}
Recall that $\delta = \sum_i u_i(1-(1-\gamma)^{x_i})$.
Since $x_i\geq c$ for all $i$,
\[
    \delta \geq \sum_i u_i(1-(1-\gamma)^{c})=1-(1-\gamma)^c,
\]
and rearranging gives $\gamma\leq 1-(1-\delta)^{\frac{1}{c}}$.

Next, the function $x\mapsto 1-(1-\gamma)^x$ is concave down on
$[0,\infty)$,
so by Jensen's inequality,
\[
    \delta=\sum_i u_i(1-(1-\gamma)^{x_i})\leq
    1-(1-\gamma)^{\sum_i u_i x_i}=1-(1-\gamma)^d.
\]
Rearranging gives $\gamma\geq 1-(1-\delta)^{\frac{1}{d}}$.
\end{proof}

We postpone the proof of
the following lemma
to \cref{subsec:properties-of-f-k-r-delta}.

\begin{lemma}\label{LM:delta-change}
Let $k,r>0$ be     such that $rk\geq 1$ and let
$f_{k,r}$ be as in \cref{EQ:f-k-r-delta-dfn}.
Then for every $x,x'\in [0,1]$, we have
\[
    |f_{k,r}(x)-f_{k,r}(x')|\leq k|x-x'|.
\]
\end{lemma}

We are   now in position to put the previous results together and
prove a version of
\cref{TH:lower-bound-strong}(ii) holding under the assumption
that the maximum vertex degree and the maximum uniformity
of the hypergraph $H$ are not too large.

\begin{theorem}\label{TH:pre-lower-bound-i}
Let $\alpha\in [0,\frac{1}{2})$, $\beta\in[0,1)$, $\delta\in
(0,1)$, $\eta\in
(0,\frac{1}{2} )$ be real numbers such that
\[
    \mu:=( 1-\alpha -\beta-\eta)(1-2\eta)-\alpha>0.
\]
Then there are   constants
$A_{\ref{TH:pre-lower-bound-i}}(k,r),
B_{\ref{TH:pre-lower-bound-i}}(k,r,\delta),
N_{\ref{TH:pre-lower-bound-i}}(k,r,\delta,\alpha,\beta,\eta)>0$
(depending only on the indicated parameters)
such that the following hold:
With notation as in the beginning of
\cref{subsec:lower-bound-with-assumpt}, suppose that the
hypergraph $H$ satisfies
\begin{equation}\label{eq:extra_assumptions-pre-lower-bound}
    \deg(v) \leq d n^\alpha~\textrm{for all}~v \in V\quad\tand\quad
    |e|\leq k n^\beta~\textrm{for all}~e \in E
\end{equation}
and also one of the following:
\begin{enumerate}[(1)]
    \item $\deg(v)\geq 3$ for all $v\in V$, or
    \item $H$ is regular.
\end{enumerate}
If $n\geq N_{\ref{TH:pre-lower-bound-i}}(k,r,\delta,\alpha,\beta,\eta)$
and $d=rk\geq 1$, then there is a set of vertices $A\subseteq V$
of density at most $\delta$
such that
\[
    \frac{|E(A)|}{|E|}
    \geq
    f_{k,r}(\delta)-
    A_{\ref{TH:pre-lower-bound-i}}(k,r) n^{-\eta}-
    B_{\ref{TH:pre-lower-bound-i}}(k,r,\delta) n^{-\mu}
\]
where
$f_{k,r}$ is as in \cref{EQ:f-k-r-delta-dfn}.
Moreover, we can take
\begin{align*}
    &A_{\ref{TH:pre-lower-bound-i}}(k,r ) =
    kA_{\ref{PR:vertex-upper-bound-flip}}(k,r ) = 2dk^2+dk=
    O(k^3r),
    \\
    & B_{\ref{TH:pre-lower-bound-i}}(k,r,\delta)  =
    2\max_{\frac{\delta}{2d}\leq \gamma\leq \delta}
    B_{\ref{PR:vertex-upper-bound-flip}}(k,r,\gamma) \leq
    \frac{512 C_{\ref{cor:thm0}}d^4}{e^2\delta^3}+2
    =O(k^4r^4\delta^{-3}),
    \\
    & N_{\ref{TH:pre-lower-bound-i}}(k,r,\delta,\alpha,\beta,\eta)
    =
    \max_{\frac{\delta}{2d}\leq \gamma\leq \delta}
    \left\{
        N_{\ref{PR:vertex-upper-bound-flip}}(k,r,\alpha,\gamma,\beta,\eta),
        \parens*{\frac{2A_{\ref{PR:vertex-upper-bound-flip}}(k,r)}{\delta}}^{\eta^{-1}},
        k^{2(1-\alpha-\beta)^{-1}}
    \right\}
    \\
    & \qquad \leq \max\{N_1,N_2,N_3,N_4,N_5 \},
\end{align*}
where:
\begin{align*}
    N_1 &= \parens*{\frac{16dk}{\delta(1-\delta)}}^{(1-2\alpha)^{-1}} &
    N_2 &= \parens*{\frac{4d}{\delta(1-\delta)}}^{ \eta^{-1}} &
    N_3 &= \max\left\{
        \frac{16}{(1- 2\eta)^2},\mu k
    \right\}^{2(1- 2\eta)^{-1}} \\
    N_4 &= \parens*{\frac{4dk+2d}{\delta}}^{\eta^{-1}} &
    N_5 &=k^{2(1-\alpha-\beta)^{-1}}
\end{align*}
\end{theorem}

\begin{proof}
Suppose $n\geq N_{\ref{TH:pre-lower-bound-i}}$.
Put $\ee = A_{\ref{PR:vertex-upper-bound-flip}}(k,r ) n^{-\eta}$ and let
\[\tilde{\delta}=\delta-\ee.\]
Since $n\geq
\parens*{\frac{2A_{\ref{PR:vertex-upper-bound-flip}}(k,r)}{\delta}}^{\eta^{-1}}$,
we have $\ee\leq \frac{\delta}{2}$, and hence
\begin{equation} \label{EQ:TH:pre-lower-bound-i:N-bound-ii}
    \tilde\delta\geq \frac{1}{2}\delta.
\end{equation}

As noted earlier, there is a unique $\gamma\in (0,1)$
such that
\[
    \tilde{\delta}=\sum_{i=1}^m u_i (1-(1-\gamma)^{d_i}).
\]
Assumptions (1) and (2) imply that $d_i\geq 1$ for all $i$, and
we have $\sum_{i=1}^m u_i d_i=d=rk$.
Thus, by
\cref{LM:gamma-bounds}
(applied with $\tilde{\delta}$ in place of $\delta$
and $(x_1,\dots,x_m)=(d_1,\dots,d_m)$),
we have
\[
    \delta\geq \tilde{\delta}=1-(1-\tilde{\delta})^{\frac{1}{1}}
    \geq
    \gamma\geq 1-(1-\tilde{\delta})^{\frac{1}{d}}.
\]
By \cref{EQ:TH:pre-lower-bound-i:N-bound-ii} and
\cref{LM:power-ineq}, this means that
\[
    1-(1-\tilde{\delta})^{\frac{1}{d}}\geq
    \frac{\tilde{\delta}}{d}
    \geq \frac{\delta}{2d}.
\]
Together, this implies
\begin{equation}\label{EQ:TH:pre-lower-bound-i:gamma-bounds}
    \frac{\delta}{2 d} \leq \gamma\leq \delta .
\end{equation}


%

Choose $\rv B\subseteq E$ at random by adding each $e\in E$ to $\rv B$
with probability $\gamma$ and put $\rv A=V(\rv B)$.
By \eqref{EQ:TH:pre-lower-bound-i:gamma-bounds} and the definition
of $N_{\ref{TH:pre-lower-bound-i}}$,
we have  $n\geq
N_{\ref{PR:vertex-upper-bound-flip}}(k,r,\alpha,\beta,\gamma,\eta)$.
Thus, \cref{PR:vertex-upper-bound-flip} tells us that
\begin{align}\label{EQ:TH:pre-lower-bound-i:density-prob}
    \Pp\sqbr*{\frac{|\rv A|}{|V|}\leq \delta}
    &=
    \Pp\sqbr*{\frac{|V(\rv B)|}{n} \leq \tilde{\delta} + \ee}
    \\
    &=
    \Pp\sqbr*{\frac{|V(\rv B)|}{n} \leq \sum_{i=1}^m u_i
        (1-(1-\gamma)^{d_i}) +
    A_{\ref{PR:vertex-upper-bound-flip}}(k,r) n^{-\eta}}
    \nonumber
    \\
    &\geq
    1-B_{\ref{PR:vertex-upper-bound-flip}}(k,r,\gamma) n^{-\mu}
    \nonumber
    \\
    &\geq
    1-\frac{1}{2} B_{\ref{TH:pre-lower-bound-i}}(k,r,\delta)
    n^{-\mu}.
    \nonumber
\end{align}

We proceed by applying \cref{PR:edges-fraction-avg} to $H$.
To that end, we must first check that $\sum_{v<e}(\deg(v)-1)<|E|$.
This clear if $|e|=0$, and otherwise,
by \eqref{eq:extra_assumptions-pre-lower-bound},
we have $\sum_{v<e}(\deg(v)-1)<kn^\beta \cdot  dn^{\alpha }=rk^2
n^{\alpha+\beta}$,
and the left hand side cannot exceed $|E|=rn$ because $n\geq
N_{\ref{TH:pre-lower-bound-i}}
\geq k^{2(1-\alpha-\beta)^{-1}}$.
Now that prerequisites of  \cref{PR:edges-fraction-avg} hold,
it follows that
\[
    \Exp\sqbr*{\frac{|E(\rv A)|}{|E|}} =
    \Exp\sqbr*{\frac{|E(V(\rv B))|}{|E|}}
    \geq
    \gamma + (1-\gamma)\prod_{i=1}^m(1-(1-\gamma)^{d_i-1})^{\frac{u_i
    d_i}{r}} =F(d_1,\dots,d_m) ,
\]
where $F(d_1,\dots,d_m)$ is as in
\cref{TH:hard-optimization}
with $\tilde{\delta}$ in place of $\delta$.
Now, if $H$ is regular, then
$F(d_1,\dots,d_m)=F(d,\dots,d)=f_{k,r}(\tilde{\delta})$ (in fact,
$m=1$ in this case),
and if $\deg(v)\geq 3$ for all $v\in V$ (equiv.\ $d_1,\dots,d_m\geq
3$), then  $F(d_1,\dots,d_m)\geq  f_{k,r}(\tilde{\delta})$ by
\cref{TH:hard-optimization}. As a result,
\begin{equation}\label{EQ:TH:pre-lower-bound-i:exp-bound}
    \Exp\sqbr*{\frac{|E(\rv A)|}{|E|}}\geq f_{k,r}(\tilde{\delta}).
\end{equation}

Finally, applying \cref{LM:Markov-modified} with $\rv X =
\frac{|E(\rv A)|}{|E|}$
and $S=\{\frac{|\rv A|}{|V|}\leq \delta\}$, and using
\cref{EQ:TH:pre-lower-bound-i:exp-bound} and
\cref{EQ:TH:pre-lower-bound-i:density-prob}, we get
\[
    \Pp\sqbr*{ \frac{|E(\rv A)|}{|E|}\geq f_{k,r}(\tilde{\delta})-
        B_{\ref{TH:pre-lower-bound-i}}(k,r,\delta) n^{-\mu}
    \tand \frac{|\rv A|}{|V|}\leq \delta}>0.
\]
As a result, there is $B\subseteq E$ such that $A=V(B)$ has
density at most $\delta$ in $V$
and
\[
    \frac{|E(A)|}{|E|}\geq f_{k,r}(\tilde{\delta})-
    B_{\ref{TH:pre-lower-bound-i}}(k,r,\delta)  n^{-\mu}.
\]
By Lemma~\ref{LM:delta-change},
$f_{k,r}(\tilde{\delta})=f_{k,r}(\delta-\ee)\geq
f_{k,r}(\delta)-k\ee$, so
\[
    \frac{|E(A)|}{|E|}\geq f_{k,r}(\delta)-
    kA_{\ref{PR:vertex-upper-bound-flip}}(r,k)
    n^{-\eta}-B_{\ref{TH:pre-lower-bound-i}}(k,r,\delta)  n^{-\mu}.
\]
This proves the theorem.
\end{proof}

\begin{remark}\label{RM:regular_is_best}
An easy modification of the proof of \cref{TH:pre-lower-bound-i}
also shows that
for every $\ee>0$, there is $N$ (depending on
$k,r,\delta,\alpha,\beta,\eta$ and $\ee$)
such that whenever $n\geq N$, there exists $A\subseteq V$ of
density at most $\delta$
satisfying
$\frac{|E(A)|}{|E|}\geq F(d_1,\dots,d_m)-\ee$;
here $F$ is as in \cref{TH:hard-optimization}. This improves the lower
bound $\frac{|E(A)|}{|E|}\geq f_{k,r}(\delta)-o(1)=F(d ,\dots,d )-o(1)$
asserted in the theorem when $H$ is not regular.
It further suggests that the best possible confiners should be regular.
\end{remark}

When $H$ is regular and uniform, we can take $\alpha=\beta=0$ in
\cref{TH:pre-lower-bound-i}. In this case, the theorem simplifies to
the following corollary.

\begin{corollary}\label{CR:lower_bound_for_reg_unif}
Let   $r>0$, $\delta\in (0,1)$ and $\eta\in (0,\frac{1}{2})$ be
real numbers,
let $k\geq 2$ be an integer,
and let $A_{\ref{TH:pre-lower-bound-i}}(k,r)$,
$B_{\ref{TH:pre-lower-bound-i}}(k,r,\delta)$ be as in
\cref{TH:pre-lower-bound-i}.
Then there is a constant
$N(r,k,\delta,\eta)>0$ (depending on the indicated paramters) such
that for any regular $k$-uniform hypergraph $H=(V,E)$ having $n$
vertices and $rn$
edges, there exists $A\subseteq V$ of density at most $\delta$
satisfying
\[
    \frac{|E(A)|}{|A|}\geq
    f_{k,r}(\delta)-A_{\ref{TH:pre-lower-bound-i}}(k,r)n^{-\eta}
    -B_{\ref{TH:pre-lower-bound-i}}(k,r,\delta)n^{-(1-\eta)(1-2\eta)}.
\]
Moreover, we can take
\[
    N(r,k,\delta,\eta)=\max\left\{
        \frac{16rk^2}{\delta(1-\delta)},
        \parens*{\frac{16}{(1- 2\eta)^2}}^{2(1- 2\eta)^{-1}},
        \parens*{\frac{r(4k^2+2k)}{\delta(1-\delta)}}^{\eta^{-1}},
        k^{2(1- 2\eta)^{-1}}
    \right\}.
\]
\end{corollary}

When $k,r,\delta$ are fixed and $n$ grows, choosing
$\eta=1-\frac{1}{\sqrt{2}}=0.2928\dots$ gives
the asymptotically fastest decay
for the term $A_{\ref{TH:pre-lower-bound-i}}(k,r)n^{-\eta}
+B_{\ref{TH:pre-lower-bound-i}}(r,k,\delta)n^{-(1-\eta)(1-2\eta)}$,
which becomes $O(n^{-\eta})$.

\subsection{Proof of \cref{TH:lower-bound-strong}(ii): General Case}
\label{subsec:removing-vrt-edges}

At this point, all that remains to prove
\cref{TH:lower-bound-strong}(ii) is to  eliminate the
assumption~\cref{eq:extra_assumptions-pre-lower-bound}
from \cref{TH:pre-lower-bound-i}.
%
This
can be done by first removing all   vertices of degree greater
than $d n^\alpha$
from the hypergraph and then all hyperdeges having more than $k n^\beta$
vertices. However, this transition changes the average uniformity and
average degree, and could also turn a regular (non-uniform)
hypergraph into an
irregular one. In order to avoid these issues, we modify $H$ in a
slightly different way that we now describe.

\medskip

Throughout, let $H=(V,E)$ denote a hypergraph with average uniformity $k$
and average vertex degree $d$.
The \emph{incidence set} of   $H $ is the set
\[
I(H):=\{(v,e)\in V\times E\,:\, v<e\}.
\]
Its elements are the \emph{incidences} of $H$.
Clearly, $H$ is determined by its incidence set.
Moreover, the average uniformity $k$ and the average degree $d$ of
$H$ can be recovered
from $I(H)$ by
\begin{equation}\label{EQ:incidences-to-k-and-d}
k=\frac{|I(H)|}{|E|}\qquad\tand\qquad d=\frac{|I(H)|}{|V|}.
\end{equation}

Let $D\geq \lceil d\rceil$.  Consider
the following algorithm that modifies $H$ by changing its incidence set:
\begin{itemize}
\item While there is $v\in V$ with $\deg(v)> D$:
    \begin{itemize}
        \item Choose some $e\in E$ with $e>v$.
        \item Choose some $u\in V$ with $\deg(u)< d$.
        \item In $I(H)$, replace the incidence $(v,e)$ with $(u,e)$.
    \end{itemize}
\item Output the modified hypergraph $H$.
\end{itemize}
Observe that the size of $I(H)$ does not change throughout the algorithm,
and therefore $d$ and $k$ do not change as well, by
\cref{EQ:incidences-to-k-and-d}. Note also that the choice of $u$
with $\dim u<d$ is always possible, because $\deg(v)>d$
and $d$ is the average vertex degree.
Finally, for every edge $e\in E$, its size $|e|$ is unaffected by
the algorithm.

We call an output $H'$ of the algorithm a \emph{vertex $D$-rewiring} of $H$,
because it rearranges the incidences of $H$
to make the maximum degree at most $D$.

\begin{lemma}\label{LM:vertex-flatification}
Let $H'$ be a vertex $D$-rewiring of
$H$ ($D\geq \lceil d\rceil$),
and let $V_0=\{v\in V\,:\, \deg(v)>D\}$.
Then for every $A\subseteq V$,
\[
    E_H(A\cup V_0)\supseteq E_{H'}(A).
\]
Furthermore, $|V_0|<  \frac{d}{D}|V|$, and for every $v\in V-V_0$,
we have $\deg_{H'}(v)\geq \deg_H(v)$.
\end{lemma}

\begin{proof}
If an incidence $(v,e)$ is removed from $I(H)$ in the process
of modifying it to $H'$, then we must have $\deg_H(v)>D$.
(Indeed, $(v,e)$ cannot be one of the incidences $(u,e)$
    that were added to $I(H)$ during the algorithm, because after adding
    $(u,e)$, we still have
$\deg(u)\leq \lceil d\rceil\leq D$.)
This means that for every $e\in E$, we have $\Vrt_H(e)\subseteq
\Vrt_{H'}(e)\cup V_0$,
and it follows that $E_H(A\cup V_0)\supseteq E_{H'}(A)$.
It also means that $\deg_{H'}(v)\geq \deg_H(v)$ if
$\deg_H(v)\leq D$, i.e.,
if $v\in V-V_0$.
Finally, if it were not the case that  $|V_0|< \frac{d}{D}|V|$,
then the average degree of $H$ would be greater than
$\frac{D|V_0|}{|V|}\geq D\cdot\frac{d}{D}=d$,
which contradicts our assumptions.
\end{proof}

Next,   let $K\geq \lceil k\rceil$, and consider
the following algorithm that modifies $H$ by changing its incidence set:
\begin{itemize}
\item While there is $e\in E$ with $|e|> K$:
    \begin{itemize}
        \item Choose some $v\in V$ with $v<e$.
        \item Choose some $e'\in E$ with $|e'|< k$.
        \item In $I(H)$, replace the incidence $(v,e)$ with $(v,e')$.
    \end{itemize}
\item Output the modified hypergraph $H$.
\end{itemize}
Again, the average uniformity $k$ and the average degree $d$ do not
change throughout
the algorithm. Moreover,   $\deg(v)$ remains unchanged for all $v
\in V$,  so if $H$ is
regular, the output would be regular as well.

We call an output $H'$ of the algorithm an \emph{edge $K$-rewiring} of $H$.
This is the same as saying that the dual hypergraph $(H')^*$
(q.v.\ \cref{subsec:hypergraphs}) is a
vertex $K$-rewiring of $H^*$.

\begin{lemma}\label{LM:edge-flatification}
Let $H'$ be an edge $K$-rewiring of $H$ ($K\geq \lceil k\rceil$),
and let $E_0=\{e\in V\,:\, |e|>K\}$.
Then for every $A\subseteq V$,
\[
    E_H(A )\supseteq E_{H'}(A)-E_0.
\]
Furthermore, $|E_0|< \frac{k}{K}|E|$.
\end{lemma}

\begin{proof}
Since $H'^*$ is a vertex $K$-rewiring of $H^*$,
that $|E_0|<\frac{k}{K}|E|$ follows from \cref{LM:vertex-flatification}.
Next, we have seen in the proof of that lemma (applied to $H^*$),
that the only
incidences of $H$ that are removed during the algorithm are of the form
$(v,e)$ with $|e|>K$. Thus, if $e\in E$ is not in $E_0$,
i.e., $|e|\leq K$,
then $\Vrt_H(e)=\Vrt_{H'}(e)$. This means that $E_H(A )\supseteq
E_{H'}(A)-E_0$.
\end{proof}

We are now ready to remove assumption
\cref{eq:extra_assumptions-pre-lower-bound}
from \cref{TH:pre-lower-bound-i}.

\begin{theorem}\label{TH:lower-bound-strong-detailed}
Let $k,r >0$, $\alpha\in [0,\frac{1}{2})$,
$\beta\in[0,1)$, $\delta\in (0,1)$, $\eta\in
(0,\frac{1}{2} )$ be real numbers such that
\[
    \mu:=( 1-\alpha -\beta-\eta)(1-2\eta)-\alpha>0.
\]
Then there are   constants
$
A_{\ref{TH:lower-bound-strong-detailed}}(k,r),
B_{\ref{TH:lower-bound-strong-detailed}}(k,r,\delta),
N_{\ref{TH:lower-bound-strong-detailed}}(k,r,\delta,\alpha,\beta,\eta)>0$
(depending on the indicated parameters) such that the following hold:
Let $H=(V,E)$ be a hypergraph with $n$ vertices, $rn$ edges and
average uniformity
$k$. Suppose further that
\begin{enumerate}[(1)]
    \item $\deg(v)\geq 3$ for all $v\in V$, or
    \item $H$ is regular.
\end{enumerate}
If $n\geq
N_{\ref{TH:lower-bound-strong-detailed}}(k,r,\delta,\alpha,\beta,\eta)$,
then there is a set of vertices $A\subseteq V$
of density at most $\delta$ such that
\[
    \frac{|E(A)|}{|E|}
    \geq
    f_{k,r}(\delta)-
    kn^{-\alpha}-n^{-\beta}-
    A_{\ref{TH:lower-bound-strong-detailed}}(k,r) n^{-\eta}-
    B_{\ref{TH:lower-bound-strong-detailed}}(k,r,\delta) n^{-\mu},
\]
where
$f_{k,r}$ is as in \cref{EQ:f-k-r-delta-dfn}.
Moreover, we can take
$A_{\ref{TH:lower-bound-strong-detailed}}(k,r)=A_{\ref{TH:pre-lower-bound-i}}(k,r)=
(2k^3+k^2)r =O(k^3r)$
and
\begin{align*}
    &
    B_{\ref{TH:lower-bound-strong-detailed}}(k,r,\delta)
    = \max_{\tilde{\delta}\in [\frac{\delta}{2},\delta]}
    B_{\ref{TH:pre-lower-bound-i}}(k,r,\tilde{\delta})
    =\frac{2^{12}C_{\ref{cor:thm0}}}{e^2} k^4 r^4 \delta^{-3}
    \leq 63194 \cdot k^4 r^4\delta^{-3} ,
    \\
    &
    N_{\ref{TH:lower-bound-strong-detailed}}(k,r,\delta,\alpha,\beta,\eta)
    = \max_{\tilde{\delta}\in [\frac{\delta}{2},\delta]}
    \{N_{\ref{TH:pre-lower-bound-i}}(k,r,\tilde{\delta},\alpha,\beta,\eta),
    (2\delta^{-1})^{\alpha^{-1}},2^{\beta^{-1}}\}
    =\poly_{\alpha,\beta,\eta}(k,r,\textstyle{\frac{1}{\delta},\frac{1}{
    1-\delta}}).
\end{align*}
\end{theorem}

\begin{proof}
Suppose $n\geq N_{\ref{TH:lower-bound-strong-detailed}}$, and
put $D=dn^\alpha$ and $K=kn^\beta$.
We claim that $D\geq \lceil d\rceil$ and $K\geq \lceil k\rceil$.
Indeed, $D\geq \lceil d\rceil$ is equivalent to $n\geq
(\frac{\lceil d\rceil}{d})^{\alpha^{-1}}$,
and this holds because $n\geq
N_{\ref{TH:lower-bound-strong-detailed}}\geq 2^{\alpha^{-1}}$
and $2\geq \frac{\lceil d\rceil}{d}$ (because $d\geq 1$).
That  $K\geq \lceil k\rceil$ is shown similarly.

Let $H'$ be an edge $K$-rewiring of $H$, and let $H''$ be a
vertex $D$-rewiring
of $H'$.
Then $H$, $H'$, $H''$ all have average uniformity $k$, average
vertex degree $d$ and
sparsity $r=\frac{d}{k}$.
Moreover, we have $\deg_H(v)=\deg_{H'}(v)$ for all $v\in V$
and $|e|_{H'}=|e|_{H''}$ for all $e\in E$. The former means that
$H'$ and $H''$
are regular when $H$ is, while latter implies that $H''$ has
maximum uniformity at most
$K=k n^\beta$; it also has maximum vertex degree at most $D=d
n^\alpha$ by construction.
Moreover, by \cref{LM:vertex-flatification}, if $H$ has minimal
vetrex degree
at least $3$, then the same holds for $H''$.
Therefore, assumption \cref{eq:extra_assumptions-pre-lower-bound}
of \cref{TH:pre-lower-bound-i} and at least of the assumptions
(1), (2) of that
theorem hold for  $H''$.
Conditions (1) and (2) of \cref{TH:pre-lower-bound-i} also imply
that $d:=rk\geq 1$.

Let $\tilde{\delta}=\delta-n^{-\alpha}$. Since $n\geq
(\frac{2}{\delta})^{\alpha^{-1}}$,
we have $\tilde{\delta}\geq \frac{\delta}{2}$, and thus
$n\geq N_{\ref{TH:lower-bound-strong-detailed}}\geq
N_{\ref{TH:pre-lower-bound-i}}(k,r,\tilde{\delta},\alpha,\beta,\eta)$.
Now, by applying \cref{TH:pre-lower-bound-i} to $H''$, there is
$\tilde{A}\subseteq V$ of density at most $\tilde{\delta}$
such that
\begin{align*}
    \frac{|E_{H''}(\tilde{A})|}{|E|}
    & \geq
    f_{k,r}(\tilde\delta)-
    A_{\ref{TH:pre-lower-bound-i}}(k,r) n^{-\eta}-
    B_{\ref{TH:pre-lower-bound-i}}(k,r,\tilde{\delta}) n^{-\mu}.
\end{align*}
By \cref{LM:delta-change} and the definition of
$A_{\ref{TH:lower-bound-strong-detailed}}$ and
$B_{\ref{TH:lower-bound-strong-detailed}}$,
this implies that
\begin{align}\label{EQ:TH:lower-bound-strong-detailed:EA-tilde-bound}
    \frac{|E_{H''}(\tilde{A})|}{|E|} & \geq f_{k,r}(\delta)-
    kn^{-\alpha}-
    A_{\ref{TH:lower-bound-strong-detailed}}(k,r) n^{-\eta}-
    B_{\ref{TH:lower-bound-strong-detailed}}(k,r,\delta) n^{-\mu}.
\end{align}

Let $V_0=\{v\in V\,:\,\deg_{H'}(v)>D\}$ and
$E_0=\{e\in E\,:\, |e|_H>K\}$, and put $A=\tilde{A}\cup V_0$.
By \cref{LM:vertex-flatification} and \cref{LM:edge-flatification},
$V_0$ has density at most $\frac{d}{D}=n^{-\alpha}$
in $V$ and $E_0$ has density at most $\frac{k}{K}=n^{-\beta}$ in $E$.
As a result,
\[
    |A|\leq |\tilde{A}|+|V_0|\leq \tilde{\delta} n+
    n^{-\alpha} n=\delta n.
\]
On the other hand, by the said lemmas and
\cref{EQ:TH:lower-bound-strong-detailed:EA-tilde-bound},
\begin{align*}
    \frac{|E_H(A)|}{|E|}
    &\geq \frac{|E_{H'}(A)-E_0|}{|E|}
    \geq \frac{|E_{H'}(\tilde{A}\cup V_0)|}{|E|}-\frac{|E_0|}{|E|}
    \geq \frac{|E_{H''}(\tilde{A} )|}{|E|}-n^{-\beta}
    \\
    &\geq
    f_{k,r}(\delta)-
    kn^{-\alpha}-n^\beta-
    A_{\ref{TH:lower-bound-strong-detailed}}(k,r) n^{-\eta}-
    B_{\ref{TH:lower-bound-strong-detailed}}(k,r,\delta) n^{-\mu}.
\end{align*}
This completes the proof.
\end{proof}

\begin{remark}
One can similarly prove
analogues of \cref{TH:lower-bound-strong-detailed}
holding under the assumption that the maximum vertex degree of
$H$ is at most $dn^\alpha$,
or under the assumption that the maximum uniformity of $H$ is at
most $k n^\beta$.
We omit the details.
\end{remark}

\begin{remark}\label{RM:lower-bound-smallest-n-exp}
In \cref{TH:lower-bound-strong-detailed}, if one fixes $r,k,\delta$
and let $n$ grow, then the   asymptotically best lower bound
on $\frac{|E(A)|}{|A|}$  is obtained
for $\alpha=\beta=\eta=\frac{1}{6}$. In this case, $\mu=\frac{1}{6}$,
and (after some simplifications)
the theorem asserts that if $n\geq  \frac{12^6 k^{12}
r^6}{\delta^6(1-\delta)^6}$, then there exists $A\subseteq V$
of density at most $\delta$ such that
\[
    \frac{|E(A)|}{|E|}
    \geq
    f_{k,r}(\delta)-A_{\ref{RM:lower-bound-smallest-n-exp}}(k,r,\delta)n^{-\frac{1}{6}},
\]
where $A_{\ref{RM:lower-bound-smallest-n-exp}}(k,r,\delta)=
63194 \cdot r^4k^4\delta^{-3}+2rk^3+rk^2+k+1$.
However, this choice of parameters may not be optimal if one
allows $r$, $k$ or $\delta$ to change with $n$.
\end{remark}

We will derive \cref{TH:lower-bound-strong}(ii)
from
the following corollary of \cref{TH:lower-bound-strong-detailed},
which eliminates the dependency on $r$ in the lower bound
on $\frac{|E(A)|}{|E|}$ at the expense of increasing the smallest value that
$n=|V|$ could take.

\begin{corollary}\label{CR:lower_bound_indep_of_r}
Let $\alpha\in [0,\frac{1}{2})$, $\beta\in[0,1) ,\eta\in
(0,\frac{1}{2}),\sigma\in (0,1)$
be real numbers such that
\[
    \mu :=(1-\alpha-\beta-\eta)(1-2\eta)-\alpha>0,
    \qquad
    \sigma<\min\{\eta,\mu\}
    \qquad
    \tand
    \qquad \sigma\leq \min\{\alpha,\beta\}.
\]
Then there are constants $a,b,u,v,C>0$ (depending on
$\alpha,\beta,\eta,\sigma$)
such that the following hold:
Let $H=(V,E)$ be a hypergraph with $n$ vertices, $rn$ edges and
average uniformity
$k>0$. Suppose further that one of the conditions
(1), (2) of \cref{TH:lower-bound-strong-detailed} holds for $H$.
Let $\delta\in (0,1)$.
If $n\geq C k^a r^b \delta^{-u}(1-\delta)^{-v}$, then there
exists $A\subseteq V$
of density at most $\delta$ such that
\[
    \frac{|E(A)|}{|E|}
    \geq f_{k,r}(\delta) - 63200 k^4\delta^{-3} n^{-\sigma}.
\]
Moreover, we can take
\begin{align*}
    a &= \max\{
        2(1-2\alpha)^{-1},
        2(1-2\eta)^{-1},
        2\eta^{-1},
        2(1-\alpha-\beta)^{-1}
    \}
    \\
    b &= \max\{
        (1-2\alpha)^{-1},
        {(\eta-\sigma)^{-1}},
        4(\mu-\sigma)^{-1}
    \},
    \\
    u &= \max\{
        (1-2\alpha)^{-1},
        \eta^{-1},
        \alpha^{-1}
    \},
    \\
    v &= \max\{
        (1-2\alpha)^{-1},
        \eta^{-1}
    \},
    \\
    C &= \max\left\{
        32^{(1-2\alpha)^{-1}},
        \parens*{\frac{4}{1-2\eta}}^{4(1-2\eta)^{-1}},
        12^{\eta^{-1}},
        2^{\alpha^{-1}},
        2^{\beta^{-1}}
    \right\}.
\end{align*}
\end{corollary}

\begin{proof}
Suppose $n\geq N_{\ref{TH:lower-bound-strong-detailed}}$.
Then by \cref{TH:lower-bound-strong-detailed}, there is
$A\subseteq V$ of density at most $\delta$ such that
\[
    \frac{|E( {A})|}{|E|}
    \geq
    f_{k,r}( \delta)
    -kn^{-\alpha}
    -n^{-\eta}
    -(2k^3+k^2)r n^{-\eta}
    -63194 \cdot k^4r^4 \delta^{-3}   n^{-\mu}.
\]
If it is moreover the case that
$rn^{-\eta}\leq n^{-\sigma}$, equiv.\
$n\geq r^{(\eta-\sigma)^{-1}}$,
and $r^4 n^{-\mu}\geq n^{-\sigma}$,
equiv.\ $n\geq r^{4(\mu-\sigma)^{-1}}$, then this implies that
\[
    \frac{|E( {A})|}{|E|}
    \geq
    f_{k,r}( \delta)
    - kn^{-\sigma}
    -n^{-\sigma}
    -(2k^3+k^2) n^{-\sigma}
    -63194  \cdot k^4 \delta^{-3}   n^{-\sigma}
    \geq f_{k,r}( \delta)- 63199k^4\delta^{-3} n^{-\sigma}.
\]
It is routine to check that $C k^a r^b \delta^{-u}(1-\delta)^{-v}\geq
\max\{N_{\ref{TH:lower-bound-strong-detailed}},r^{(\eta-\sigma)^{-1}},
r^{4(\mu-\sigma)^{-1}}\}$, hence the corollary.
\end{proof}

We finally prove \cref{TH:lower-bound-strong}(ii).

\begin{proof}[Proof of \cref{TH:lower-bound-strong}(ii)]
Let $\sigma\in (0,\frac{1}{6})$.
We have seen in Remark~\ref{RM:lower-bound-smallest-n-exp} that
taking $\alpha=\beta=\eta=\frac{1}{6}$ gives
$\mu=\frac{1}{6}$. Applying \cref{CR:lower_bound_indep_of_r}
with these $\alpha,\beta,\eta$ and $\sigma$ then gives
$a=12$, $b=\frac{24}{1-6\sigma}$, $u=v=6$ and $C=2985984\leq
3\cdot 10^6$,
which proves
\cref{TH:lower-bound-strong}(ii) for a general $\sigma$.
When $\sigma=\frac{7-\sqrt{33}}{16}$, taking
$\alpha=\beta=\sigma$ and $\eta=2\sigma$ gives
the parameters $a,b,u,v,C$ specified in the theorem for that
particular $\sigma$. (Note that   this makes $b$ much
smaller than before, namely, $b=7+\sqrt{33}\leq 12.75$.)
\end{proof}

\begin{remark}
The value of the parameter $b$ in \cref{TH:lower-bound}
and \cref{TH:lower-bound-epsilon-view} (the exponent of $r$ and
$\ee^{-1}$, respectively).
turns out to be $\min\{b,\sigma^{-1}\}$ for the $b$ and
$\sigma$ chosen in
\cref{CR:lower_bound_indep_of_r}. Computer experiments
suggest that choosing $\sigma=\frac{7-\sqrt{33}}{16}$,
$\alpha=\beta=\sigma$ and $\eta=2\sigma$ minimizes
$\min\{b,\sigma^{-1}\}$, and this is why we considered this
choice of parameters
in the proof of
\cref{TH:lower-bound-strong}.
%
%
\end{remark}

\section{Proof of Analytic Results}
\label{sec:analytic}

In this section we complete the proofs of the results of
\cref{sec:lowerbounds} by proving all the analytic claims that were
used in those proofs.
Specifically, \cref{subsec:proof-of-hard-optimization} concerns with
proving \cref{TH:hard-optimization} and in
\cref{subsec:properties-of-f-k-r-delta} we prove the properties of
the function $f_{k,r}(\delta)$, see \cref{EQ:f-k-r-delta-dfn}, that
we used earlier.

\subsection{Proof of \cref{TH:hard-optimization}}
\label{subsec:proof-of-hard-optimization}

We recall the notation introduced in \cref{TH:hard-optimization}:
We are given  $u_1,\dots,u_m\in [0,1]$ are such that $\sum_{i=1}^m u_i=1$,
$k\in [1,\infty)$ and $\delta\in (0,1)$.
Then, for every $x_1,\dots,x_m\in (1,\infty)$, we let
$d=\sum_{i=1}^m x_iu_i$, $r=\frac{d}{k}$,
define   $\gamma$ to be   the unique element
of $(0,1)$ such that
\begin{equation*}
\delta=\sum_{i=1}^m u_i(1-(1-\gamma)^{x_i}),
\end{equation*}
and put
\[
F(x_1,\dots,x_m)=\gamma+(1-\gamma)\prod_{i=1}^m(1-(1-\gamma)^{x_i-1})^{\frac{u_i
x_i}{r}}.
\]
Our goal is to show that if $x_1,\dots,x_m\geq 3$, then
\[
F(x_1,\dots,x_m)\geq F(d,\dots,d)=
f_{k,r})\delta):=1-(1-\delta)^{\frac{1}{rk}}+
(1-\delta)^{\frac{1}{rk}}(1-(1-\delta)^{\frac{rk-1}{rk}})^k,
\]
and equality holds if and only if $x_1=\dots=x_m=d$.

Note first that
$F(d,\dots,d)$ does evaluate to
$1-(1-\delta)^{\frac{1}{rk}}+
(1-\delta)^{\frac{1}{rk}}(1-(1-\delta)^{\frac{rk-1}{rk}})^k$.
Indeed, if $x_i=d$ for all $i$, then $\delta=\sum_i
u_i(1-(1-\gamma)^d)=1-(1-\gamma)^d$
and so $\gamma=1-(1-\delta)^{\frac{1}{d}}$. It follows that
\begin{align*}
F(d,\dots,d)&=\gamma+(1-\gamma)\prod_{i}(1-(1-\gamma)^{d-1})^{\frac{u_id}{r}}
\\
&=1-(1-\delta)^{\frac{1}{d}}+(1-\delta)^{\frac{1}{d}}(1-(1-\delta)^{\frac{d-1}{d}})^{\sum_i
\frac{u_i d}{r}}
\\
&=1-(1-\delta)^{\frac{1}{rk}}+
(1-\delta)^{\frac{1}{rk}}(1-(1-\delta)^{\frac{rk-1}{rk}})^k.
\end{align*}
The hard part of   theorem
is proving the inequality $F(x_1,\dots,x_m)\geq F(d,\dots,d)$.

\begin{lemma}\label{LM:hard-optimization-core}
With notation as in \cref{TH:hard-optimization},
put
\[
f(x_1,\dots,x_m)=\sum_{i=1}^m u_i x_i \ln(1-(1-\gamma)^{x_i-1}).
\]
If $x_1,\dots,x_m\geq 3$, then $f(x_1,\dots,x_m)\geq
\frac{\ln(1-\delta)}{\ln(1-\gamma)}\ln\left(\frac{\delta-\gamma}{1-\gamma}\right)$,
and equality holds if and only if $x_1=\dots=x_m $.
\end{lemma}

The rationale for considering $f(x_1,\dots,x_m)$
is that $F(x_1,\dots,x_m)=\gamma+(1-\gamma)\exp(\frac{1}{r}f(x_1,\dots,x_m))$.

\begin{proof}
For each $i$, substitute
\[c_i=1-(1-\gamma)^{x_i-1}.\]
This is equivalent to
\[
x_i=1+\frac{\ln(1-c_i)}{\ln(1-\gamma)}.
\]
We have $c_i\leq 1$, and
the assumption $x_i\geq 3$ implies   that
\[
1-(1-\gamma)^2\leq c_i.
\]

Observe that
\begin{align*}
1-(1-\gamma)^{x_i} &=
1-(1-\gamma)(1-\gamma)^{x_i-1}
\\
&=\gamma+(1-\gamma)-(1-\gamma)(1-\gamma)^{x_i-1}
\\
&=\gamma + (1-\gamma)(1-(1-\gamma)^{x_i-1})=\gamma+(1-\gamma)c_i.
\end{align*}
As a result, 
\[
\delta =\sum_iu_i(1-(1-\gamma)^{x_i})=
\sum_i u_i(\gamma+(1-\gamma)c_i)=\gamma+(1-\gamma)\sum_i u_i c_i,
\]
hence
\begin{equation}\label{EQ:ci-average}
\sum_i u_i c_i = \frac{\delta-\gamma}{1-\gamma}.
\end{equation}

Next, observe that
\begin{align*}
f(x_1,\dots,x_m)&=\sum_i u_i x_i\ln(1-(1-\gamma)^{x_i-1})
=\sum_i u_i
\left(1+\frac{\ln(1-c_i)}{\ln(1-\gamma)}\right)\ln c_i=:(\star).
\end{align*}
Since $c_1,\dots,c_m\in [1-(1-\gamma)^2,1]=:I$,
if it were the case that the function
$g(x)=\left(1+\frac{\ln(1-x)}{\ln(1-\gamma)}\right)\ln x$
is strictly convex in $I$, then Jensen's inequality and \cref{EQ:ci-average}
would imply that
\begin{align*}
(\star)&
\geq
\left(1+\frac{\ln(1-\frac{\delta-\gamma}{1-\gamma})}{\ln(1-\gamma)}\right)\ln
\left( \frac{\delta-\gamma}{1-\gamma}\right)
\\
&=\frac{\ln(1-\gamma)+\ln(1-\gamma-(\delta-\gamma))-\ln(1-\gamma)}{\ln(1-\gamma)}
\cdot \ln \left( \frac{\delta-\gamma}{1-\gamma}\right)
\\
&=
\frac{\ln(1-\delta)}{\ln(1-\gamma)}\ln \left(
\frac{\delta-\gamma}{1-\gamma}\right),
\end{align*}
with equality holding if and only if $c_1=\dots =c_m$, equiv.\ $x_1=\dots=x_m$,
and the lemma will be proved.

Thus, it is enough to prove that
$g(x)=\left(1+\frac{\ln(1-x)}{\ln(1-\gamma)}\right)\ln x$
is strictly convex in  $ [1-(1-\gamma)^2,1]$. This is the content of the
following \cref{LM:convexity-i}.
\end{proof}

\begin{lemma}\label{LM:convexity-i}
For every $\gamma\in (0,1)$, the function
$g(x)=\left(1+\frac{\ln(1-x)}{\ln(1-\gamma)}\right)\ln x$
is strictly convex in $I:=[1-(1-\gamma)^2,1]$.
\end{lemma}

Unfortunately, $g(x)$ is usually not convex in
$[1-(1-\gamma)^1,1]=[\gamma,1]$ (e.g.,
take $\gamma=0.1$); if it were the case, then \cref{TH:hard-optimization}
would actually hold under the milder assumption $x_1,\dots,x_m\in [2,\infty)$.

\begin{proof}[Proof of \cref{LM:convexity-i}.]
Since $\ln(1-\gamma)\leq 0$, we may replace $g(x)$ with
\[
h(x):=-\ln(1-\gamma) g(x)=-(\ln(1-x)+\ln(1-\gamma))\ln x.
\]
It is therefore enough to show that $h''(x)\geq 0$
for all  $x\in I':=[1-(1-\gamma)^2,1)$.

By direct computation,
\[
h''(x)=\frac{\ln
x}{(1-x)^2}+\frac{2}{x(1-x)}+\frac{\ln(1-x)+\ln(1-\gamma)}{x^2}.
\]
Multiplying by $x^2(1-x)^2$, we are reduced into showing that
\begin{equation*}
x^2\ln x+2x(1-x)+(1-x)^2\ln (1-x)+(1-x)^2\ln(1-\gamma)\geq 0
\qquad\forall x\in I'.
\end{equation*}
Rearranging, this is equivalent to
\[
x^2\ln x+2x(1-x)+(1-x)^2\ln (1-x)\geq -(1-x)^2\ln(1-\gamma)
\qquad\forall x\in I'.
\]
It is therefore enough to prove that for every $x\in I'$, we have
\begin{enumerate}
\item[(1)] $x^2\ln x+2x(1-x)+(1-x)^2\ln (1-x)> 0.5x(1-x)$ and
\item[(2)] $0.5x(1-x)\geq -(1-x)^2\ln(1-\gamma)$.
\end{enumerate}
Suppose henceforth that $x\in I'=[1-(1-\gamma)^2,1)$.

\smallskip

{\it Proof of (1).} This is equivalent to showing
\[
x^2\ln x+1.5x(1-x)+(1-x)^2\ln (1-x)> 0.
\]
By applying \cref{LM:ln-approximation} twice, we get
\begin{align*}
x^2\ln x&+1.5x(1-x)+(1-x)^2\ln (1-x)
\\
&>
x^2\cdot \frac{-(1-x)+0.5(1-x)^2}{x}+1.5x(1-x)+(1-x)^2\cdot
\frac{-x+0.5x^2}{1-x}
\\
&=
x(1-x)[-1+0.5(1-x)]+1.5x(1-x)+x(1-x)[-1+0.5x]=0,
\end{align*}
which is what we want.

\smallskip

{\it Proof of (2).}
Dividing by $1-x$, we are reduced into showing
\[
0.5 x \geq -(1-x)\ln(1-\gamma).
\]
Solving this inequality for $x$ (noting that $\ln(1-\gamma)<0$), this
is equivalent to
\[
x\geq -\frac{\ln(1-\gamma)}{0.5-\ln(1-\gamma)}.
\]
Since $x\geq 1-(1-\gamma)^2$, it is enough to prove
\begin{equation*}
1-(1-\gamma)^2\geq -\frac{\ln(1-\gamma)}{0.5-\ln(1-\gamma)},
\end{equation*}
and by rearranging, this becomes
\[
\frac{-\gamma+0.5\gamma^2}{(1-\gamma)^2}\leq \ln(1-\gamma).
\]
Now, by \cref{LM:ln-approximation} and the assumption $\gamma\in (0,1)$,
\[
\ln(1-\gamma)\geq \frac{-\gamma+0.5\gamma^2}{ 1-\gamma }\geq
\frac{-\gamma+0.5\gamma^2}{(1-\gamma)^2},
\]
so (2) holds and the proof is complete.
\end{proof}

We are now ready to prove \cref{TH:hard-optimization}

\begin{proof}[Proof of \cref{TH:hard-optimization}.]
By \cref{LM:hard-optimization-core},
\begin{align*}
F(x_1,\dots,x_m)&=\gamma+(1-\gamma)\exp\left(\frac{1}{r}f(x_1,\dots,x_m)\right)
\\
&\geq \gamma+(1-\gamma)\exp\left( \frac{
\ln(1-\delta)}{r\ln(1-\gamma)}\ln\left(\frac{\delta-\gamma}{1-\gamma}\right)\right)
\\
&=\gamma+(1-\gamma)(1-\delta)^{\frac{\ln(\delta-\gamma)-\ln(1-\gamma)}{r\ln(1-\gamma)}}
\\
&=\gamma+(1-\gamma)(1-\delta)^{\frac{1}{r}\left(\frac{\ln(\delta-\gamma)}{\ln(1-\gamma)}-1\right)}=:(\star),
\end{align*}
and equality holds   only if $x_1=\dots=x_m$.
Put
\[
f(x)=
x+(1-x)(1-\delta)^{\frac{1}{r}\left(\frac{\ln(\delta-x)}{\ln(1-x)}-1\right)}
\]
and suppose $f(x)$ is increasing in the interval
$I=[1-(1-\delta)^{\frac{1}{d}},
1-(1-\delta)^{\frac{1}{3}}]$. By \cref{LM:gamma-bounds}, $\gamma\in I$,
so
\begin{align*}
(\star)&\geq f(1-(1-\delta)^{\frac{1}{d}}) \\
&=1-(1-\delta)^{\frac{1}{d}}+(1-\delta)^{\frac{1}{d}}(1-\delta)^{
\frac{1}{r}\left(\frac{\ln(\delta-1+(1-\delta)^{1/d})}{\ln((1-\delta)^{1/d})}-1\right)}
\\
&=
1-(1-\delta)^{\frac{1}{d}}+(1-\delta)^{\frac{1}{d}}(1-\delta)^{
\frac{d}{r}\log_{1-\delta}((1-\delta)^{1/d}-(1-\delta))-\frac{1}{r}}
\\
&=
1-(1-\delta)^{\frac{1}{d}}+(1-\delta)^{\frac{1}{d}}
\left((1-\delta)^{1/d}-(1-\delta)\right)^{\frac{d}{r}}(1-\delta)^{-\frac{1}{r}}
\\
&=
1-(1-\delta)^{\frac{1}{d}}+(1-\delta)^{\frac{1}{d}}(1-\delta)^{\frac{1}{d}\cdot\frac{d}{r}}
(1-(1-\delta)^{\frac{d-1}{d}})^{\frac{d}{r}}(1-\delta)^{-\frac{1}{r}}
\\
&=
1-(1-\delta)^{\frac{1}{rk}}+(1-\delta)^{\frac{1}{rk}}
(1-(1-\delta)^{\frac{rk-1}{rk}})^{k},
\end{align*}
and the theorem is proved.

We are therefore reduced into showing that $f(x)$ is increasing
in the interval $I$, which follows from the next  lemma.
\end{proof}

\begin{lemma}\label{LM:increasing-i}
For every $r>0$ and $\delta\in (0,1)$, the function
$f(x)=
x+(1-x)(1-\delta)^{\frac{1}{r}\left(\frac{\ln(\delta-x)}{\ln(1-x)}-1\right)}$
is increasing in the interval
$(0,1-\sqrt{1-\delta}]$.
\end{lemma}

\begin{proof}
It enough to show that $f'(x)\geq 0$ for $x\in (0,1-\sqrt{1-\delta}]=:I$.
Suppose that $x\in I$ throughout.

A routine computation shows that
\[
f'(x)=1+(1-\delta)^{\frac{1}{r}\left(\frac{\ln(\delta-x)}{\ln(1-x)}-1\right)}
\left(-1 +
\frac{(1-x)\ln(1-\delta)}{r(\ln(1-x))^2}
\left(
    \frac{\ln(\delta-x)}{1-x} - \frac{\ln(1-x)}{\delta-x}
\right)
\right).
\]
Thus, $f'(x)\geq 0$ if and only if
\begin{align*}
& -1 \leq
(1-\delta)^{\frac{1}{r}\left(\frac{\ln(\delta-x)}{\ln(1-x)}-1\right)}
\left(-1 +
\frac{(1-x)\ln(1-\delta)}{r(\ln(1-x))^2}
\left(
    \frac{\ln(\delta-x)}{1-x} - \frac{\ln(1-x)}{\delta-x}
\right)
\right) & \iff
\\
&
-(1-\delta)^{-\frac{1}{r}\left(\frac{\ln(\delta-x)}{\ln(1-x)}-1\right)}
\leq
-1 +
\frac{(1-x)\ln(1-\delta)}{r(\ln(1-x))^2}
\left(
\frac{\ln(\delta-x)}{1-x} - \frac{\ln(1-x)}{\delta-x}
\right)
& \iff
\\
&
1-(1-\delta)^{-\frac{1}{r}\left(\frac{\ln(\delta-x)}{\ln(1-x)}-1\right)}
\leq
\frac{(1-x)\ln(1-\delta)}{r(\ln(1-x))^2}
\left(
\frac{\ln(\delta-x)}{1-x} - \frac{\ln(1-x)}{\delta-x}
\right).
\end{align*}
Now, by \cref{LM:power-ineq}, we have
\begin{align*}
1-(1-\delta)^{-\frac{1}{r}\left(\frac{\ln(\delta-x)}{\ln(1-x)}-1\right)}
&\leq \frac{1}{r}\left(\frac{\ln(\delta-x)}{\ln(1-x)}-1\right)\cdot
\ln(1-\delta)
=\frac{\ln(1-\delta)}{r}\left(\frac{\ln(\delta-x)}{\ln(1-x)}-1\right).
\end{align*}
%
Thus, it is enough to prove
\[
\frac{\ln(1-\delta)}{r}\left(\frac{\ln(\delta-x)}{\ln(1-x)}-1\right)
\leq
\frac{(1-x)\ln(1-\delta)}{r(\ln(1-x))^2}
\left(
\frac{\ln(\delta-x)}{1-x} - \frac{\ln(1-x)}{\delta-x}
\right).
\]
Multiplying both sides by $\frac{r\ln(1-x)}{\ln(1-\delta)}$,
this is becomes
\[
\ln(\delta-x)-\ln(1-x)\leq \frac{(1-x)}{\ln(1-x)}\left(
\frac{\ln(\delta-x)}{1-x} - \frac{\ln(1-x)}{\delta-x}
\right)
=\frac{\ln(\delta-x)}{\ln(1-x)}-\frac{1-x}{\delta-x},
\]
and after rearranging, we reduce into showing
\begin{equation}\label{EQ:LM:increasing-i:reduction-i}
-\ln\left(\frac{1-x}{\delta-x}\right)+\frac{1-x}{\delta-x}\leq
\frac{\ln(\delta-x)}{\ln(1-x)}.
\end{equation}

By elementary analysis, the function $t\mapsto \frac{1-t}{\delta-t}$
is increasing on $[0,\delta)\supseteq I$ and the function
$y\mapsto -\ln y+y$ is increasing on $[1,\infty)$.
This means in particular that $\frac{1-t}{\delta-t}\in
[f(0),\infty)=[\delta^{-1},\infty)\subseteq [1,\infty)$, and thus the
function $g(t)= -\ln\left(\frac{1-t}{\delta-t}\right)+\frac{1-t}{\delta-t}$
is increasing on $[0,\delta)$ and in particular on
$I=(0,1-\sqrt{1-\delta}]$.
It follows that
\[
1=-\ln 1+1\leq g(x)\leq g(1-\sqrt{1-\delta}).
\]
Noting that
$
\frac{1-(1-\sqrt{1-\delta})}{\delta-(1-\sqrt{1-\delta})}
=\frac{\sqrt{1-\delta}}{\sqrt{1-\delta}-(1-\delta)}=
\frac{1}{1-\sqrt{1-\delta}}
$, we conclude that
\begin{equation}\label{EQ:LM:increasing-i:h-dfn}
1\leq
g(x)=-\ln\left(\frac{1-x}{\delta-x}\right)+\frac{1-x}{\delta-x}\leq
\ln(1-\sqrt{1-\delta})+\frac{1}{1-\sqrt{1-\delta}}=:h(\delta).
\end{equation}
Thus, in order to prove \cref{EQ:LM:increasing-i:reduction-i}, it is enough
to prove that
\[
h(\delta)\leq \frac{\ln(\delta-x)}{\ln(1-x)},
\]
which can be rearranged into
\begin{equation}\label{EQ:LM:increasing-i:reduction-ii}
(1-x)^{h(\delta)}\geq \delta- x.
\end{equation}

Since $h(\delta)\geq 1$ (see \cref{EQ:LM:increasing-i:h-dfn}), we have
$(1-x)^{h(\delta)}\geq 1-h(\delta)x$, so
\cref{EQ:LM:increasing-i:reduction-ii}
would follow if we establish that
$1-h(\delta)x \geq \delta-x$.
This holds if and only if
\[
x\leq \frac{1-\delta}{h(\delta)-1}.
\]
Since $x\in (0,1-\sqrt{1-\delta}]$, we are therefore reduced into proving
\begin{equation*}
1-\sqrt{1-\delta}\leq \frac{1-\delta}{h(\delta)-1}
\qquad\forall \delta\in (0,1).
\end{equation*}
As $h(\delta)-1\geq0$  (see \cref{EQ:LM:increasing-i:h-dfn}),
this is equivalent to
\begin{align*}
1-\delta
&\geq (1-\sqrt{1-\delta})(h(\delta)-1)
\\
&=(1-\sqrt{1-\delta})\left(\ln(1-\sqrt{1-\delta})+\frac{1}{1-\sqrt{1-\delta}}-1\right)
\\
&=(1-\sqrt{1-\delta})\ln(1-\sqrt{1-\delta})+\sqrt{1-\delta}.
\end{align*}
which rearranges into
\begin{equation}\label{EQ:LM:increasing-i:reduction-iii}
(1-\delta)-\sqrt{1-\delta}\geq
(1-\sqrt{1-\delta})\ln(1-\sqrt{1-\delta}).
\end{equation}
Put $y=1-\sqrt{1-\delta}$ and note that $y>0$. Then
\cref{EQ:LM:increasing-i:reduction-iii}
is equivalent to
\[
y^2-y=(1-y)^2-(1-y)\geq y\ln y.
\]
Diving by $y$, we further reduce into showing
that $y-1\geq \ln y$, which holds by elementary analysis.
We have therefore proved \cref{EQ:LM:increasing-i:reduction-iii}  and
thus completed the proof of
the lemma.
\end{proof}

\subsection{Properties of the functions $f_{k,r}(x)$}
\label{subsec:properties-of-f-k-r-delta}

Recall from \cref{EQ:f-k-r-delta-dfn} that for every $k>0$, $r\geq
\frac{1}{k}$ we defined
\[
f_{k,r}(x)=
1-(1-x)^{\frac{1}{rk}}
+(1-x)^{\frac{1}{rk}}
(1-(1-x)^{\frac{rk-1}{rk}})^k.
\]
We now establish some properties of these functions in the domain $[0,1]$.

\begin{proposition}\label{PR:derivative}
For every $k>0$ and $r\geq \frac{1}{k}$, we
have   $\frac{1}{rk}< f'_{k,r}(x)\leq  k$ for all $x\in (0,1)$. In particular,
$f_{k,r}(x)$ is strictly increasing on $[0,1]$.
\end{proposition}

\begin{proof}
Write
\[g(x)=-f_{k,r}(1-x)=-1+x^{\frac{1}{rk}}
-x^{\frac{1}{rk}}
(1-x^{1-\frac{1}{rk}})^k.\]
Then $g'(x)=f'_{k,r}(1-x)$, and so it is enough to show that
$0< g'(x)\leq  k$ for all $x\in (0,1)$.

Suppose $x\in (0,1)$.
We have
\begin{align*}
g'(x)&=\frac{1}{rk}x^{\frac{1}{rk}-1}-
\parens*{\frac{1}{rk}x^{\frac{1}{rk}-1}(1-x^{1-\frac{1}{rk}})^k
+x^{\frac{1}{rk}}k(1-x^{1-\frac{1}{rk}})^{k-1}\parens*{-\parens*{1-\frac{1}{rk}}x^{-\frac{1}{rk}}}}
\\
&=\frac{1}{rk}x^{\frac{1}{rk}-1}\parens*{1-  (1-x^{1-\frac{1}{rk}})^k}
+\parens*{k-\frac{1}{r}}(1-x^{1-\frac{1}{rk}})^{k-1}.
\end{align*}
It is clear that  $(k-\frac{1}{r})(1-x^{1-\frac{1}{rk}})^{k-1}\in
(0,k-\frac{1}{r}]$
(here we need $rk\geq 1$),
so it is enough to show that
\[
h(x):=\frac{1}{rk}x^{\frac{1}{rk}-1}\parens*{1-
(1-x^{1-\frac{1}{rk}})^k}\in \sqbr*{\frac{1}{rk},\frac{1}{r}}.
\]
Indeed, substitute $u=x^{1-\frac{1}{rk}}\in [0,1]$. Then
\begin{align*}
h(x) &=\frac{1}{rk}\cdot\frac{1-(1-u)^k}{u}
=
\frac{1}{rk}\cdot\frac{1-(1-u)^k}{1-(1-u)}=\frac{1}{rk}(1+(1-u)+\dots+(1-u)^{k-1})
\in \sqbr*{\frac{1}{rk},\frac{1}{rk}\cdot
k}=\sqbr*{\frac{1}{rk},\frac{1}{r}},
\end{align*}
and the proof is complete.
\end{proof}

We use \cref{PR:derivative} to prove \cref{LM:delta-change}.

\begin{proof}[Proof of \cref{LM:delta-change}]
This follows immediately from \cref{LM:diff_trick}
and \cref{PR:derivative}.
\end{proof}

Before proving more properties of $f_{k,r}(x)$, we need to establish
two technical lemmas.

\begin{lemma}\label{LM:g-prime-decreasing}
Let $x \in (0,1)$ and $k>1$.
For every $u\in [1-x,1]$, define
\[p(u)=-1+\left(1-\frac{1-x}{u}\right)^{k-1}\left(1+(k-1)\frac{1-x}{u}\right).\]
Then $p(u)\leq p(1)= -1+x^{k-1}(1+(k-1)(1-x))$.
\end{lemma}

\begin{proof}
It is enough to show that $p(u)$ is increasing in the interval
$ [1-x,1]$.
Put $t=1-\frac{1-x}{u}$ and observe that $t\in [0,x]$. Then
\[
p(u)=-1+t^{k-1}(k-(k-1)t)=-1+kt^{k-1}-(k-1)t^k.
\]
Denoting the left hand side by $h(t)$, we have
\[
h'(t)=k(k-1)t^{k-2}-k(k-1)t^{k-1}\geq 0
\]
for $t\in [0,x]$, so $h(t)$ is increasing in $[0,x]$.
Since $u\mapsto 1-\frac{1-x}{u}$ is increasing on $[1-x,1]$,
it follows that $p(u)=h(1-\frac{1-x}{u})$ is also increasing on
that interval.
\end{proof}

\begin{lemma}\label{LM:ineq-for-distance-bound}
Let $x\in[0,1]$ and $k>1$. Then
\[
1-x^{k-1}(1+(k-1)(1-x))\geq
(1-x)^2\min\left\{1,\frac{k(k-1)}{2}\right\}.
\]
\end{lemma}

\begin{proof}
Put
\[
f(x)=\frac{1-x^{k-1}(1+(k-1)(1-x))}{(1-x)^2}=\frac{1-kx^{k-1}+(k-1)x^k}{(1-x)^2}.
\]
It is enough to show that $f(x)\geq \min\{1,\frac{k(k-1)}{2}\}$ for
$x\in [0,1)$.
Suppose that $x\in [0,1)$ henceforth.

By straightforward computation,
\[
f'(x)=\frac{x^{k-2}(-k(k-1)+2k(k-2)x-(k-2)(k-1)x^2)+2}{(1-x)^3}.
\]
Put $g(x)=x^{k-2}(-k(k-1)+2k(k-2)x-(k-2)(k-1)x^2)+2$, so that
$f'(x)=\frac{g(x)}{(1-x)^3}$.
We claim that if $k\leq 2$, then $g(x)\leq 0$, and hence $f'(x)\leq
0$, and if $k>2$, then $g(x)\geq 0$, and hence $f'(x)\geq 0$. Indeed,
\[
g'(x)=x^{k-3}(-k(k-1)(k-2)+2k(k-1)(k-2)x - k(k-1)(k-2)x^2)=
-k(k-1)(k-2)x^{k-3}(1-x)^2.
\]
Since $x\in [0,1)$, this means that $g'(x)\geq 0$ if $k\leq 2$, and
$g'(x)\leq 0$ if $k>2$.
As a result, if $k\leq 2$, we have
\[
g(x)\leq g(1)=-k(k-1)+2k(k-2)-(k-2)(k-1)+2=0,
\]
and if $k>2$, we have
\[
g(x)\geq g(1)=0.
\]
This proves our claim.

Returning to discuss $f(x)$, we have shown that if $k>2$, then $f(x)$
is increasing in $[0,1)$,
and
so
\[
f(x)\geq f(0)=1.
\]
On the other hand, if $k\leq 2$, then $f(x)$ is decreasing in
$[0,1)$, so
\begin{align*}
f(x)
&\geq \lim_{x\to 1} f(x)=
\lim_{x\to 1}
\frac{1-kx^{k-1}+(k-1)x^k}{(1-x)^2}
=
\lim_{x\to 1}
\frac{-(k-1)kx^{k-2}+(k-1)kx^{k-1}}{-2(1-x)}
\\
&=
\lim_{x\to 1}
\frac{-(k-1)k((k-2)x^{k-3}-(k-1)x^{k-2})}{2}=\frac{k(k-1)}{2},
\end{align*}
where we have used L'Hopital's Rule twice in the evaluation
of the limit.
Since in both cases we have $f(x)\geq \min\{1,\frac{k(k-1)}{2}\}$,
we are done.
\end{proof}

We can now establish two more properties of $f_{k,r}(x)$.

\begin{proposition}\label{LM:f-k-r-decrease-in-r}
Let $k>0$ and let $r,r'$ be real numbers such that
$\frac{1}{k}\leq r< r'$.
Then  $f_{k,r}(x)< f_{k,r'}(x)$ for all $x\in (0,1)$.
\end{proposition}


\begin{proof}
For the proof, we fix $x$ and $k$ and think of $r$
as a variable.
Substituting $u=(1-x)^{\frac{1}{rk}}$ gives
\[
f_{k,r}(x)=1-(1-\delta)^{\frac{1}{rk}}+
(1-x)^{\frac{1}{rk}}(1-(1-x)^{\frac{rk-1}{rk}})^k
=1-u+u\left(1-\frac{1-x}{u}\right)^k .
\]
We denote the left hand side by $g(u)$ and think of it as a
function of $u$.
Note that $u\in   [1-x,1]$. Since the function $r\mapsto
(1-x)^{\frac{1}{rk}}$
is increasing, it is enough to prove that $g(u)$ is
decreasing on $[1-x,1]$.
By straightforward computation, one finds that $g'(u)$ is exactly $p(u)$
from \cref{LM:g-prime-decreasing}, so by that lemma, $g'(u)\leq
-1+x^{k-1}(1+(k-1)(1-x))$.
By \cref{LM:ineq-for-distance-bound}, $-1+x^{k-1}(1+(k-1)(1-x))\leq
-(1-x)^2\min\left\{1,\frac{k(k-1)}{2}\right\}<0$,
so we are done.
\end{proof}

\begin{remark}
Plotting the graph of $f_{r,k}(x)$ for various $r,k$ suggests
that $f_{r,k}(x)\geq f_{r,k'}(x)$ whenever $k<k'$. However, we do not
know how to show this.
\end{remark}

%

\begin{theorem}\label{TH:distance-to-delta-to-k}
Let $x\in [0,1]$ and let $k>1$. Then for every $r\geq \frac{1}{k}$,
we have
\[
f_{k,r}(x)
\geq x^k+C_{\ref{TH:distance-to-delta-to-k}}(k) \frac{x(1-x)^2}{r},
\]
where $
C_{\ref{TH:distance-to-delta-to-k}}(k)=\min\{\frac{1}{k},\frac{k-1}{2}\}$.
\end{theorem}




\begin{proof}
The theorem is clear if $x=0,1$, so assume $x \in (0,1)$.
As in the proof of \cref{LM:f-k-r-decrease-in-r},
we substitute   $u=(1-x)^{\frac{1}{rk}}$, so that
$u\in [1-x, 1)$.
Then
\begin{align*}
1-(1-x)^{\frac{1}{rk}}+
(1-x)^{\frac{1}{rk}}(1-(1-x)^{\frac{rk-1}{rk}})^k-x^k
&=1-u+u\left(1-\frac{1-x}{u}\right)^k-x^k.
\end{align*}
We denote the right hand side by $g(u)$ and think of it as a
function of $u$.
In order to prove the theorem, we need to show that
\[
g(u)\geq
C_{\ref{TH:distance-to-delta-to-k}}(k)\frac{x(1-x)}{r}\qquad\forall
u\in [1-x,1).
\]
We assume that $u\in [1-x,1)$ henceforth.

By straightforward computation, $g'(u)$ is the function $p(u)$ from
\cref{LM:g-prime-decreasing}. Thus, by that lemma and
\cref{LM:ineq-for-distance-bound},
\[
g'(u)\leq g'(1)=-1+x^{k-1}(1+(k-1)(1-x))\geq
-\min\left\{1,\frac{k(k-1)}{2}\right\}(1-x)^2=-kC_{\ref{TH:distance-to-delta-to-k}}(k)(1-x)^2,
\]
or rather,
\[
-g'(u)\geq kC_{\ref{TH:distance-to-delta-to-k}}(k)(1-x)^2.
\]

Now, notice that $g(1)=0$. Thus,
\begin{align*}
g(u) &=
(-g(1))-(-g(u))=\int_u^1-g'(t)\,dt
\geq
\int_u^1  kC_{\ref{TH:distance-to-delta-to-k}}(k)(1-x)^2 \,dt
\\
&=
(1-u)kC_{\ref{TH:distance-to-delta-to-k}}(k)(1-x)^2
=
(1-(1-x)^{\frac{1}{rk}})kC_{\ref{TH:distance-to-delta-to-k}}(k)(1-x)^2=
(\star).
\end{align*}
Using \cref{LM:power-ineq} and the assumption $rk\geq 1$, we further get
\[
(\star)\geq
\frac{x}{rk} \cdot
kC_{\ref{TH:distance-to-delta-to-k}}(k)(1-x)^2=C_{\ref{TH:distance-to-delta-to-k}}(k)\frac{x(1-x)^2}{r},
\]
and theorem is proved.
\end{proof}

%

\appendix

\section{Relation  to Dispersers}
\label{apx:dispersers}

In this appendix we translate results about dispersers appearing in
\cite{RadhakrishnanTS00}
to the language of confiners.

A hypergraph
$H$ with $N$ vertices and $M$ edges  is called a
$(K,\epsilon)$-disperser ($0\leq K\leq N$, $0\leq \ee\leq 1$)
if every set of $K$ or more vertices touches more than $(1-\ee)M$ edges
\cite[Dfn.~1.2]{RadhakrishnanTS00}.\footnote{
Dispersers are usually defined as two-sided graphs. This is
equivalent to our
hypergraph point of view by considering the left side of the
two-sided graph
as the vertices and the right side as the hyperedges.
}
Dispersers
are important for the construction of random \emph{extractors}; see
\cite{Shaltiel04} and \cite[Prob.~6.20]{Vadhan12}, for instance.
It is straightforward to see that   $H$
is a $(K,\epsilon)$-disperser if and only  if
it
is an $(\epsilon,1-\frac{K}{N})$-confiner.
Otherwise stated, for a hypergraph with $N$
vertices, being an $(\epsilon,p)$-confiner
is the same as being an $(\lceil N(1-p)\rceil,\epsilon)$-disperser.
However, while the important parameters for studying confiners
are $\epsilon$, the (average) uniformity $k$ and sparsity $r=\frac{M}{N}$,
in the study of dispersers, one usually
wants the the average degree $D$ to be as small as possible
while keeping $M$ as large as possible. (The latter is equivalent to
making $\lg(\frac{DN}{M})$, known as the \emph{entropy loss}
of the disperser,
as small as possible.)
Dispersers are also required to be regular in most texts.



In \cite[Thm.~1.5(a), 1.10(a)]{RadhakrishnanTS00}, the authors
determine essentially
optimal bounds on the parameters $D,N,M$ for which
$(K,\epsilon)$-dispersers exist.
By interpreting them in the language of confiners, we obtain the
following results.


\begin{theorem}[Derived from {\cite[Thm.~1.5(a)]{RadhakrishnanTS00}}]
\label{TH:disp_to_conf_lower_bound}
Let $p\in (0,1)$, $r>0$, $k>1$ and $\epsilon\in (0,\frac{1}{2})$.
If $H$ is an $(\epsilon,p)$-confiner  with $n\geq\max\{ 4k,
\frac{2}{1-p}\}$ vertices and $rn$
edges
having average uniformity $k$, then
\[r\geq \Omega(k^{-1}\epsilon^{-1} \ln((1-p)^{-1})),\]
or equivalently, $\epsilon\geq \Omega(k^{-1}r^{-1}\ln((1-p)^{-1}))$.
\end{theorem}

Recall, however, that  we must have $\epsilon\geq p^k-\frac{2k}{n}$
by   \cref{TH:worst_error_confine}, so  $\epsilon$ cannot approach $0$.

\begin{proof}
Put $m=rn$ and $K=\lceil (1-p)n \rceil$. As noted earlier, $H$ is a
$(K,\epsilon)$-disperser
of average degree $d=rk$ (\cref{PR:degree-to-unif-ratio}).
Our assumptions that $n\geq 4k$ and $\epsilon\leq
\frac{1}{2}$ imply that
$d=rk\leq \frac{rn}{4}=\frac{(1-\epsilon) m}{2}$.
Thus, we may apply \cite[Thm.~1.5(a)]{RadhakrishnanTS00} to get
that
$rk=d\geq \Omega(\epsilon^{-1}
\ln(\frac{n}{K}))=\Omega(\epsilon^{-1}\ln((1-p)^{-1}))$
(we used $n\geq \frac{2}{1-p}$ in the last equality).
Dividing both sides by $k$ gives the theorem.
\end{proof}


\begin{theorem}[Derived from
{\cite[Thm.~1.10(a)]{RadhakrishnanTS00}}]\label{TH:disp_to_conf_random_const}
Fix some $0<c<1$, let $n$ be a natural number and let $p\in (0,1)$
and $k>1$ be real. Then:
\begin{enumerate}[(i)]
\item For every   $\epsilon\geq e^{-(1-p)k+1}$,
there exists
an
\hyperref[eq:eps_confiner]{$(\epsilon,p)$-confiner}
having $n$
vertices, average uniformity $\leq k$
and sparsity
\[r\leq
\frac{\epsilon^{-1}(\ln((1-p)^{-1})+1)+1}{k-(1-p)^{-1}(\ln(\epsilon^{-1})+1)}+\frac{1}{n}.\]
\item For every   $\frac{1}{(1-c)k}+\frac{1}{n}< r\leq
\frac{1}{(1-c)k}[e^{c(1-p)k-1}\ln((1-p)^{-1})+1)+1]$,
there is an \hyperref[eq:eps_confiner]{$(\epsilon,p)$-confiner}
on $n$ vertices with average uniformity $\leq k$
and sparsity $\leq r$, where
\[\epsilon\leq
\frac{\ln((1-p)^{-1})+1}{(1-c)k(r-\frac{1}{n})-1}.\]
\end{enumerate}
\end{theorem}

Note that in part (i), if we assume that $\varepsilon\geq
\exp(-c[p(k+1)+1])$, then we
get the more convenient bound
\[r\leq \frac{\varepsilon^{-1}(\ln((1-p)^{-1})+1)+1}{(1-c)k}+\frac{1}{n}.\]
When applying  (ii), if possible, it is best to choose the smallest
possible $c>0$
for which 
the constraints on $r$ hold.
We further note that (ii) is derived from (i) and therefore cannot
be used it to decrease
$\epsilon$ below $e^{-(1-p)k+1}$.

\begin{proof}
Let $m\in\NN$ and put $K= \lceil (1-p)n\rceil$ and $r=\frac{m}{n}$.
We recall that according to \cite[Thm.~1.10(a)]{RadhakrishnanTS00},
for every $\epsilon>0$, there exists a regular
$(K,\epsilon)$-disperser having $n$, vertices $m$ hyperedges
and degree
\begin{equation}\label{EQ:d_m_ineq}
d=\left\lceil\epsilon^{-1}(\ln(n/K)+1)
+
\frac{m}{K}(\ln(\epsilon^{-1})+1)\right\rceil
\leq
\epsilon^{-1}(\ln((1-p)^{-1})+1)
+
\frac{r}{1-p}(\ln(\epsilon^{-1})+1)+1
\end{equation}
This disperser is also an
\hyperref[eq:eps_confiner]{$(\epsilon,p)$-confiner} of
average uniformity
$ \frac{d}{r}$.
%

Now, in order to prove (i), all we need to do is to fix $\epsilon$
and apply the above conclusion with
the smallest possible $m=r n$ which would guarantee that the average
uniformity $\frac{d}{r }$
does not exceed $k$. To that end, it is enough to make sure that
\begin{equation}\label{EQ:rk_m_ineq}
r k\leq   \epsilon^{-1}(\ln((1-p)^{-1})+1)
+
\frac{r}{1-p}(\ln(\epsilon^{-1})+1)+1.
\end{equation}
Solving for $r$,
this is equivalent to
\[
r \geq
\frac{\epsilon^{-1}(\ln((1-p)^{-1})+1)+1}{k-(1-p)^{-1}(\ln(\epsilon^{-1})+1)}=:r_0.
\]
Since $m$ must be an integer, we cannot take $m=r_0n$
and have to satisfy with taking $m=\lceil r_0 n\rceil$.
This increases the resulting sparsity $r=\frac{\lceil r_0
n\rceil}{n}$ by at most $\frac{1}{n}$,
and (i) follows.
(We remark that if we do not assume that $\epsilon\geq e^{-(1-p)k+1}$,
then there is no $r>0$ for which \cref{EQ:rk_m_ineq} is satisfied.)

Next, we derive (ii) from (i).
Let $\epsilon=\frac{\ln((1-p)^{-1})+1}{(1-c)k(r-\frac{1}{n})-1}$
(note that this makes
sense by our assumption
$r> \frac{1}{(1-c)k}+\frac{1}{n}$).
The upper bound which we assume on $r$
guarantees that $\epsilon\leq  e^{-c(1-p)k+1}$,
so by (i),
it follows that there exists an
\hyperref[eq:eps_confiner]{$(\epsilon,p)$-confiner} on $n$ vertices
having sparsity at most
\[
\frac{\epsilon^{-1}(\ln((1-p)^{-1})+1)+1}{k-(1-p)^{-1}(\ln(\epsilon^{-1})+1)}+\frac{1}{n}
\leq \frac{(1-c)k(r-\frac{1}{n})}{k-c k}\leq r,
\]
where in the first inequality, we  used $\epsilon\leq  e^{-c(1-p)k+1}$.
This is exactly what we want.
(We   remark that the $\epsilon$ we got is not far from the
optimum, because
\cref{EQ:rk_m_ineq} is clearly false for $\epsilon\leq\frac{1}{rk-1}$.)
%
\end{proof}

\begin{remark}\label{RM:limit_of_rand_disp}
Unfortunately, we cannot use  \cite[Thm.~1.10(a)]{RadhakrishnanTS00}
to deduce the existence of sparse
\hyperref[eq:eps_confiner]{$(\epsilon,p)$-confiners} with
average uniformity
$k$ and   $\epsilon\leq e^{-(1-p)k}$.
Indeed, the proof  of \cref{TH:disp_to_conf_random_const}
breaks in this case because the inequality
\[
rk=d\leq \epsilon^{-1}(\ln((1-p)^{-1})+1)
+
\frac{r}{1-p}(\ln(\epsilon^{-1})+1),
\]
which   implies \cref{EQ:d_m_ineq} if it holds,
has no   positive solution for $r$
if  $\epsilon < e^{-(1-p)k+1}$.
\end{remark}
\bibliography{vedat}
\bibliographystyle{alpha}

\end{document}